\documentclass[conference]{IEEEtran}
\IEEEoverridecommandlockouts
\pdfoutput=1
\usepackage{cite}
\usepackage{amsmath}
\usepackage{amssymb}
\usepackage{amsfonts}
\usepackage{diagbox}
\usepackage{calc}
\usepackage{bbding}
\usepackage{balance}

\usepackage{graphicx}
\usepackage{textcomp}
\usepackage{xcolor}
\usepackage{colortbl}
\usepackage{caption}
\usepackage{fancyhdr}
\usepackage{url}
\usepackage{subcaption}
\usepackage{makecell}
\usepackage{amsthm}
\usepackage{algpseudocode}
\usepackage{dsfont}
\def\BibTeX{{\rm B\kern-.05em{\sc i\kern-.025em b}\kern-.08em
    T\kern-.1667em\lower.7ex\hbox{E}\kern-.125emX}}

\makeatletter
\def\endthebibliography{%
	\def\@noitemerr{\@latex@warning{Empty `thebibliography' environment}}%
	\endlist
}
\makeatother

\definecolor{gray}{rgb}{0.5,0.5,0.5}
\newcommand{\eat}[1]{}

\newcommand{\prob}{TIC\xspace}
\newcommand{\probno}{TONIC\xspace}
\newcommand{\fx}{$f(\cdot)$\xspace}
\newtheorem{definition}{Definition}
\newtheorem{remark}{Remark}

\newtheorem{lemma}{Lemma}
\newtheorem{example}{Example}
\newtheorem{theorem}{Theorem}
\newtheorem{corollary}{Corollary}
\newtheorem{problem}{Problem}

\makeatletter
\newif\if@restonecol
\makeatother

\usepackage[linesnumbered,ruled,vlined]{algorithm2e}
\SetKwRepeat{Do}{do}{while}

\DeclareMathOperator*{\avg}{avg}

\renewcommand{\baselinestretch}{0.97}

\newcommand{\rev}[1]{\textcolor[rgb]{0,0,0}{#1}}

\def\BibTeX{{\rm B\kern-.05em{\sc i\kern-.025em b}\kern-.08em
    T\kern-.1667em\lower.7ex\hbox{E}\kern-.125emX}}
    
\begin{document}

\title{Finding Top-r Influential Communities under Aggregation Functions\\
\thanks{\textsuperscript{*}Equal contribution}}
\makeatletter
\newcommand{\linebreakand}{%
  \end{@IEEEauthorhalign}
  \hfill\mbox{}\par
  \mbox{}\hfill\begin{@IEEEauthorhalign}
}
\makeatother
\author{{You Peng\textsuperscript{*}, Song Bian\textsuperscript{*}, Rui Li, Sibo Wang, Jeffrey Xu Yu}

\vspace{1.6mm}\\
\fontsize{10}{10}
\selectfont\itshape The Chinese University of Hong Kong\\
\fontsize{9}{9} \selectfont\ttfamily\upshape
\{ypeng,sbian,lirui,swang,yu\}@se.cuhk.edu.hk
}

\maketitle


\begin{abstract}
Community search is a problem that seeks cohesive and connected subgraphs in a graph that satisfy certain topology constraints, e.g., degree constraints. The majority of existing works focus exclusively on the topology and ignore the nodes' influence in the communities. To tackle this deficiency, influential community search is further proposed to include the node's influence. Each node has a weight, namely influence value, in the influential community search problem to represent its network influence. The influence value of a community is produced by an aggregated function, e.g., $max$, $min$, $avg$, and $sum$, over the influence values of the nodes in the same community. The objective of the influential community search problem is to locate the top-$r$ communities with the highest influence values while satisfying the topology constraints. Existing studies on influential community search have several limitations: {\em (i)} they focus exclusively on simple aggregation functions such as $min$, which may fall short of certain requirements in many real-world scenarios, and {\em (ii)} they impose no limitation on the size of the community, whereas most real-world scenarios do. This motivates us to conduct a new study to fill this gap.

We consider the problem of identifying the top-$r$ influential communities with/without size constraints while using more complicated aggregation functions such as $sum$ or $avg$. We give a theoretical analysis demonstrating the hardness of the problems and propose efficient and effective heuristic solutions for our top-$r$ influential community search problems. Extensive experiments on real large graphs demonstrate that our proposed solution is significantly more efficient than baseline solutions.

\end{abstract}

\section{Introduction}

In reality, graph data becomes increasingly complicated and diverse. The vertex of the graph is filled with relevant information~\cite{fang2019survey,fang2016vldb,fang2017vldb,bian2020efficient,hao2021distributed,yang2021huge,peng2019towards,peng2020answering,peng2021answering,peng2021dlq}. The information could be obtained from either raw data, e.g., H-index and income, or from the topological structure of the graph, e.g., PageRank, Closeness, Degree, and Betweenness. For instance, Twitter could be abstracted into a graph, with each vertex representing a user, and each edge representing whether two individuals follow each other. Then, the information of each vertex could be represented by their influence values. 
Meanwhile, numerous networks have a community structure. The community structure has a wide range of applications in a variety of disciplines, including social network mining~\cite{lwdy17,fang2019densest}, biology analysis~\cite{bv06}, and financial markets~\cite{djdl09,qiu2018real}. Thus, extracting community structure from a graph is a fundamental problem in graph mining. 

A substantial body of previous works have concentrated exclusively on discovering cohesive subgraphs from a large graph, ignoring the attribute of each vertex. Considering this, some works~\cite{lqym15, bclz17} investigate a new community model based on the concept of $k$-core~\cite{seid83,wu2015core}, which is utilized to locate top-$r$ $k$-influential communities over massive graphs. Due to the existence of additional cohesive requirements, the new model is extended to include additional cohesiveness metrics, e.g., $k$-truss~\cite{cohen08}. \rev{As mentioned in~\cite{lqym15, bclz17}, the existing influential communities are mainly based on the $k$-core model.}

Although existing models and approaches are practical and effective, the influence value of community is determined by the minimum value of vertices in their models. This assumption limits their applications. Therefore, the existing model is not always capable of satisfying the users' requirements. Thus, we aim to investigate the top-$r$ $k$-influential community, whose influence value is determined by real-world applications. Based on this intuition, we study a generic influential community model in this paper. The influential community should be constrained by the following criteria: (1) it is a connected subgraph; (2) each vertex of the subgraph has at least $k$ neighbors; and (3) there does not exist a supergraph, such that the influence value of the supergraph is the same as that of it. Additionally, the supergraph satisfies (1) and (2). The influence value of the community, on the other hand, should be determined by various aggregation functions, e.g., $avg$, $sum$, over the entire community, rather than by the vertex's minimum influence value. Additionally, the community search results should be non-overlapping.

\noindent \textbf{Applications.}~Some real-life scenarios are listed as follows to demonstrate our motivations.


\noindent \textit{\underline{(1) Engagement.}} It is common for team members' engagement~\cite{zhang2016engagement,chitnis2013preventing,malliaros2013stay,wu2013arrival} to be determined by the number of friends in the same group. Also, the ability of each member is different. When the team encountered a financial crisis, it was forced to lay off several members. The leader wishes to reduce the size of the squad while maintaining its strength. Then, we could abstract the relationship in the team as a graph and assign each node an influence value as their ability. By identifying the top-$r$ $k$-influential community, we could determine who should be laid off.

\rev{In this application, $max$ could retain the most critical members. By using $weight$ $density$, a highly connected and influential community could be reserved. $Balanced$ $density$ is a variant of $density$, which requires the lay-off members are also highly connected.}

\noindent \textit{\underline{(2) Group Recommendation.}} In a social network, a user may choose to search groups with similar interests in social networks~\cite{amer2009group,cao2018attentive,kim2010group}. For instance, a user could search for keywords on Facebook\footnote{https://www.facebook.com/} or Twitter\footnote{https://twitter.com/} to discover several communities with similar interests. We may assign a similar value to each user's influence in such a social network. The user then seeks out a community with the maximum influence value. The influence value of the community is determined by the average of its members.

\noindent \textit{\underline{(3) Influential Research Groups Identification.}} Mining a research community has been studied for more than two decades~\cite{hcqty14, bclz17, lqym15, yc21, hl17,chen2022answering,feng2022towards,yuan2022efficient}. To locate influential research groups in a research network, e.g., DBLP\footnote{https://dblp.uni-trier.de/}, we could use the influence value of each vertex as the H-index and extract a community with a maximum influence value. Nevertheless, it is worth noting that a significant number of freshly graduated students have joined the group as new professors recently.

\rev{In this application, $min$ and $avg$ could be used to discover a group of highly cited researchers. Nevertheless, they are suitable for different citation metrics, e.g., $i$-$10$ index, $G$ index. It could be seen from the case study in Section VI.C that $G$-$index$ is suitable for $avg$, while $i$-$10$ index is appropriate for $min$. As for $sum$, it could discover high-quality research community with more diversity.}

Inspired by the aforementioned scenarios, it is necessary to investigate the top-$r$ $k$-influential community search under various aggregation functions with or without size constraint, which could solve many real-life issues.  

\rev{
    In local community detection, a goodness metric is usually used to measure whether a subgraph forms a community. The existing goodness metrics for local community detection can be categorized into three classes. The first class optimizes the internal denseness of a subgraph, i.e., the set of nodes in a community should be densely connected with each other.
    Such metrics include the classic density definition~\cite{saha2010dense}, edge-surplus~\cite{tsourakakis2013denser}, and minimum degree\cite{sozio2010community}. The second class optimizes both the internal denseness and the external sparseness. That is, the set of nodes in the community are not only densely connected with each other, but also sparsely connected with the nodes that are not in the community. Such metrics include subgraph modularity~\cite{luo2008exploring}, density-isolation~\cite{lang2007finding}, and external conductance~\cite{andersen2006local}. The local modularity measures the sharpness of the community boundary and belongs to the third class~\cite{clauset2005finding}. Using this metric, the set of nodes in the boundary of the community are highly connected to the nodes in the community but sparsely connected to the nodes outside the community.
}


\begin{table}[!t]
	\centering
	\begin{small}
		\renewcommand{\arraystretch}{1.2}
		\caption{Aggregation Functions under $k$-core Model}
		\vspace{-2mm}
		\label{table:metrics}
		\begin{tabular}{|c|c|c|} \hline
			\cellcolor{gray!25}\textbf{Aggregation functions} & \cellcolor{gray!25}\textbf{Formulas} $f(H)$ & \cellcolor{gray!25}\textbf{Hardness} \\ \hline
			Minimum & $\min_{v \in H} w(v)$ & P \\ \hline
			Maximum & $\max_{v \in H} w(v)$ & P \\ \hline
			Sum & $w(H) = \sum_{v \in H} w(v)$ & P \\ \hline
			Sum-surplus & $w(H) + \alpha |H|$ & P \\ \hline
			Average & $w(H) / |H|$ & NP-hard \\ \hline
			Weight Density & $w(H) - \beta |H|$ & NP-hard \\ \hline
			Balanced Density & $\frac{w(H)}{w(H) - w(V \setminus H)}$ & NP-hard \\ \hline
		\end{tabular}
		\vspace{-2mm}
	\end{small}
\end{table}

\noindent \textbf{Challenges and Contributions.}~The purpose of this paper is to investigate the problem of determining the top-$r$ $k$-influential community over a massive graph using various aggregation functions, e.g., $avg$, $sum$. Table~\ref{table:metrics} lists a collection of commonly used aggregation functions\footnote{The NP-hardness of Weight Density and Balanced Density is given in Appendix of our full version https://bit.ly/3Fa6YdW .}. We would primarily discuss the impact of aggregation functions on the top-$r$ $k$-influential community search. However, in this study, we disregard the procedure of computing the weight of each vertex. 

We have demonstrated that individuals occasionally like to select some influential communities with no overlaps or to identify several influential communities with size constraints. The top-$r$ non-overlapping $k$-influential community search problem is investigated with or without size constraint.

By examining the top-$r$ $k$-influential community search problem under various aggregation functions, we claim that the problem could be solved in polynomial time under some different aggregation functions, e.g., $min$, $max$. We develop a global search algorithm. Then, an improved algorithm is proposed for the problem if the aggregation function is $sum$. However, problems under some aggregation functions, e.g., $avg$, are NP-hard. Unless P = NP, they cannot be addressed in polynomial time. Thus, there are no solutions to these problems that are approximated by the constant-factor. 

When the aggregation function is $avg$ or $sum$, problems with size constrained are NP-hard. In light of this, we propose several efficient heuristic algorithms for the NP-hard problems based on local search. The main contributions of our paper are summarized as follows:

%


\begin{itemize}
	\item \textit{Various Aggregation Functions.}~We extend the original influential community model to various aggregation functions. We analyze the hardness of the problem under different aggregation functions and propose efficient approaches for the influential community search problem.
	\item \textit{Size-Constrained Influential Community.}~We advocate a cohesive subgraph model: size-constrained influential community. We analyze the hardness of the new cohesive subgraph model, and propose some efficient heuristic algorithms. Additionally, we extend our approach to the top-$r$ non-overlapping $k$-influential community search problems.
	\item \textit{Efficiency and Effectiveness.}~Extensive experiments on real networks demonstrate the efficiency of our techniques. In addition, a case study on a real dataset demonstrates the effectiveness of our model and algorithms. 
\end{itemize}

\vspace{1mm}
\noindent \textbf{Roadmap}.
The rest of the paper is organized as follows.
Section~\ref{sec:pre} formally defines the problem.
Section~\ref{sec:hardness} provides hardness analysis of the top-$r$ non-overlapping $k$-influential community search problems with or without size constraint.
Solutions to top-$r$ size-unconstrained (constrained) $k$-influential community problem are proposed in Section~\ref{sec:unconstrained} and~\ref{sec:size_constrained}, respectively.
 followed by the empirical study in Section~\ref{sec:exp}.
Section~\ref{sec:related} surveys important related work.
Section~\ref{sec:conclusion} concludes the paper.
\section{preliminaries}\label{sec:pre}

In this section, we will begin by providing some basic background. Following that, we define the problems that will be discussed in this paper.

\begin{table}[!t]
	\centering
	\begin{small}
		\renewcommand{\arraystretch}{1.2}
		\caption{Notation Table}
		\vspace{-2mm}
		\label{table:Notation Table}
		\begin{tabular}{|p{0.68in}|p{2.268in}|} \hline
		    \cellcolor{gray!25}\textbf{Notations} & \cellcolor{gray!25}\textbf{Description}\\ \hline
			$G(V, E, w)$ & a weighted and undirected graph, where $V$ is the set of vertices, $E$ is the set of edges, and $w$ is a weighted function \\ \hline
			$G[H]$ & the subgraph induced by $H$ \\ \hline 
			$n$ (m) & the number of nodes (edges) in $G$ \\ \hline
			$k$ & the degree constraint for subgraph \\ \hline
			$s$ & the size constraint for subgraph \\ \hline
			$f$ & the aggregation function \\ \hline
			$g(\cdot)$ & the objective function of problem \\ \hline
			$N(u,G)$ & the set of neighbors of vertex $u$ in $G$ \\ \hline
			$N(u,H)$ & the set of neighbors of vertex $u$ in $G[H]$ \\ \hline
			$d(u,G)$ & the degree of vertex $u$ in $G$ \\ \hline
			$d(u,H)$ & the degree of vertex $u$ in $G[H]$ \\ \hline
			$\delta(G) / \delta(H)$ & the minimum degree of $G/G[H]$ \\ \hline
			$f(G)/f(H)$ & the influence value of $G/G[H]$ \\ \hline
		\end{tabular}
		\vspace{-2mm}
	\end{small}
\end{table}

\subsection{Problem Definitions}

Let $G(V,E,w)$ be an undirected and weighted graph, where $V$ denotes a set of vertices, $E \subseteq V \times V$ indicates a set of edges, and $w$ is a weighted function that assigns each vertex $u \in V$ with a non-negative weight value. Throughout this work, we refer to $w(u, G)$ as the weight of vertex $u$ in $G$. The weight assigned to each vertex could reflect the centrality of each vertex such as Pagerank, Betweenness, Closeness, or other attributes~\cite{lqym15}. Moreover, $N(u, G) = \{v \in V| (u,v) \in E \}$ denotes the set of neighbors of vertices $u \in V$ in $G$. The degree of $u$ is $d(u, G) = |N(u, G)|$. When the context is clear, we omit graph $G$ in the notation. Table~\ref{table:Notation Table} lists the notations used in this paper.

This paper is mostly concerned with the $k$-core model. Let $H$ be a subset vertices of $V$, which implies that $H \subseteq V$. The induced subgraph, denoted by $G[H] = (V_H, E_H, w)$, is a $k$-core if it conforms to the following definition:

\begin{definition}[$k$-core]
\label{def:kcore}
Given a graph $G = (V,E,w)$, a subgraph $G[H] = (V_H, E_H, w)$ is a $k$-core of $G$, if $G[H]$ satisfies the following constraints: 
\begin{enumerate}
    \item \textbf{Cohesive:}~For any $u \in V_H$, $d(u) \geq k$. 
    \item \textbf{Maximal:}~$H$ is maximal, i.e., for any vertex set $H' \supset H$, $G[H']$ is not a $k$-core.
\end{enumerate}
\end{definition}

\eat{In Algorithm~\ref{alg:k_core}, we initialize the result $H_k$ with $G$ (Line~\ref{k_core:init}). Then, $k$-core could be obtained by removing vertices from the graph if the degree of a vertex is less than $k$ (Lines~\ref{k_core:while} and~\ref{k_core:remove}). 
The time complexity of Algorithm~\ref{alg:k_core} is $O(m)$~\cite{bz03}.

\begin{algorithm}[!htbp]
    $H_{k} \leftarrow G$\;
    \label{k_core:init}
    \While{$\exists$ $u \in H_{k}$ $\land$ $d(u, H_{k}) < k$}{
    \label{k_core:while}
		$H_{k} \leftarrow H_{k} \setminus \{u\}$\;
		    \label{k_core:remove}
	}
	\Return $H_{k}$\;
	\caption{\textsc{$k$-core}($G$, $k$)}
	\label{alg:k_core}
\end{algorithm}

\eat{\begin{algorithm}[!htbp]
	\caption{\textsc{Compute $k$-core}}
	\label{alg:k_core}
	\KwIn{A graph $G=(V,E,w)$, degree constraint $k$}
	\KwOut{The $k$-core of $G$}
	$H^{k} \leftarrow G$\;
	\While{there exists $u \in H^{k}$ and $d(u, H^{k}) < k$}{
		$H^{k} \leftarrow H^{k} \setminus \{u\}$\;
	}
	\Return $H^{k}$\;
\end{algorithm}}
}

In this paper, we aim to identify the influential communities in large networks. The influential community is a cohesive subgraph whose cohesiveness is based on $k$-core. The influential community has an influence value. Before introducing the concept of influential community, we refer to previous works~\cite{lqym15, bclz17} that define the influence value of an induced subgraph. The definition is given below:

\begin{definition}($f(G[H])$).
\label{def:inf_value}
Let $G[H]$ be a subgraph of $G$, and $f$ denotes an aggregation function. The influence value of subgraph $G[H]$ is denoted by $f(G[H])$, or simply $f(H)$ when the context is clear.
\end{definition}

Nevertheless, previous works~\cite{lqym15, bclz17} concentrate exclusively on the influential community, where the influence value of a community is based on the minimum weight of the vertices in it. On the contrary, a user is more likely to find an influential community whose influence value is determined by various aggregation functions to solve the issues mentioned above. As a result, in contrast to previous research, we provide a general definition of influential community here.

\begin{definition}[$k$-Influential Community]
\label{def:inf_com}
Given an undirected and weighted graph $G=(V,E,w)$, a vertex set $H \subseteq V$ and an aggregation function $f$, the induced subgraph $G[H] = (V_H, E_H, w)$ is a $k$-influential community if 
\begin{enumerate}
\item \textbf{Cohesive:}~For any $u \in V_H$, \rev{$d(u, H) \geq k$}. 
\item \textbf{Connected:}~$G[H]$ is a connected subgraph. 
\item \textbf{Maximal:}~There is no other vertex set $H' \supset H$, such that induced subgraph $G[H'] = (V_{H'}, E_{H'}, w)$ satisfies 1) and 2), and $f(H') = f(H)$.
\end{enumerate}
\end{definition}


Additionally, in certain real-life scenarios, we require a community with a limited size. Thus, we define the size-constrained influential community below.

\begin{definition}[Size-Constrained $k$-Influential Community]
Given a weighted and undirected graph $G=(V,E,w)$, the degree constraint $k$ and the size constraint $s$. A size-constrained $k$-influential community \rev{$G[H]=(V_{H}, E_{H}, w)$} is a $k$-influential community with $|V_{H}| \leq s$.
\end{definition}

We principally focus on the following two influential community search problems in this paper:



\begin{problem}[Top-$r$ size-constrained $k$-Influential Community]
\label{prb:overlap}
Given a weighted and undirected graph $G=(V,E,w)$, the degree constraint $k$, an integer $r$, the size constraint $s$ and an aggregation function $f$, the problem is to find \underline{T}op-$r$ size-constrained $k$-\underline{I}nfluential \underline{C}ommunity (\prob) with the highest influence value under the aggregation function $f$.
\end{problem}

If the size constraint is not emphasized, where we set the size constraint $s=|V|$, the problem would be size-unconstrainted. The following example demonstrates how aggregation functions impact the top-$r$ $k$-influential community search problem, as well as the difference between size-constrained problem and size-unconstrained problem.

\begin{figure}[!htbp]
	\vspace{-2mm}
	\centering
	\includegraphics[height=45mm]{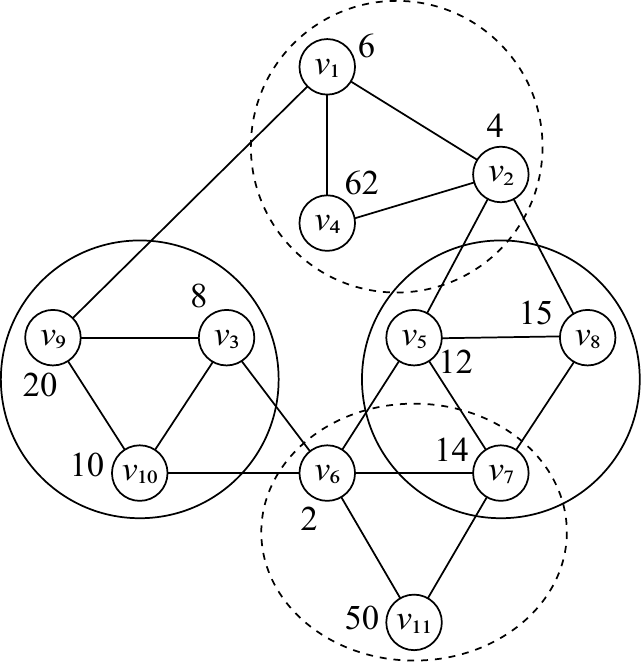}
	\caption{An example network}
	\label{fig:example}
	\vspace{-2mm}
\end{figure}

\begin{example}
As shown in Figure~\ref{fig:example}, if the aggregation function is sum and $k=2$, the top-$2$ $k$-influential community are $\{v_1, v_2, \dots, v_{11}\}$ and $\{v_1, v_2, v_4, \dots, v_{11}\}$. However, when the aggregation function is avg and $k=2$, the top-$2$ $k$-influential community are $\{v_1, v_2, v_4\}$ and $\{v_6, v_7, v_{11}\}$.
If we change the aggregation function to min but maintain $k = 2$, then the top-$2$ $k$-influential community become $\{v_5, v_7, v_8\}$ and $\{v_3, v_9, v_{10}\}$.

The above illustrates the difference under different aggregation functions. Following that, we impose size constraint on the subgraph. We set $f$ as sum, $k=2$, and $s=4$, then $\{v_3, v_6, v_9, v_{10}\}$ is a size-constrained $k$-influential community with influence value $40$. Although another community, $\{v_1, v_2, \dots, v_{11}\}$, has a higher influence value $203$, it is not retrieved due to the community's size being larger than $4$.
\end{example}

To avoid redundancy in results to \prob \ problem, some works~\cite{lqym15, bclz17} study the top-$r$ non-contained $k$-influential community without size constraint when the aggregation function is $min$. The non-containment constraint, on the other hand, does not work well if the aggregation function is not $min$.

In Figure~\ref{fig:example} for instance, we assume that $k=2$ and the aggregation function \fx $=$ $avg$. We could obtain that $\{v_6, v_7, v_{11}\}$, $\{v_5, v_6, v_7\}$, and $\{v_5, v_7, v_8\}$ are all $k$-influential community. The problem is that these communities have overlaps with each other, which is not permitted in certain real-world scenarios. We propose the definition of non-overlapping community search based on this. The definition is given below:

\begin{definition}[Non-overlapping]
\label{def:non_overlap}
Given a weighted and undirected graph $G=(V,E,w)$, the degree constraint $k$, and integer $r$ and an aggregation function $f$. We suppose that the result of top-$r$ $k$-influential community search problem is $\{ H_1$, $H_2$, $\dots$, $H_{r} \}$. For any two communities $H_{i}$ and $H_{j}$, if $H_{i} \cap H_{j} = \emptyset$, we refer the result to non-overlapping.
\end{definition}


\begin{example}\label{exp:nonoverlap}
As shown in Figure~\ref{fig:example}, we assume that $k=2$ and the aggregation function \fx $=$ $\avg$. Following that, we aim to extract top-$3$ non-overlapping $k$-influential communities. The results are $\{v_{1}, v_{2}, v_{4}\}$, $\{v_{6}, v_{7}, v_{11}\}$, and $\{v_{3}, v_{9}, v_{10}\}$. The influence value of each community is $24$, $22$ and $38/3$, respectively. There is no overlap between any two communities in the result.
\end{example}

Example~\ref{exp:nonoverlap} is used to illustrate Definition~\ref{def:non_overlap}. In fact, a non-overlapping constraint could avoid result overlaps. Thus, we give the definition of top-$r$ non-overlapping $k$-influential community search problem below:

\begin{problem}(Top-$r$ Non-overlapping Size-Constrained $k$-Influential Community).
\label{prb:nonoverlap}
Given a weighted and undirected graph $G=(V,E,w)$, the degree constraint $k$, an integer $r$, the size constraint $s$, and an aggregation function $f$, the problem is to find \underline{TO}p-$r$ \underline{N}on-overlapping size-constrained $k$-\underline{I}nfluential community (\probno) with the highest influence value under the aggregation function $f(\cdot)$.
\end{problem}

Unfortunately, regardless of whether \fx $=$ $avg$ or $sum$, the top-$r$ size-constrained $k$-influential community search problem is NP-hard. The theoretical analysis would be presented in Section~\ref{sec:hardness}.

\section{Problem Hardness}\label{sec:hardness}

We provide hardness analysis of top-$r$ $k$-influential community search problems with or without size constraint in this section. We focus on two aggregation functions: $sum$ and $avg$, since \fx $=$ $min$ has been investigated by previous works~\cite{lqym15, bclz17}. Additionally, the algorithms discussed in the preceding studies could simply be extended to the cases when \fx $=$ $max$. 

The hardness of problems can be different with various aggregation functions. When \fx $=$ $min$ or $max$, previous works have shown that it could be solved by polynomial-time algorithms for the top-$r$ $k$-influential community search problem. It is $true$ that the problem could also be solved in polynomial time when \fx $=$ $sum$. While \fx $=$ $avg$, the problem is NP-hard. When the size constraint is considered, the top-$r$ $k$-influential community search problem is NP-hard, no matter what the aggregation function is. The hardness analysis is listed below.
 
\textbf{Top-$r$} \textbf{$k$-influential} \textbf{Community} \textbf{Search.} We provide polynomial-time algorithms for the top-$r$ $k$-influential community search problem if $f(\cdot)$ $=$ $sum$. The algorithms would be presented in Section~\ref{sec:unconstrained} and we would analyze the correctness of the algorithms.

Nonetheless, when $f(\cdot)$ $=$ $avg$, the top-$r$ $k$-influential community problem is NP-hard. It could not be solved in polynomial time unless P = NP. We demonstrate the hardness of the top-$r$ $k$-influential community search problem by reducing an NP-complete problem, the decision version of the maximum clique search problem, to our problem. The decision version of the maximum clique search problem is to determine whether a graph $G$ contains a clique of size $k$.

\begin{theorem}\label{thm:avg_topr}
When \fx $=$ $avg$, the top-$r$ $k$-influential community search problem is NP-hard.
\end{theorem}

\begin{proof}
Given a graph $G=(V,E,w)$, we assign each vertex $v_i \in V$ with weight $0$. Then, we build another graph $G'=(V',E',w')$ by adding a new vertex $u$ that connects all vertices in $V$. We set the weight of the new vertex $u$ as $w_c$. Suppose that there exists a polynomial-time algorithm to address the top-$r$ $k$-influential community search problem. Then, we could determine whether there exists a ($k$-$1$)-clique since the influence value of top-$1$ $k$-influential community is $(w_c + k \cdot 0) / (k+1)$ if there exists a ($k$-$1$)-clique in graph $G$. Notably, adding any new vertex (or vertices) into such a clique would only increase the denominator of the influence value. However, the decision version of the maximum clique search problem is NP-complete. It is a contradiction. Thus, the top-$r$ $k$-influential community search problem is NP-hard, when \fx $=$ $avg$.
\end{proof}




The objective function is $g(H) = \mathds{1}_{\delta(H) \geq k} \cdot f(H)$\footnote{Generally, $\mathds{1}_{\delta(H) \geq k} = 1$ if $\delta(H) \geq k$ holds. Otherwise, $\mathds{1}_{\delta(H) \geq k} = 0$.} to denote the objective function of the top-$r$ $k$-influential community search problem, where $\delta(H)$ stands for the minimum degree of the subgraph $H$ and $f(H)$ is the aggregation function. Then we could obtain the following theorems:

\begin{theorem}\label{thm:avg_submodular}
If the aggregation function \fx $=$ $avg$, the objective function $g(\cdot)$ of the $k$-influential community search problem, is neither \textit{submodular} nor \textit{monotonic}.
\end{theorem}

\begin{proof}
For two arbitrary vertex sets $A$ and $B$, if $g(\cdot)$ is \textit{submodular}, it must hold that $g(A) + g(B) \geq g(A \cup B) + g(A \cap B)$. We reconsider Figure~\ref{fig:example}, if $k=2$, $A = \{v_5\}$ and $B=\{v_6, v_7\}$. Then, $g(A) + g(B) = 0 < g(A \cup B) + g(A \cap B) = 14 / 3$.

In terms of monotonic analysis, Figure~\ref{fig:example} is used as an example again. If $k=2$, $A = \{v_5\}$ and $B = \{v_5, v_6, v_7\}$, then $g(A) = 0 < g(B) = 14/3$. Nevertheless, if $k$ keeps unchanged, and let $A = \{v_6, v_7, v_8 \}$, $B = \{ v_5, v_6, v_7, v_8\}$. We could deduce that $g(A) = 7 > g(B) = 22/4$. This indicates that the objective function is not \textit{monotonic}.
\end{proof}



We demonstrated that the top-$r$ $k$-influential community search problem is NP-hard when \fx $=$ $avg$. Additionally, we aim to demonstrate that no constant-factor approximated solutions exist for this problem. Before presenting the theoretical analysis, we introduce the \underline{M}imimum \underline{S}ubgraph of \underline{M}inimum \underline{D}egree $\geq k$ (MSMD$_{k}$). The objective of MSMD$_{k}$ problem is to identify a subset $H$ of the vertex set $V$ such that $|H|$ is minimized and $\delta(H) \geq k$.

\begin{theorem}
When \fx $=$ $avg$, there does not exist any constant-factor approximated approaches for the top-$r$ $k$-influential community search problem. 
\end{theorem}

\begin{proof}
\cite{apps12} demonstrates that, for $k \geq 3$, the MSMD$_k$ problem does not permit any constant-factor approximation, unless $P = NP$. We show that this is also $true$ for our problem. Given a graph $G=(V,E,w)$, we assign each vertex $v_i \in V$ with weight $w_c$. Then, a dummy vertex $u$ is added, whose weight $w_u = |V| \cdot w_c$. 

Moreover, the vertex $u$ is connected to every vertex of $G$. Let $\alpha < 1$, if there exists an $\alpha$-approximated algorithm for top-$1$ $(k+1)$-influential community search problem, then we could find a $(4/\alpha)$-approximation algorithm for MSMD$_k$ problem. The reason for this is as follows: we use $S_{opt}$ to denote the optimal solution to MSMD$_k$ problem, whereas $S^{*}$ to denote the approximated solution. Then, according to the definition of average aggregation function, we have
\begin{align*}
	\frac{(|S^{*}| + |V|)\cdot w_c}{|S^{*}| + 1} \geq \alpha \frac{(|S_{opt}| + |V|)\cdot w_c}{|S_{opt}| + 1}
\end{align*}
Then, it is obvious that 
\begin{align*}
\frac{|S^{*}|}{|S_{opt}|} \leq 2 \cdot \frac{|S^{*}| + 1}{|S_{opt}| + 1} \leq \frac{2}{\alpha} \cdot \frac{|S^{*}|+|V|}{|S_{opt}| + |V|} \le \frac{2}{\alpha} \cdot \frac{2 |V|}{|V|} \le \frac{4}{\alpha}
\end{align*}
Thus, there is a contradiction.
\end{proof}

\noindent \textbf{Top-$r$ Size-constrained $k$-influential Community Search.} Since we demonstrated that if \fx $=$ $avg$, the top-$r$ $k$-influential community search problem is NP-hard. Then, the top-$r$ size-constrained $k$-influential community search problem is also NP-hard when \fx $=$ $avg$. Thus, we focus on the condition that the aggregation function is $sum$ in this part. We reduce the $k$-clique search problem to the top-$r$ size-constrained $k$-influential community search problem again to prove that the latter one is NP-hard.

\begin{theorem}
Given an aggregation function \fx $=$ $sum$ and size constraint $s$, the top-$r$ size-constrained $k$-influential community search problem is NP-hard.
\end{theorem}

\begin{proof}
We are given a weighted and undirected graph $G=(V,E,w)$ and $s = k + 1$. If we could solve the top-$r$ size-constrained $k$-influential community search problem in polynomial-time, then there also exists a polynomial-time solution to $k$-clique search problem, which is a contradiction.
\end{proof}

\section{solutions to size-unconstrained problem}
\label{sec:unconstrained}
In this section, we investigate top-$r$ size-unconstrained $k$-influential community search problem. We focus primarily on the top-$r$ $k$-influential community search problem, when \fx $=$ $sum$. We claim some critical properties of these problems. We utilize $sum$ as an example to illustrate how some pruning techniques are used to accelerate this polynomial-time problem. Moreover, we analyze the time complexity and correctness analysis of the algorithms, respectively. Furthermore, the polynomial-time algorithm could be extended to other aggregation functions.

As for certain circumstances when the problem is NP-hard, we demonstrate the relationship between the size-unconstrained and size-constrained problems. Some heuristic methods are introduced in Section~\ref{sec:size_constrained}, which could also be used to address problems with size constraint.

\subsection{Polynomial-Time Problems}

To the best of our knowledge, there are two types of top-$r$ $k$-influential community search problems that could be solved in polynomial-time. The first is a single node which dominates the influence value of the influential community, such as $\min$, $\max$. The second is that the influence value of the influential community is proportional to the number of nodes, such as $sum$, $sum$-$surplus$. We formally define them as follows:
\begin{definition}[Node Domination Aggregation Function]
\label{def:nodeDFunction}
An aggregation function $f(\cdot)$ is called a Node Domination Function if $\forall$ subgraph $H \in G$, $\exists v \in H$, s.t., $f(G[H]) = f(\{ v \})$.
\end{definition}

\begin{definition}[Size Proportional Aggregation Function]
\label{def:nodeDFunction2}
An aggregation function $f(\cdot)$ is called a Size Proportional Function if for any subgraphs $H, H' \in G$, and $H \subset H'$, then $f(G(H)) \leq f(G(H'))$.
\end{definition}
The first class of problems have been studied in~\cite{lqym15, bclz17}, where they propose some strategies for decreasing the search space. We demonstrate that the second type of problem could also be accelerated by using of the properties of aggregation function. \rev{Note that it is p-solvable if aggregation functions are monotonic.}

We take the \fx $=$ $sum$ as an example. Given aggregation function $sum$, the na\"ive approach to solve the top-$r$ $k$-influential community search problem is to compute the maximal $k$-core as well as all the connected components in the maximal $k$-core. Then, we iteratively visit each vertex in the original graph to check if it is contained in any of the aforementioned connected components. If the answer is $true$, remove the vertex from the connected component it belongs to. We maintain a priority list containing the top-$r$ connected components in each iteration.

In Algorithm~\ref{alg:sum_naive}, $L_0$ is a set of all disjoint connected components of $k$-core($G$) (Line~\ref{sum:L0}), since a $k$-core of $G$ could be divided into several disjoint connected components. For all the components, their influential values could be easily computed using $sum$. After that, $L$ is a set of the top $r$ influential components (Line~\ref{sum:L}).

Due to the property of $sum$ and all the influence values being nonnegative, we could safely prune a candidate community and all of its subgraphs if it is not in the top $r$ connected components. Thus, we try to remove one vertex from all the top $r$ influential 
communities (Lines~\ref{sum:for}-\ref{sum:endFor}). Some newly connected communities are generated (Line~\ref{sum:new}), and we combine them with the current top $r$ influential communities to produce the new top $r$ (Lines~\ref{sum:update} and \ref{sum:update2}).

\begin{algorithm}[!htbp]
    $L_{0} \leftarrow$ all the disjoint connected components of \textsc{$k$-core}($G$, $k$)\;
    \label{sum:L0}
	$L \leftarrow$ the top-$r$ influential disjoint connected components of $L_{0}$\;
	\label{sum:L}
	\For{$i \leftarrow 1$ \textbf{to} $|V|$}{
	\label{sum:for}
		$L_c \leftarrow \emptyset$\;
		\For{$j \leftarrow 1$ \textbf{to} $|L|$}{
		\label{sum:forL}
			\If{$v_i \in L[j]$}{
			\label{sum:check}
				$H \leftarrow L[j] \setminus \{v_i\}$\;
				$C \leftarrow$ \rev{all the disjoint connected components of} \textsc{$k$-core}($H$, $k$)\;
				\label{sum:new}
				$L_{c} \leftarrow L_{c} \cup C$\;
				\label{sum:update}
			}
		}
		$L \leftarrow$ the top-$r$ connected components of $L \cup L_{c}$\;
		\label{sum:update2}
	}
	\label{sum:endFor}
	\Return $L$\;
    \caption{\textsc{Sum-Na\"ive}($G$, $k$, $r$)}
	\label{alg:sum_naive}
\end{algorithm}

\noindent \textbf{Correctness.} We present Corollary~\ref{crl:sum_naive} and Theorem~\ref{thm:sum_naive} to demonstrate  that Algorithm~\ref{alg:sum_naive} could output top-$r$ $k$-influential communities correctly.

\begin{corollary}\label{crl:sum_naive}
Each connected component obtained by Algorithm~\ref{alg:sum_naive} is a $k$-influential community.
\end{corollary}

\begin{theorem}\label{thm:sum_naive}
Algorithm~\ref{alg:sum_naive} correctly identifies the top-$r$ $k$-influential communities, when \fx $=$ $sum$, given a graph $G=(V,E,w)$, degree constraint $k$, and output size constraint $r$.
\end{theorem}

\begin{proof}
When \fx $=$ $sum$, removing any vertex from a community reduces its influence value. Therefore, if a community is not considered as a top-$r$ $k$-influential community, its subgraphs cannot be regarded as a top-$r$ $k$-influential community. Furthermore, we could prune this impossible influential community without missing any top-$r$ $k$-influential community.
\end{proof}

We assume that there exists an $H$ that is a top-$r$ $k$-influential community, but is not included in the final outcome $L$. Line $10$ implies that there must be at least $r$ influential communities with a greater influence value. As a consequence, we prove the correctness of Algorithm~\ref{alg:sum_naive}.

\noindent \textbf{Complexity.} Clearly, Algorithm~\ref{alg:sum_naive} is a polynomial-time algorithm. We use the notation of $n$ to signify the number of vertices, and $m$ to denote the number of edges in graph $G$'s maximal $k$-core. Lines~\ref{sum:L0}-\ref{sum:L} require $O(n+m)$ time to complete. After that, we have a maximum of $n$ iterations. In each iteration, we delete one vertex from the connected component~(Line~\ref{sum:for}). Additionally, we check whether the vertex is in connected components at each iteration (Line~\ref{sum:check}), which requires $O(1)$ if using a hash table.

If the connected component contains the vertex, the vertex would be removed first (Lines~\ref{sum:forL}-\ref{sum:check}). The preceding step takes $O(r)$ time to complete, since we need to check at most $r$ components. Following that, we utilize breadth-first search to determine the connected components of the original component that do not include one of its vertex (Line~\ref{sum:new}); this takes $O(n+m)$. In conclusion, Algorithm~\ref{alg:sum_naive} has a time complexity of $O(n \cdot r (n+m))$.

Although the na\"ive solution is a polynomial-time algorithm, Algorithm~\ref{alg:sum_naive} is unsatisfactory as the size of the graph increases. Algorithm~\ref{alg:sum_naive} is inefficient since it requires checking if the component includes the vertex. Additionally, we ignore a key point: the influence value of community is strictly decreased, as shown in Corollary~\ref{crl:improve_framework}. We could integrate some pruning techniques into our algorithms as a result of this. Furthermore, in some cases, we are not required to determine the exact top-$r$ solutions. Thus, we propose an $\epsilon$-approximated algorithm with theoretical guarantees, where $\epsilon$ denotes the approximation ratio.

Algorithm~\ref{alg:improve_framework} illustrates the improved algorithm, while Theorem~\ref{thm:improve_framework} provides the theoretical analysis.
The main procedure of Algorithm~\ref{alg:improve_framework} is similar to Algorithm~\ref{alg:sum_naive}, but a lower bound $LB$ is used to pruned some unnecessary candidates. To begin, we initialize some variables. Among them, $L_0$ and $L$ are the same as those in Algorithm~\ref{alg:sum_naive}. Notably, a lower bound $LB$ is defined as $(1 - \epsilon) \times f(L_{max})$, where $\epsilon$ is a predefined parameter\footnote{$\epsilon = 0.1$ by default in this paper.} and $f(L_{\max})$ is the maximum influence value among all the disjoint connected components of $L_0$. In the While-loop (Lines~\ref{imp:while}-\ref{imp:whileEnd}), we set $L_{max}$ as the community with the largest influence value in $L$ (Line~\ref{imp:lmax}), and $LB$ as the new lower bound as $(1 - \epsilon)$ $\times$ the largest influence value (Line~\ref{imp:LB2}). Then, we try to remove one vertex of the $L_{max}$ to produce the new candidates (Lines~\ref{imp:for1}-\ref{imp:for2}). The main reason is that the $L_{max}$ has a high probability to produce new candidates with an influence value larger than $LB$. Two pruning rules are used for the new candidates (Lines~\ref{imp:prune1} and \ref{imp:prune2}). After that, the new candidates are added to the $L$ (Line~\ref{imp:LC}) and then set $L$ as the top $r$ influential communities (Line~\ref{imp:Lr}).
\begin{algorithm}[!htbp]
	$L_{0} \leftarrow$ compute the disjoint connected components of maximal $k$-core of $G$\;
	\label{imp:init1}
	$L \leftarrow$ the top-$r$ disjoint connected components of $L_{0}$\;
	$L_{r} \leftarrow$ the $r$-th largest influence value community in $L$\;
	\label{imp:initLr}
	$L_{\max} \leftarrow$ the community with largest influence value in $L$\;
	$LB \leftarrow f(L_{\max}) \times (1 - \epsilon)$\;
	$R \leftarrow$ communities in $L$ with influence value $\geq LB$\;
	\label{imp:init2}
	\While{$|R| < r$}{
	\label{imp:while}
		$L_{\max} \leftarrow$ the community with largest influence value in $L$\;
		\label{imp:lmax}
		$LB \leftarrow f(L_{\max}) \times (1 - \epsilon)$\;
		\label{imp:LB2}
		$L \leftarrow L \setminus L_{\max}$\;
		\For{$v \in L_{max}$}{
		\label{imp:for1}
			$H \leftarrow L_{max} \setminus \{v\}$\;
			\If{$f(H) > f(L_{r})$}{
			\label{imp:prune1}
				$C \leftarrow$ compute the connected $k$-core of $H$\;
				\For{$i = 1$ \textbf{to} $|C|$}{
					\If{$f(C[i]) \geq LB$}{
					\label{imp:prune2}
						$R \leftarrow R \cup C[i]$\;
					}
				}
				$L \leftarrow L \cup C$\;
				\label{imp:LC}
				$L \leftarrow$ the top-$r$ connected components of $L$\;
				\label{imp:Lr}
			}
		}
		\label{imp:for2}
	}
	\label{imp:whileEnd}
	\Return $R$\;
    \caption{\textsc{\prob -Improved}($G$, $k$, $r$, $f(\cdot)$, $\epsilon$)}
	\label{alg:improve_framework}
\end{algorithm}

\eat{
\begin{algorithm}[!htbp]
	\caption{\textsc{\prob -Improved}}
	\label{alg:improve_framework}
	\KwIn{A graph $G=(V,E,w)$, degree constraint $k$, output size constraint $r$, aggregation function $f$, and approximation ratio $\epsilon$}
	\KwOut{The top-$r$ $k$-influential community}
	$L_{0} \leftarrow$ compute the disjoint connected components of maximal $k$-core of $G$\;
	$L \leftarrow$ the top-$r$ disjoint connected components of $L_{0}$\;
	$L_{r} \leftarrow$ the $r$-th largest influence value community in $L$\;
	$L_{\max} \leftarrow$ the community with largest influence value in $L$\;
	$LB \leftarrow f(L_{\max}) \times (1 - \epsilon)$\;
	$R \leftarrow$ communities in $L$ with influence value $\geq LB$\;
	\While{$|R| < r$}{
		$L_{\max} \leftarrow$ the community with largest influence value in $L$\;
		$LB \leftarrow f(L_{\max}) \times (1 - \epsilon)$, $L \leftarrow L \setminus L_{\max}$\;
		\For{$v \in L_{max}$}{
			$H \leftarrow L_{max} \setminus \{v\}$\;
			\If{$f(H) > f(L_{r})$}{
				$C \leftarrow$ compute the connected $k$-core of $H$\;
				\For{$i = 1$ \textbf{to} $|C|$}{
					\If{$f(C[i]) \geq LB$}{
						$R \leftarrow R \cup C[i]$\;
					}
				}
				$L \leftarrow L \cup C$\;
				$L \leftarrow$ the top-$r$ connected components of $L$\;
			}
		}
	}
	\Return $R$\;
\end{algorithm}
}

\noindent \textbf{Correctness.} We analyze the correctness of Algorithm~\ref{alg:improve_framework}, where Corollary~\ref{crl:improve_framework} and Theorem~\ref{thm:improve_framework} are presented below.

\begin{corollary}\label{crl:improve_framework}
If we remove any vertices from the influential community, the influence value of the $k$-influential community would decrease.
\end{corollary}

\begin{remark}
\rev{
If $f$ does not satisfy Corollary~\ref{crl:improve_framework}, then Algorithm~\ref{alg:improve_framework} could not cope with it. Nevertheless, we could revise the corresponding function in the local search to solve it.}
\end{remark}

\rev{
According to Lemma 2 of~\cite{yang2010fast,ilyas2008survey,kim2013efficient}, we define the approximation factor:
\begin{definition}[Approximation Factor]
    If the influential values for the exact top $r$ result are $\{ v_1, v_2, ..., v_r \}$, a result set $\{iv_1, iv_2, ..., iv_r \}$ is defined as a $1 - \epsilon$ approximation if
\begin{equation}
    iv_r \geq (1 - \epsilon) v_r
\end{equation}
\end{definition}
}

\begin{theorem}\label{thm:improve_framework}
Given a weighted and undirected graph $G=(V,E,w)$, a degree constraint $k$, an output size constraint $r$, and an approximation ratio $\epsilon$, we use $r_{e}$ to denote the influence value of the $r$-th influential community of the output of the exact algorithm, and $r_{a}$ to denote the influence value of the $r$-th largest influence value community of the output of the Algorithm~\ref{alg:improve_framework}. If \fx $=$ $sum$, we could obtain that $r_{a}/r_{e} \geq 1 - \epsilon$.
\end{theorem}

\begin{proof}
At each iteration, we choose the maximal influential community from the candidate communities list. The maximal influential community is larger or equal to $r$-th largest influence community in terms of influence value. Thus, $r_{a}/r_{e} \geq 1 - \epsilon$.
\end{proof}

\noindent \textbf{Complexity.} Algorithm~\ref{alg:improve_framework} is straightforward and efficient. We would instantly produce a $k$-influential community whose influence value exceeds the lower bound. The time complexity of Algorithm~\ref{alg:improve_framework} is $O(rn(m+n))$. Additionally, it takes $O(n + m)$ time in Line~\ref{imp:init1}. However, in Line~\ref{imp:while}, we simply calculate $r$ iterations, and the time complexity of Lines~\ref{imp:for1}-\ref{imp:for2} is $O(r(n+m))$. To put it in a nutshell, the time complexity is $O(rn(n+m))$.


\noindent \textbf{Non-overlapping.} when \fx $=$ $sum$ and the community is size-unconstrained, we merely execute Lines~\ref{imp:init1}-\ref{imp:initLr} of Algorithm~\ref{alg:improve_framework} to compute the top-$r$ non-overlapping $k$-influential community. This is due to the fact that we would obtain the community with the largest influence each time. After obtaining it, we would remove it from the graph. As a consequence, we could obtain the correct result by performing Lines~\ref{imp:init1}-\ref{imp:initLr} of Algorithm~\ref{alg:improve_framework}.

\noindent \textbf{Discussion.} If \fx $\neq$ $sum$, the preceding algorithm could potentially be expanded to solve other top-$r$ $k$-influential communities. For instance, \fx $=$ $sum$-$surplus$ also satisfies Corollary~\ref{crl:improve_framework}. Thus, we could use Algorithm~\ref{alg:improve_framework} to solve the top-$r$ $k$-influential community search problem for $sum$-$surplus$.

\noindent \textit{Power-Law Graph.} In practice, the degree distribution of the graph conforms to the power-law distribution. Thus, it is critical to analyze the complexity of our algorithm under power-law distribution.

\begin{definition}[Power-Law Graph]
\label{def:power_law}
Given a graph $G=(V,E)$, the degree distribution of the graph follows a power-law distribution, if the fraction $P(k)$ of nodes in the graph having $k$ connections to other nodes goes for large values of $k$ as $P(k) \thicksim k^{-\gamma}$, where $2 < \gamma < 3$.
\end{definition}

\begin{lemma}
When we consider the power-law graph, the number of nodes with a degree greater or equal to $k$ is $n/((\gamma-1)k^{\gamma-1})$, and the number of edges is bounded by $n/(2(\gamma-2)k^{\gamma-2})$.
\end{lemma}

\begin{proof}
According to the definition of the power-law graph, the number of node with a degree larger or equal to $k$ is $n/((\gamma-1)k^{\gamma-1})$. Then, the number of nodes whose degree is equal to $k$ is $n/k^{\gamma}$. Thus, the total number of edges is bounded by
\begin{align*}
\frac{n}{2} \sum_{d=k}^{\infty} \frac{1}{d^{\gamma-1}} \leq \frac{n}{2} \int_{k}^{\infty} x^{-\gamma+1} dx \leq \frac{n}{2(\gamma-2)k^{\gamma-2}}
\end{align*}
This completes the proof.
\end{proof}

According to the definition of a power-law graph, then the time complexity of our algorithm under a power-law graph could be $O(\frac{rn^{2}((k+2)\gamma - k - 4)}{2(\gamma-1)^{2}(\gamma-2)k^{(2\gamma - 2)}})$, when the degree constraint is $k$, where $2 < \gamma < 3$.

\subsection{NP-Hardness}
We investigate the \prob problem under several aggregation functions when it could be solved in polynomial time. Nevertheless, the problem would be NP-hard with some aggregation functions. In Section~\ref{sec:pre}, we demonstrate this problem's hardness analysis. As a result, heuristic algorithms must be developed to address the NP-hard problem.

\section{solutions to size-constrained problem}\label{sec:size_constrained}


Section~\ref{sec:pre} demonstrates that when the size of an influential community is constrained, the \prob problem is NP-hard. In this section, we provide the exact algorithms for the \prob problem. Given that the top-$r$ size-constrained is NP-hard, we present a heuristic algorithm to address it in polynomial time. Note that all the techniques could be easily extended to the \probno problem.

\subsection{Exact Algorithm}

The exact algorithm is quite time-consuming, and the na\"ive exact algorithm for the \prob problem is illustrated in Algorithm~\ref{alg:exact}. The key idea is to enumerate all possible solutions and return the top-$r$ $k$-influential community. In Algorithm~\ref{alg:exact}, we set the candidate community set as $\emptyset$ (Line~\ref{exact:init}). Then, a for-loop (Lines~\ref{exact:for1}-\ref{exact:for2}) enumerates all possible communities whose size ranges from $k+1$ to $s$ (Line~\ref{exact:enum}). Then, the $k$-core and connected constraints are verified in Line~\ref{exact:conditions}. The newly added candidates would be inserted into $L$ (Line~\ref{exact:add}), and then the top-$r$ influential connected components would be returned (Line~\ref{exact:return}).

\begin{algorithm}[!htbp]
$L \leftarrow \emptyset$\;
\label{exact:init}
	\For{$i = k+1$ \textbf{to} $s$}{
	\label{exact:for1}
		$L_{c} \leftarrow$ enumerate all possible vertex sets whose size is $i$\;
		\label{exact:enum}
		\For{$c \in L_{c}$ and $c$ is a connected $k$-core}{
		\label{exact:conditions}
			$L \leftarrow L \cup c$\;
			\label{exact:add}
		}
	}
	\label{exact:for2}
	\Return top-$r$ connected components of $L$\;
	\label{exact:return}
    \caption{\textsc{\prob -Exact}($G$, $k$, $r$, $s$, $f(\cdot)$)}
	\label{alg:exact}
\end{algorithm}

\eat{
\begin{algorithm}[!htbp]
	\caption{\textsc{\prob -Exact}}
	\label{alg:exact}
	\KwIn{A graph $G=(V,E,w)$, degree constraint $k$, output size constraint $r$, greedy signal $gs$, size constraint $s$, and aggregation function $f$}
	\KwOut{The top-$r$ $k$-influential community}
	$L \leftarrow \emptyset$\;
	\For{$i = k+1$ \textbf{to} $s$}{
		$L_{c} \leftarrow$ enumerate all possible vertex sets whose size is $i$\;
		\For{$c \in L_{c}$ and $c$ is a connected $k$-core}{
			$L \leftarrow L \cup c$\;
		}
	}
	\Return top-$r$ connected components of $L$\;
\end{algorithm}
}

\noindent \textbf{Time Complexity.}~There are $\sum_{i=k+1}^{s}{C_{n}^{i}}$ possible vertex sets. For each of them, Line~\ref{exact:conditions} takes $O(m+n)$ to verify the $k$-core and connected constraints. Thus, the time complexity of na\"ive exact algorithm is $O(\sum_{i=k+1}^{s}{C_{n}^{i}} \cdot (m+n) )$.
To tackle this issue, we provide efficient heuristic solutions to the top-$r$ size-constrained $k$-influential community search problem, namely \textit{Local Search}.

\subsection{Local Search Algorithm}

The local search algorithm is to begin with each vertex $u$ in the graph, and then search the $s$ nearest neighbors of $u$. Following that, we determine whether $u$ could construct $k$-core in the absence of its neighbors. Algorithm~\ref{alg:local_search} summarizes the approach.
In Algorithm~\ref{alg:local_search}, $L$ is the set of disjoint connected components after the $k$-core($G$) (Line~\ref{local:L}). Then, a for-loop explores every vertex $v$ $\in V$ (Lines~\ref{local:for1}-\ref{local:for2}). The $s$-nearest neighbors of $v_{i}$ would be assigned to $V_i$ (Line~\ref{local:vi}), which would be verified in the Strategy Procedure according to various $f(\cdot)$ (Line~\ref{local:str}). \rev{Notably, if $v_i$ does not have $s$ neighbors, we would explore its $2$-hop neighbors. Thus, a BFS could be conducted in Line~\ref{local:vi}.} If the $greedy$ signal is set, the vertices in $V_i$ would be sorted in the descending order of influence value (Lines~\ref{local:greedy}-\ref{local:greedy2}). Line~\ref{local:str} would add new candidates to $L$. After that, the top $r$ influential communities of $L$ would be sorted and returned (Lines~\ref{local:sort}-\ref{local:return}).

\begin{algorithm}[!htbp]
    \caption{\textsc{LocalSearch}($G$, $k$, $r$, $greedy$, $s$, $f(\cdot)$)}
	\label{alg:local_search}
	$L$ $\leftarrow$ Compute the maximal $k$-core of $G$\;
	\label{local:L}
	\For{$i=1$ \textbf{to} $|V|$}{
	\label{local:for1}
		\If{$v_i$ is not removed}{
			$V_i \leftarrow$ the vertex set of an $s$-nearest neighbor of $v_{i}$\;
			\label{local:vi}
			\If{$greedy$ $=$ $True$}{
			\label{local:greedy}
				Sort the vertices in $V_i$ in the descending order of influence value\;
				\label{local:greedy2}

			}
			$L \leftarrow$ \textsc{Strategy}($V_i, L, gs, f$)\;
			\label{local:str}
		}
	}
	\label{local:for2}
	Sort $L$ by the influence value of each element\;
	\label{local:sort}
	\Return $L$\;
	\label{local:return}
\end{algorithm}
\eat{
\begin{algorithm}[!htbp]
	\caption{\textsc{Local Search}}
	\label{alg:local_search}
	\KwIn{A graph $G=(V,E,w)$, degree constraint $k$, output size constraint $r$, greedy signal $gs$, size constraint $s$, and aggregation function $f$}
	\KwOut{The top-$r$ $k$-influential community}
	Compute the maximal $k$-core of $G$\;
	\For{$i=1$ \textbf{to} $|V|$}{
		\If{$v_i$ is not removed}{
			$V_i \leftarrow$ the vertex set of an $s$-nearest neighbor of $v_{i}$\;
			\If{$gs$ $=$ \textbf{True}}{
				Sort the vertices in $V_i$ in the descending order\;
			}
			$L \leftarrow$ \textsc{Strategy}($V_i, L, gs, f$)\;
		}
	}
	Sort $L$ by the influence value of each element\;
	\Return $L$\;
\end{algorithm}
}

With different aggregation functions, we would use different strategies in Algorithm~\ref{alg:local_search} Line~\ref{local:str}. We present two strategies to illustrate how to design heuristic strategies with various aggregation functions. As for $sum$, the first procedure is used to tackle the top-$r$ size-constrained $k$-influential community problem. In $SumStrategy$,
the candidate set $C$ is set $\emptyset$ and $L_r$ is set as the $r$-th largest influential community in Lines~\ref{sums:init1} and~\ref{sums:init2}, respectively. A while-loop (Lines~\ref{sums:while1s}-\ref{sums:while1e}) selects the first $s$ vertices as a candidate community. Then, a while-loop (Lines~\ref{sums:while2s}-\ref{sums:while2e}) compute the connected $k$-core using the vertices in $C$ and update (Line~\ref{sums:update}) the result influential community $L$ if $f(C) > f(L_{r})$.

\rev{
    \begin{remark}
    The effects of local search heuristic depend on the diameter of the result community. If small, local search works wells since local search preferentially search the closer candidates.
    \end{remark}
}
\begin{procedure}[!htbp]
	$C \leftarrow \emptyset$\;
	\label{sums:init1}
	$L_{r} \leftarrow$ the $r$-th largest influence value community\;
	\label{sums:init2}
	\While{$|C| < s$}{
	\label{sums:while1s}
		$C \leftarrow C \cup V.front$, $V \leftarrow V \setminus V.front$\;
		\Comment{$V.front$ is the first element of $V$}
	}
	\label{sums:while1e}
	\While{$|C| > k$ and $f(C) > f(L_{r})$}{
	\label{sums:while2s}
		\uIf{$C$ is $k$-core}{
			$L \leftarrow L \setminus L_{r}$, $L \leftarrow L \cup C$\;
			\label{sums:update}
			\textbf{break}\;
		}
		\Else{
			$C \leftarrow C \setminus C.last$\;\Comment{$C.last$ is the last element of $C$}
		}
	}
	\label{sums:while2e}
	\Return $L$\;
	\caption{SumStrategy(V, L, f)}
	\label{proc:sum}
\end{procedure}

\noindent \textbf{Time Complexity.} In Algorithm~\ref{alg:local_search}, we compute the maximal $k$-core of graph $G$ (Line~\ref{local:L}) in $O(n+m)$ time. Then, in Line ~\ref{local:for1}, it computes $n$ iterations. Line~\ref{local:vi} computes the $k$-nearest neighbor of $v_i$ in $O(k)$ time. If the greedy signal is $true$, then the vertices are sorted, which takes $O(s \log s)$. Moreover, the Procedure~\ref{proc:sum} requires $O(s^2)$. Thus, if utilizing the greedy strategy, the time complexity of Algorithm~\ref{alg:local_search} is $O(nk s^{3} \log s)$. Otherwise, the time complexity of Algorithm~\ref{alg:local_search} is $O(nks^{3})$.

Nonetheless, if \fx = $avg$, we will modify the procedure by using its characteristics. The Procedure~\ref{proc:avg} concludes the strategy. In Procedure~\ref{proc:avg}, we initialize candidate vertex set $C$ and influential candidate community set $L_c$ as $\emptyset$ (Line~\ref{avgs:init}). Line~\ref{avgs:lr} sets $L_r$ as the top $r$-th largest influential community. A while-loop (Lines~\ref{avgs:while1s}-\ref{avgs:while1e}) tries to add vertex $\in V$ to $C$ and test if $C$ could be a new candidate community to $L$. If the $greedy$ signal is used (Line~\ref{avgs:greedy}), we could safely prune $L_r$ from $L$ (Line~\ref{avgs:update}) and add $C$ to $L$, since the influence value of the latter vertex $\in V$ is no larger than that of the current one. Otherwise, we add $C$ to $L_c$ as a candidate (Line~\ref{avgs:nog}). Then, we choose the $L_c[i]$ with the maximum influence value among $L_c$, and add it to $L$.
The time complexity of Algorithm~\ref{alg:local_search} is the same as the analysis outlined above.

\noindent \textbf{Non-overlapping.} The objective is to compute the top-$r$ size-constrained non-overlapping $k$-influential community. As a result, we could slightly modify the Local Search Algorithm (Algorithm~\ref{alg:local_search}). In Line~\ref{local:str}, when designing strategy, we could remove each $k$-influential community once it is obtained by a local search algorithm.

\begin{procedure}[!htbp]
	$C \leftarrow \emptyset$, $L_{c} \leftarrow \emptyset$\;
	\label{avgs:init}
	$L_{r} \leftarrow$ the $r$-th largest influence value community\;
	\label{avgs:lr}
	\While{$|C| < s$}{
	\label{avgs:while1s}
		$C \leftarrow C \cup V.front$, $V \leftarrow V \setminus V.front$\;
		\If{$|C| > k$, $f(C) > f(L_{r})$ and $C$ is $k$-core}{
			\uIf{$greedy$ $=$ $True$}{
			\label{avgs:greedy}
				$L \leftarrow L \setminus L_{r}$, $L \leftarrow L \cup C$\;
				\label{avgs:update}
				\textbf{break}\;
			}
			\Else{
				$L_{c} \leftarrow L_{c} \cup C$\;
				\label{avgs:nog}
			}
		}
	}
	\label{avgs:while1e}
	\If{$greedy$ $=$ $False$}{
		\label{avgs:nog2}
		$L_{max} \leftarrow \mathop{\arg\max}\limits_{i \in \{1, 2, \dots, |L_{c}|\}} L_{c}[i]$\;
		$L \leftarrow L \setminus L_{r}$, $L \leftarrow L \cup L_{\max}$\;
	}
	\label{avgs:nog3}
	\Return $L$\;
	\caption{AvgStrategy(V, L, greedy, f)}\label{proc:avg}
\end{procedure}

\section{experiments}\label{sec:exp}

In this section, we conduct extensive experiments on our proposed technique. All algorithms are implemented in C++. All experiments are conducted on a Linux machine with an Intel Xeon $2.70$GHz CPU and $400$GB memory.

\begin{table}[!htbp]
\centering
\begin{small}
\renewcommand{\arraystretch}{1.2}
\caption{Datasets}
\vspace{-2mm}
\label{table:Dataset Table}
\setlength{\tabcolsep}{1.8mm}{
\scalebox{0.9}{
\begin{tabular}{|l|l|l|l|l|l|} \hline
\cellcolor{gray!25}\textbf{Dataset} & \cellcolor{gray!25}\#vertices & \cellcolor{gray!25}\#edges & \cellcolor{gray!25}$d_{\max}$ & \cellcolor{gray!25}$d_{\avg}$ & \cellcolor{gray!25}$k_{\max}$  \\ \hline
\rev{DomainPub} & \rev{$22,692$} &  \rev{$60,830$}  & \rev{$125$} & \rev{$5.35$} & \rev{$31$}  \\ \hline
Email & $36,692$ & $183,831$ & $1,383$ & $10.02$ & $43$  \\ \hline
DBLP & $317,080$ & $1,049,866$ & $343$ & $6.62$ & $113$  \\ \hline
Youtube & $1,134,890$ & $2,987,624$ & $28,754$ & $5.27$ & $51$  \\ \hline
Orkut & $3,072,441$ & $117,185,083$ & $33,313$ & $76.28$ & $253$  \\ \hline
LiveJournal & $3,997,962$ & $34,681,189$ & $14,815$ & $17.35$ & $360$  \\ \hline
FriendSter & $65,608,366$ & $1,806,067,135$ & $5,214$ & $55.06$ & $304$  \\ \hline
\end{tabular}}
}
\vspace{-2mm}
\end{small}
\end{table}

\begin{figure*}[!t]
	\centering
	\vspace{-1mm}
	\begin{small}
		\begin{tabular}{cccccc}
			\hspace{-6mm} \includegraphics[height=22mm]{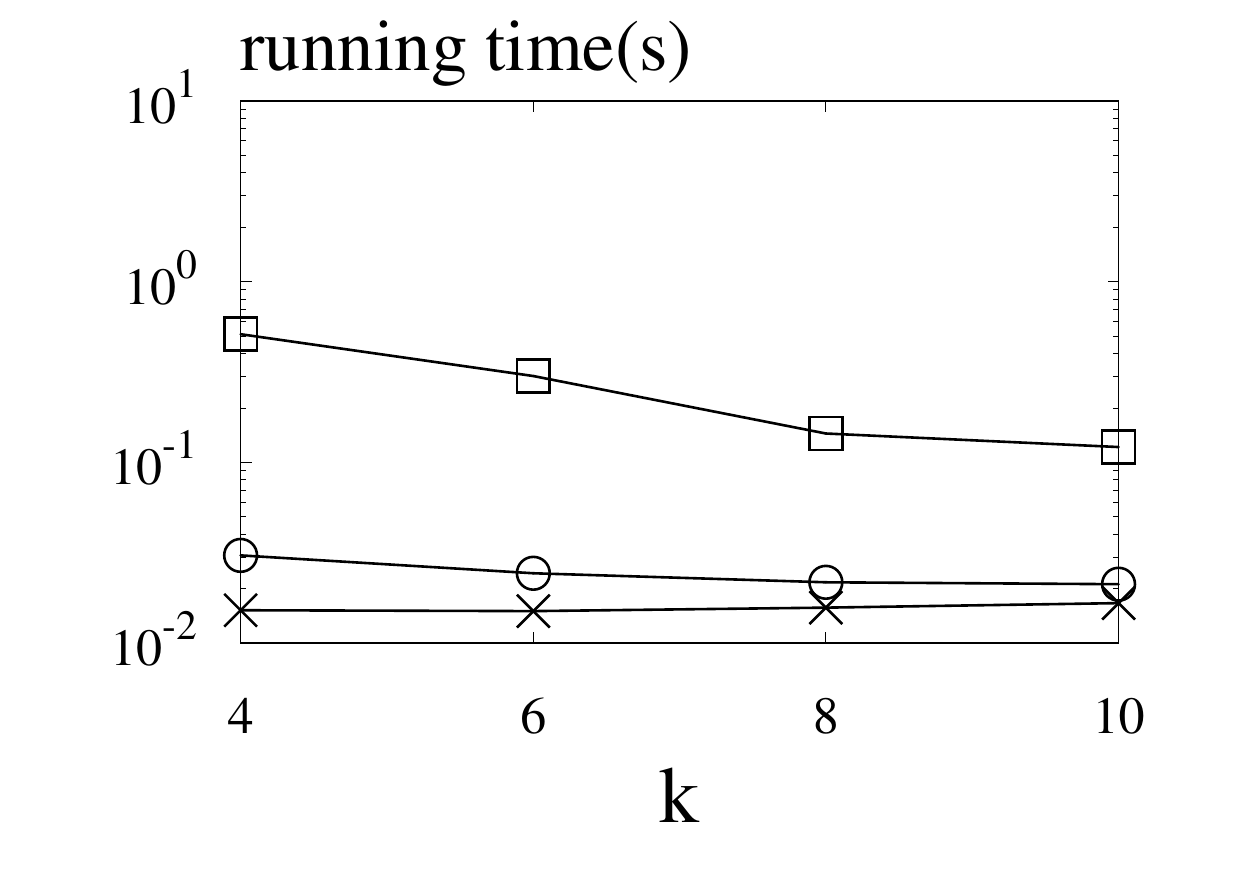} &
			\hspace{-6mm} \includegraphics[height=22mm]{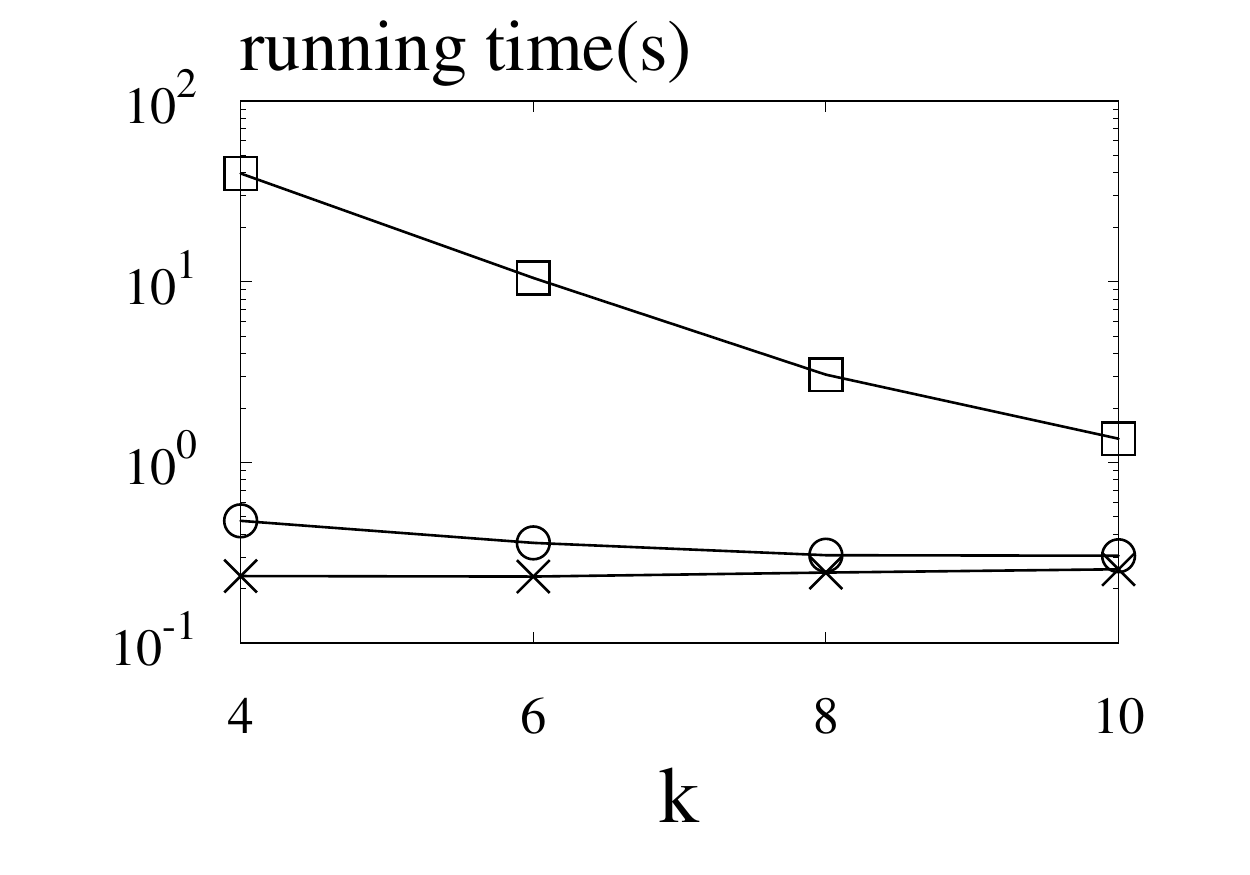} &
			\hspace{-6mm} \includegraphics[height=22mm]{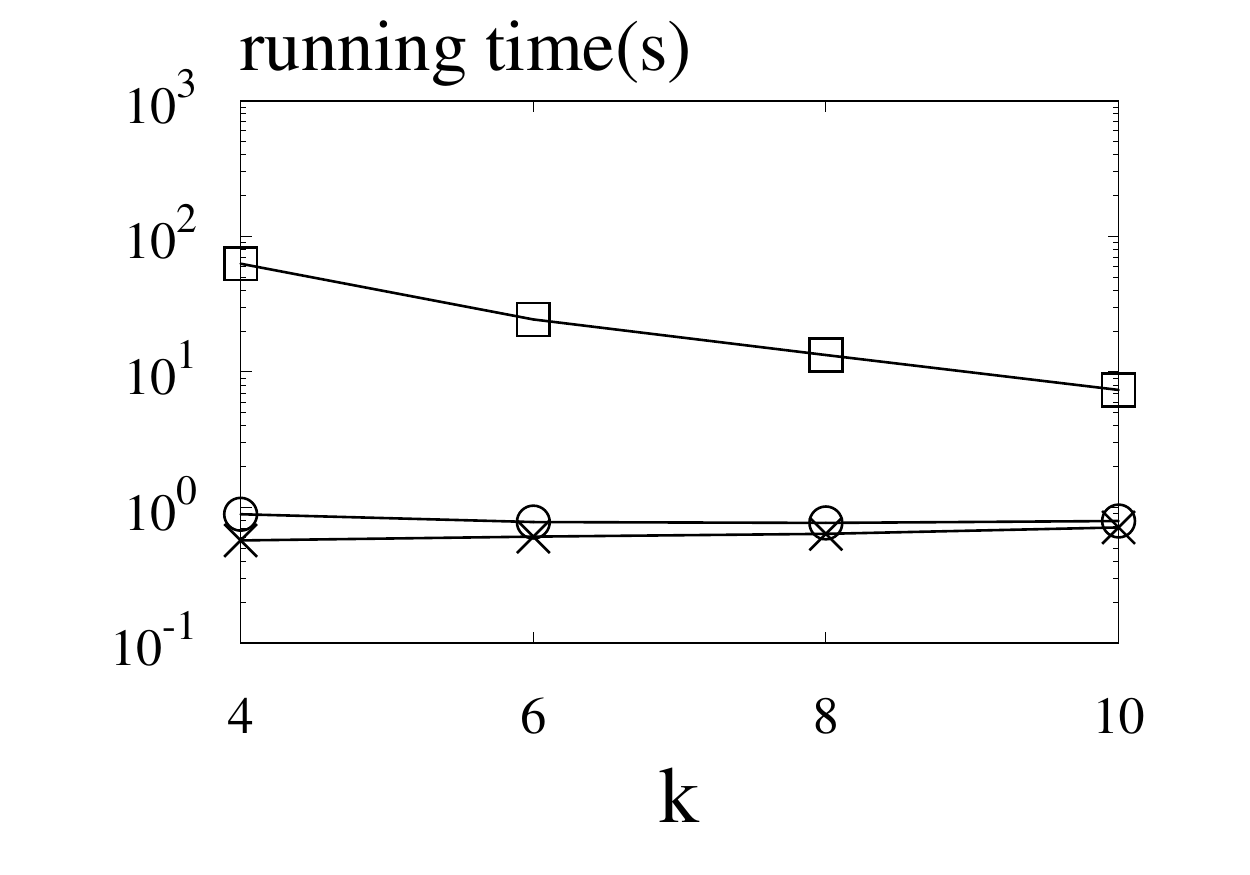} &
			\hspace{-6mm} \includegraphics[height=22mm]{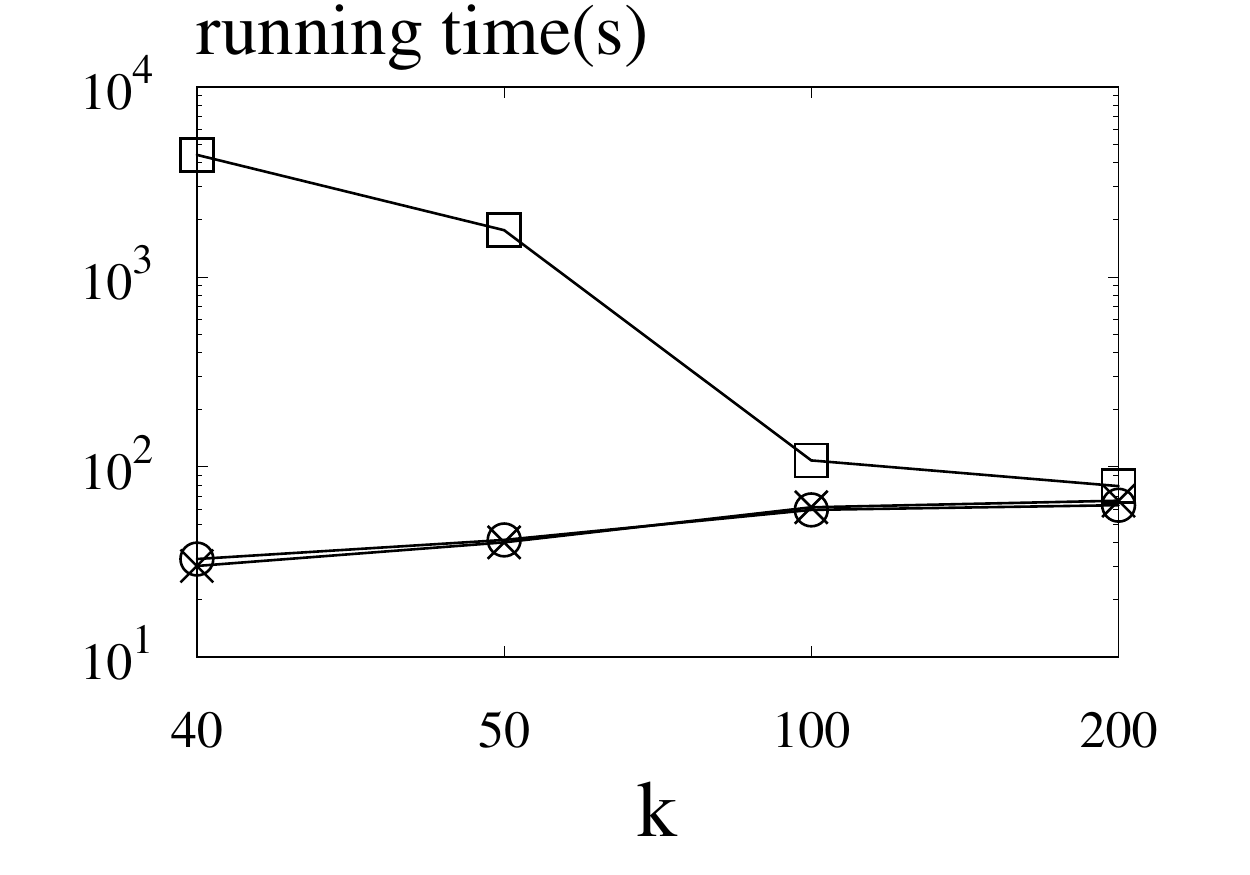} &
			\hspace{-6mm} \includegraphics[height=22mm]{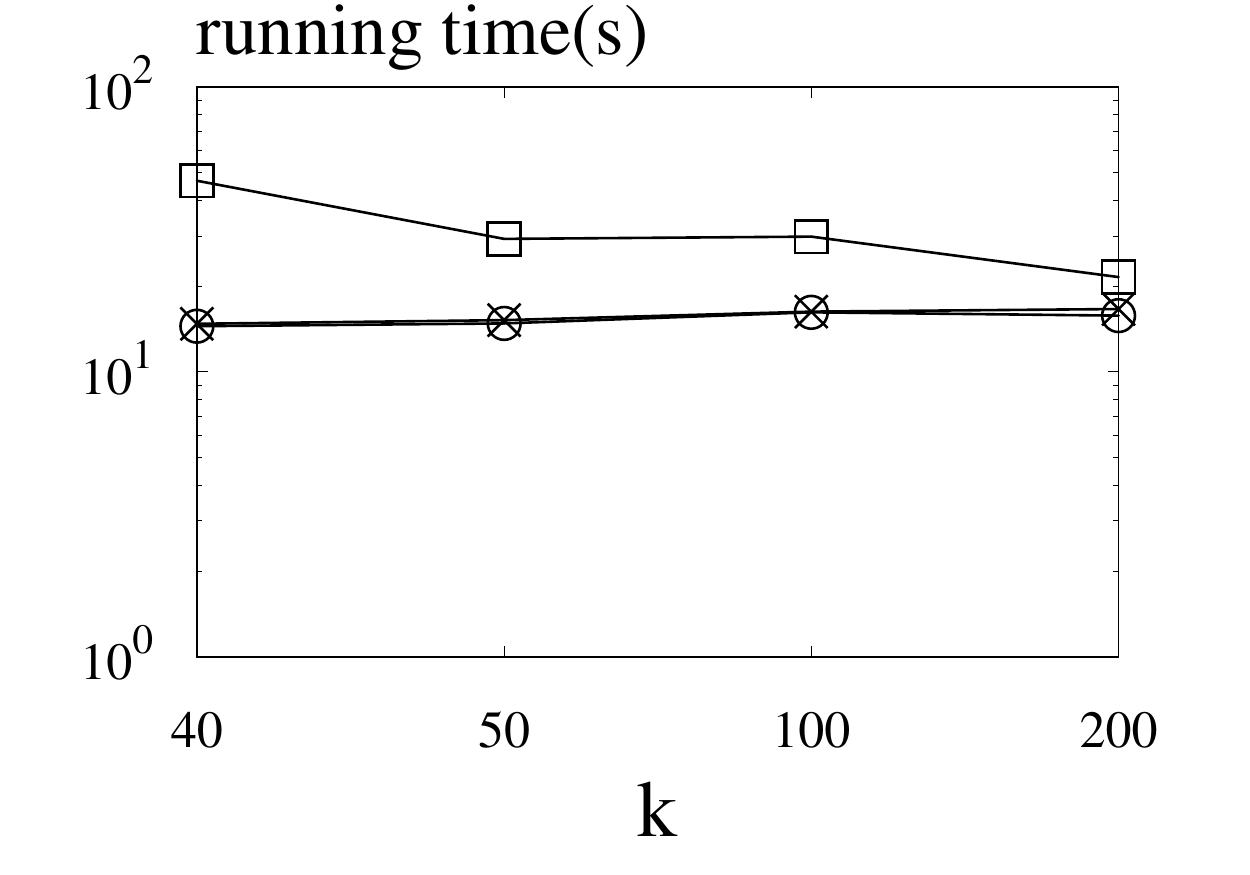} &
			\hspace{-6mm} \includegraphics[height=22mm]{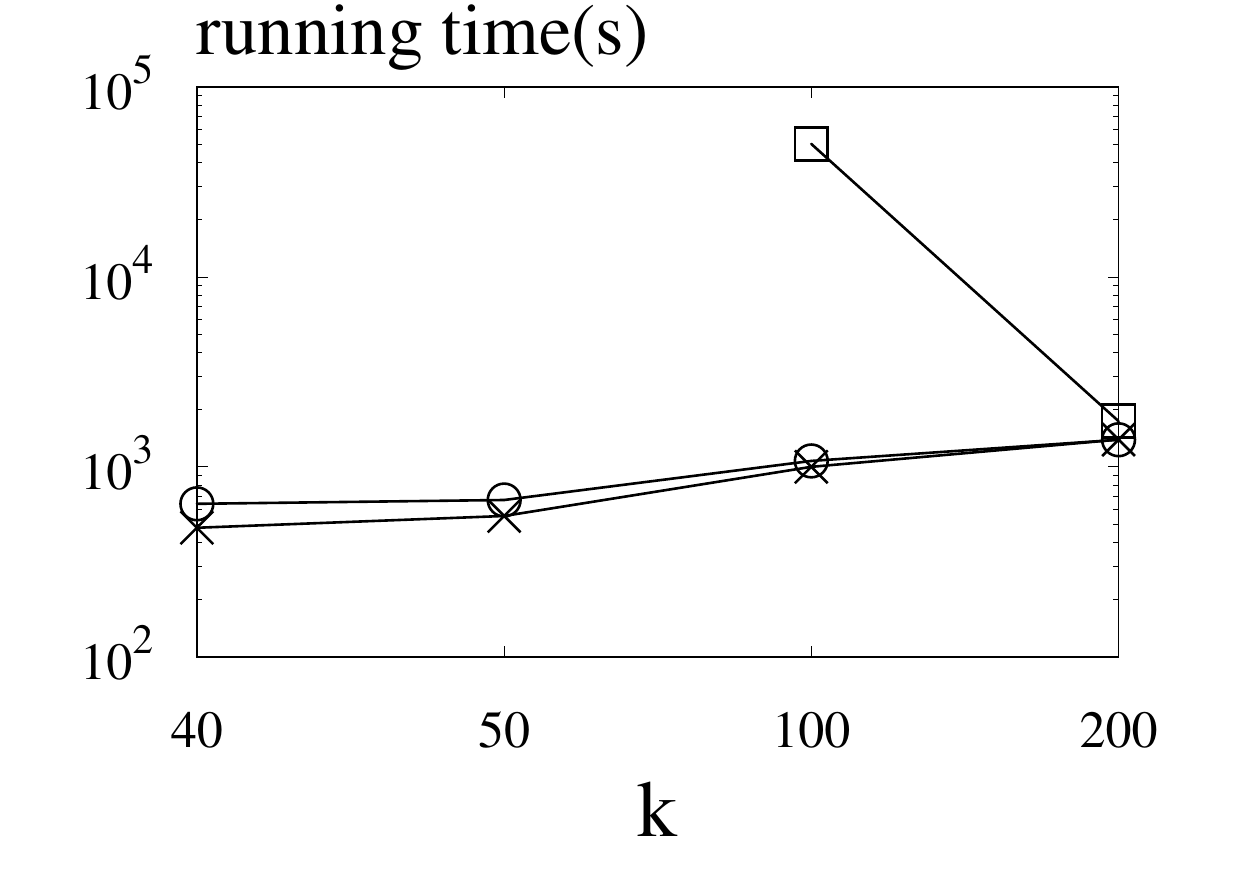}
			\\[-1mm]
			\hspace{-4mm} (a) Email &
			\hspace{-4mm} (b) DBLP &
			\hspace{-4mm} (c) Youtube &
			\hspace{-4mm} (d) Orkut &
			\hspace{-4mm} (e) Livejournal &
			\hspace{-4mm} (f) FriendSter \\[-1mm]
		\end{tabular}
		\vspace{-2mm}
		\caption{Running time vs. $k$ (sum, size-unconstrained)}
		\label{fig:time vs k sum}
		\vspace{-3mm}
	\end{small}
\end{figure*}

\begin{figure*}[!t]
	\centering
	\vspace{-1mm}
	\begin{small}
		\begin{tabular}{cccccc}
			\multicolumn{6}{c}{\hspace{-6mm} \includegraphics[height=10mm]{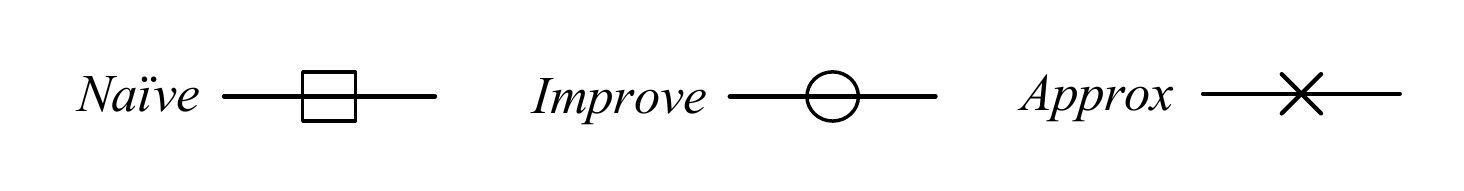}}  \\[-3mm]
			\hspace{-6mm} \includegraphics[height=22mm]{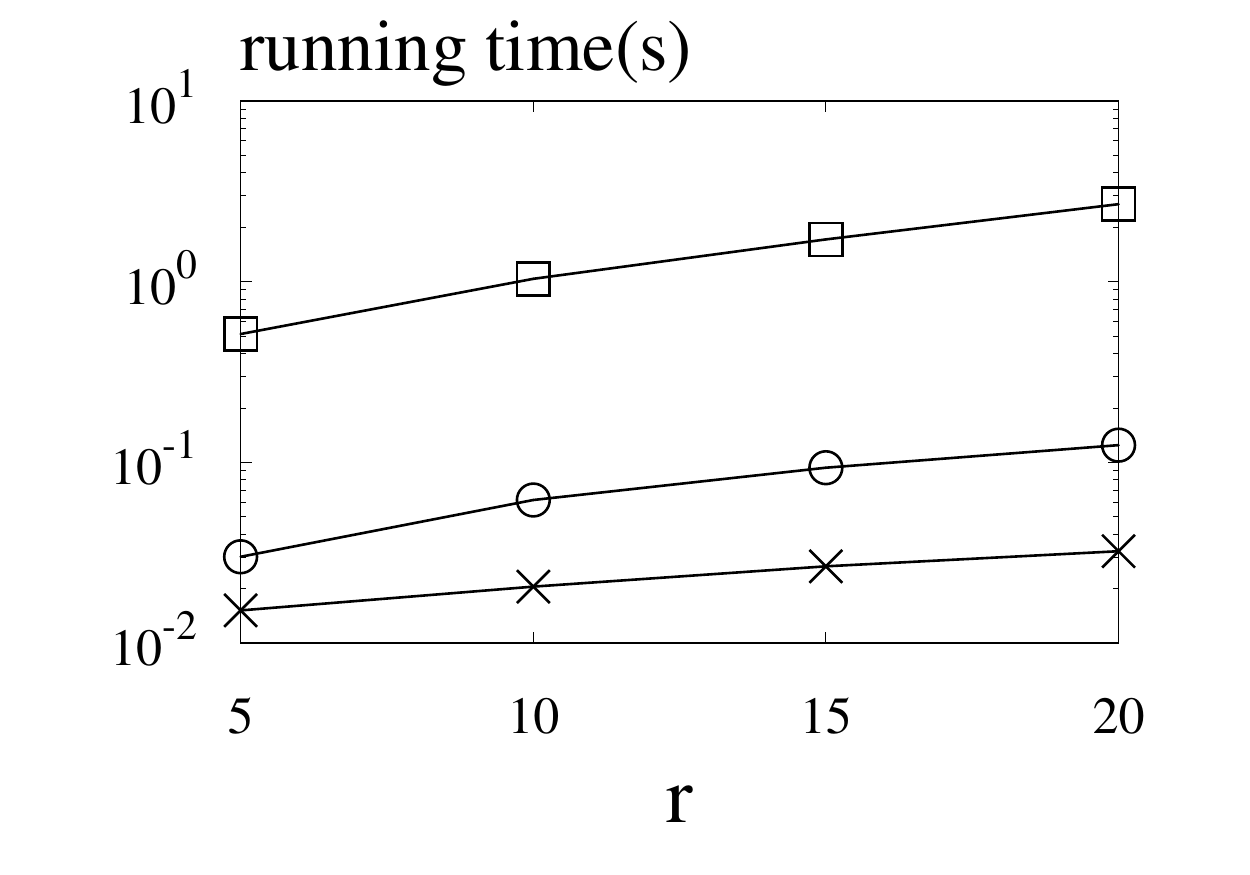} &
			\hspace{-6mm} \includegraphics[height=22mm]{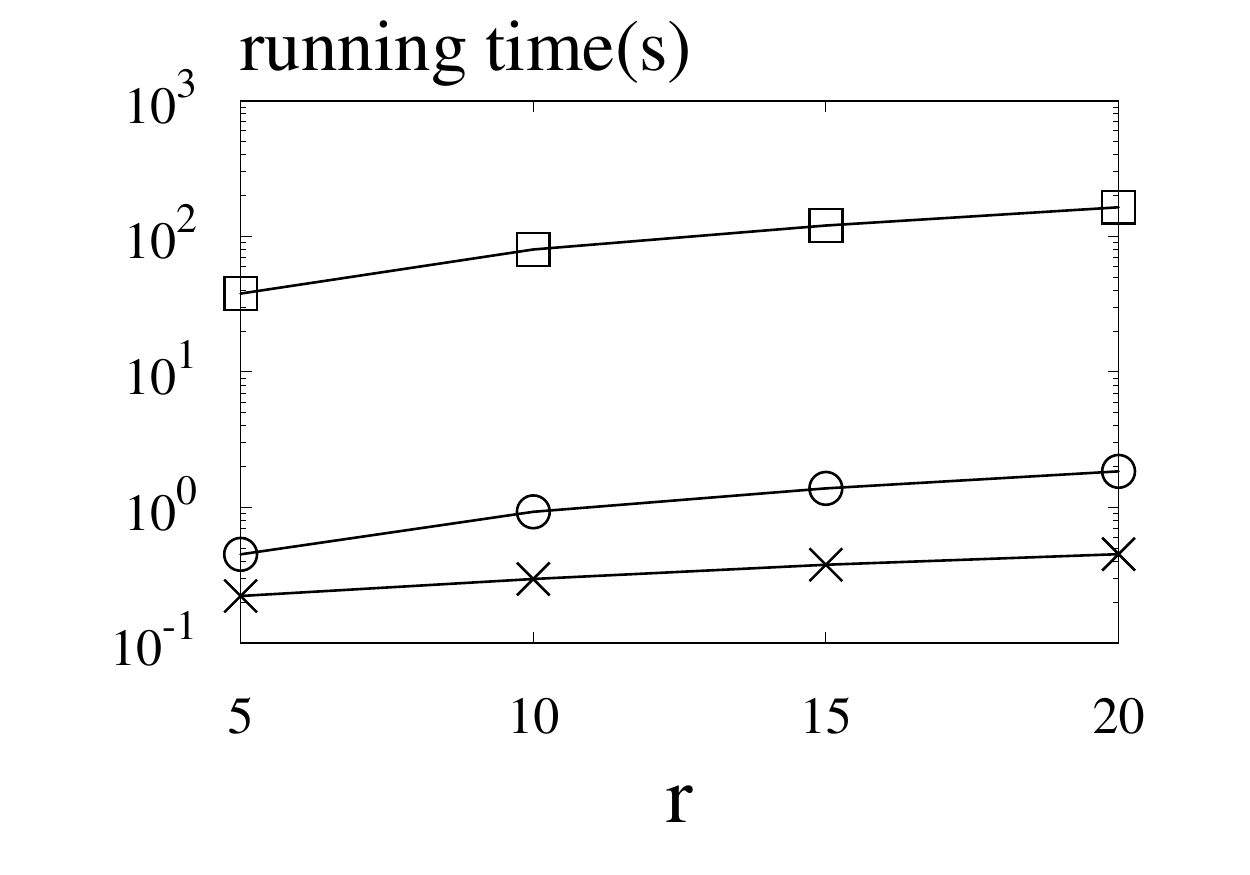} &
			\hspace{-6mm} \includegraphics[height=22mm]{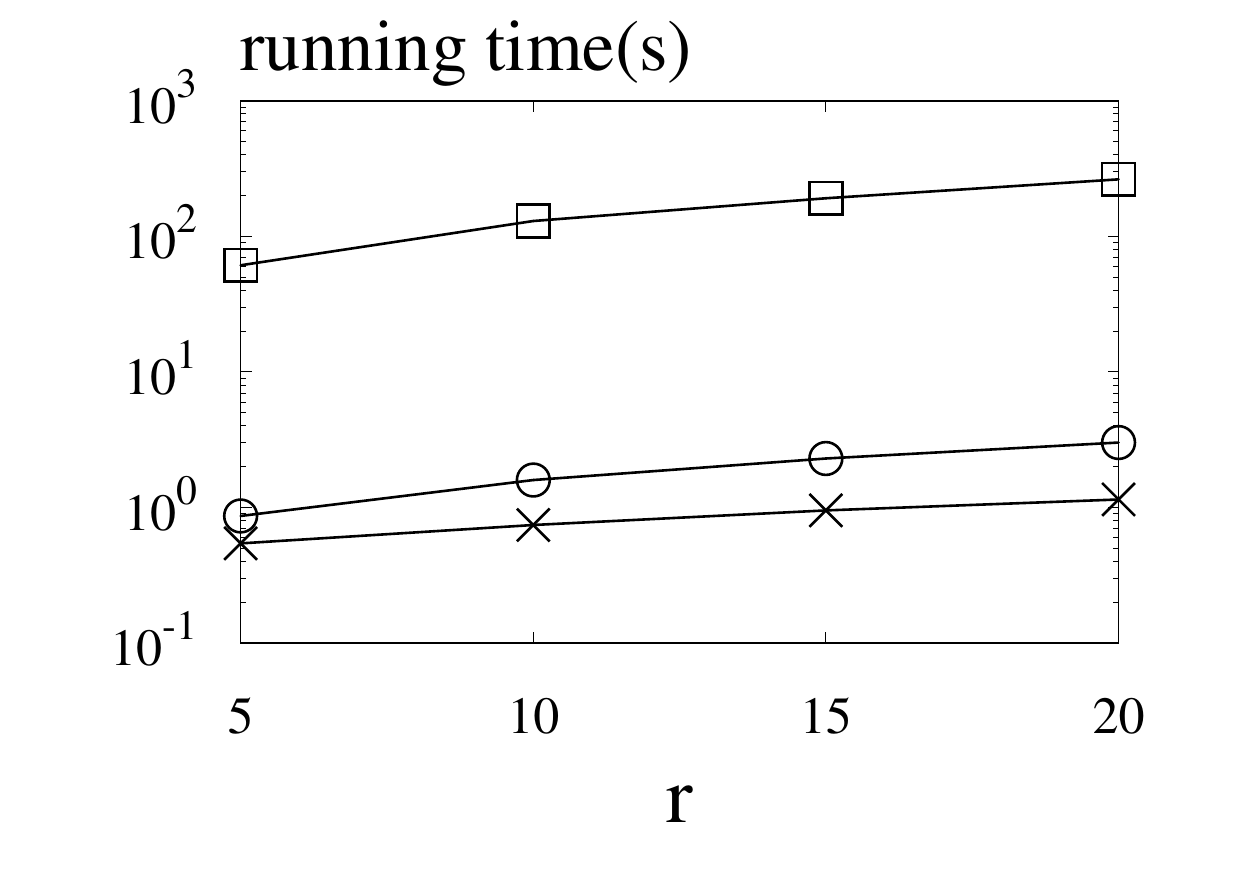} &
			\hspace{-6mm} \includegraphics[height=22mm]{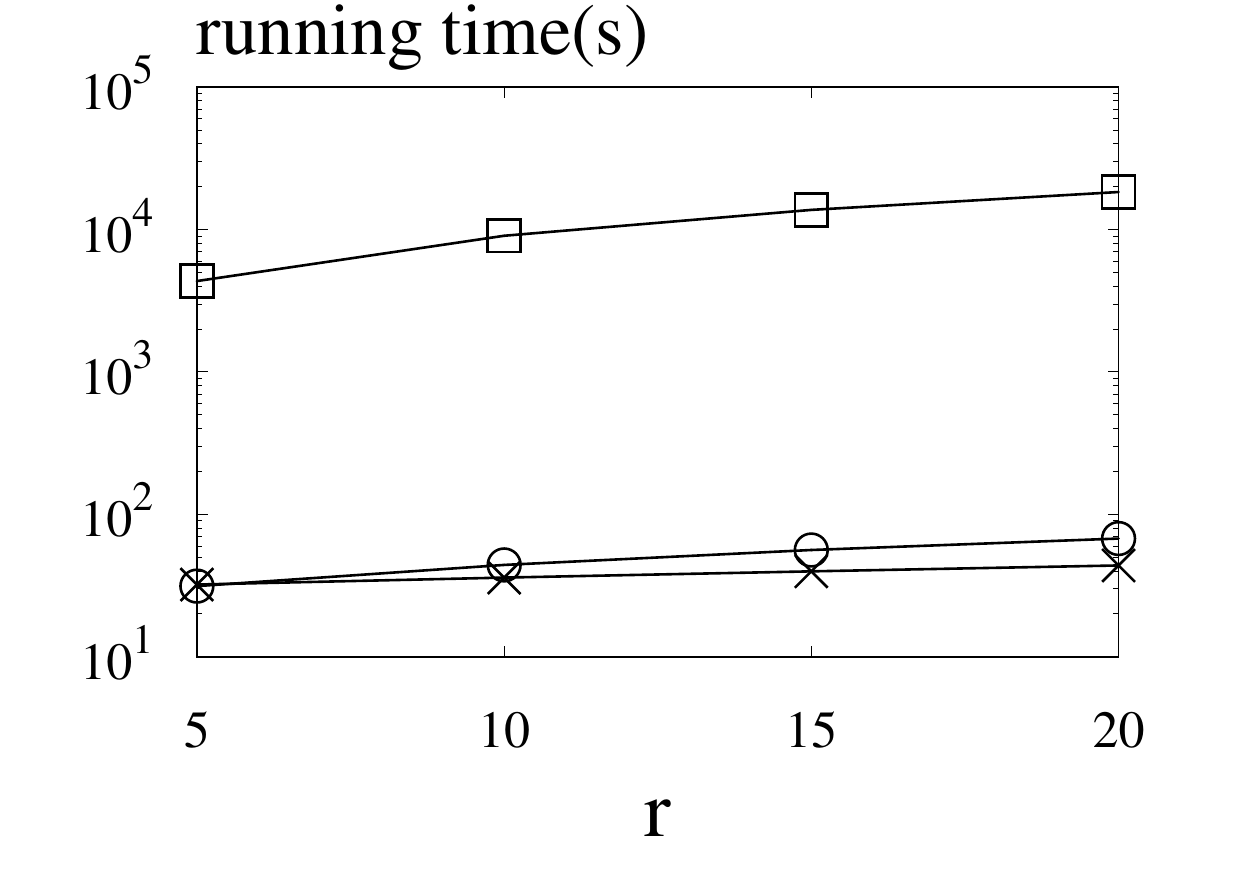} &
			\hspace{-6mm} \includegraphics[height=22mm]{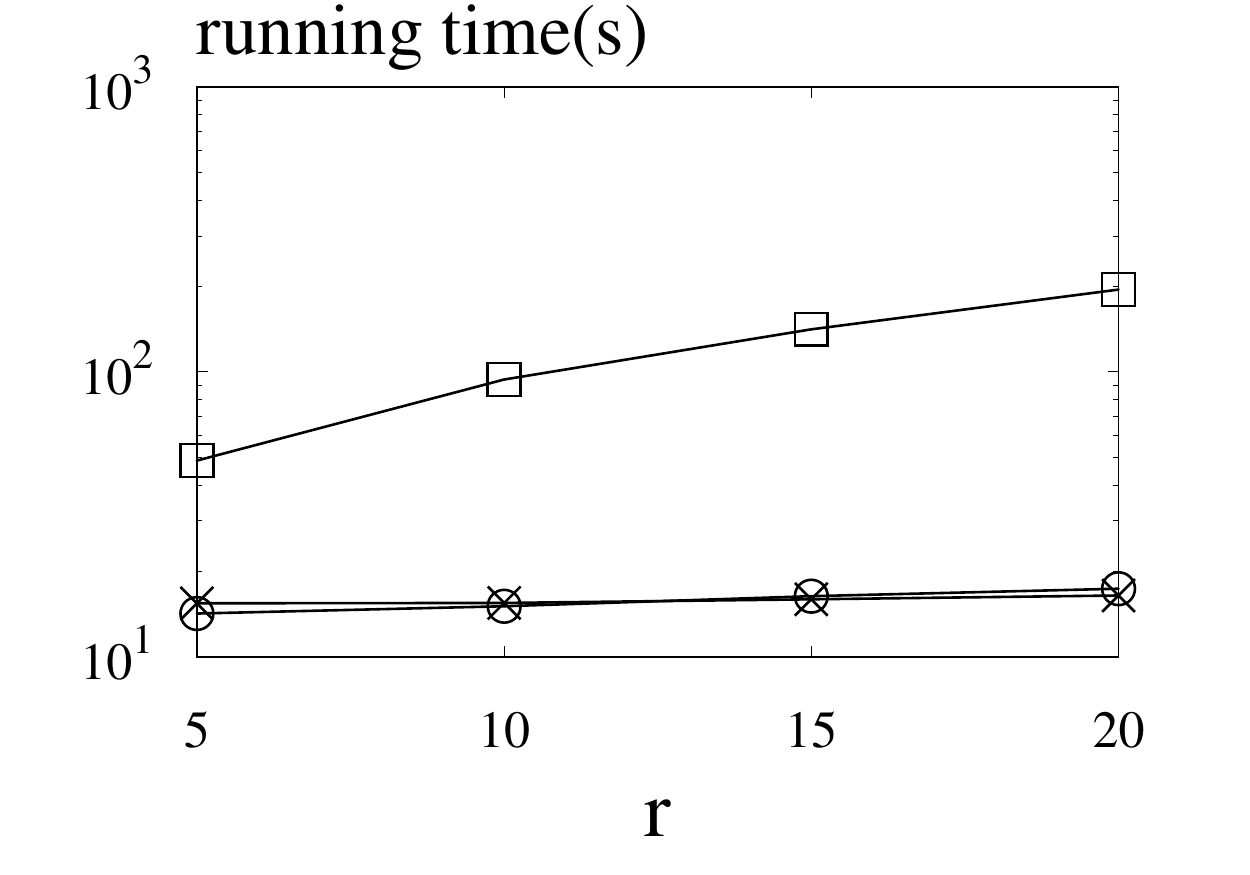} &
			\hspace{-6mm} \includegraphics[height=22mm]{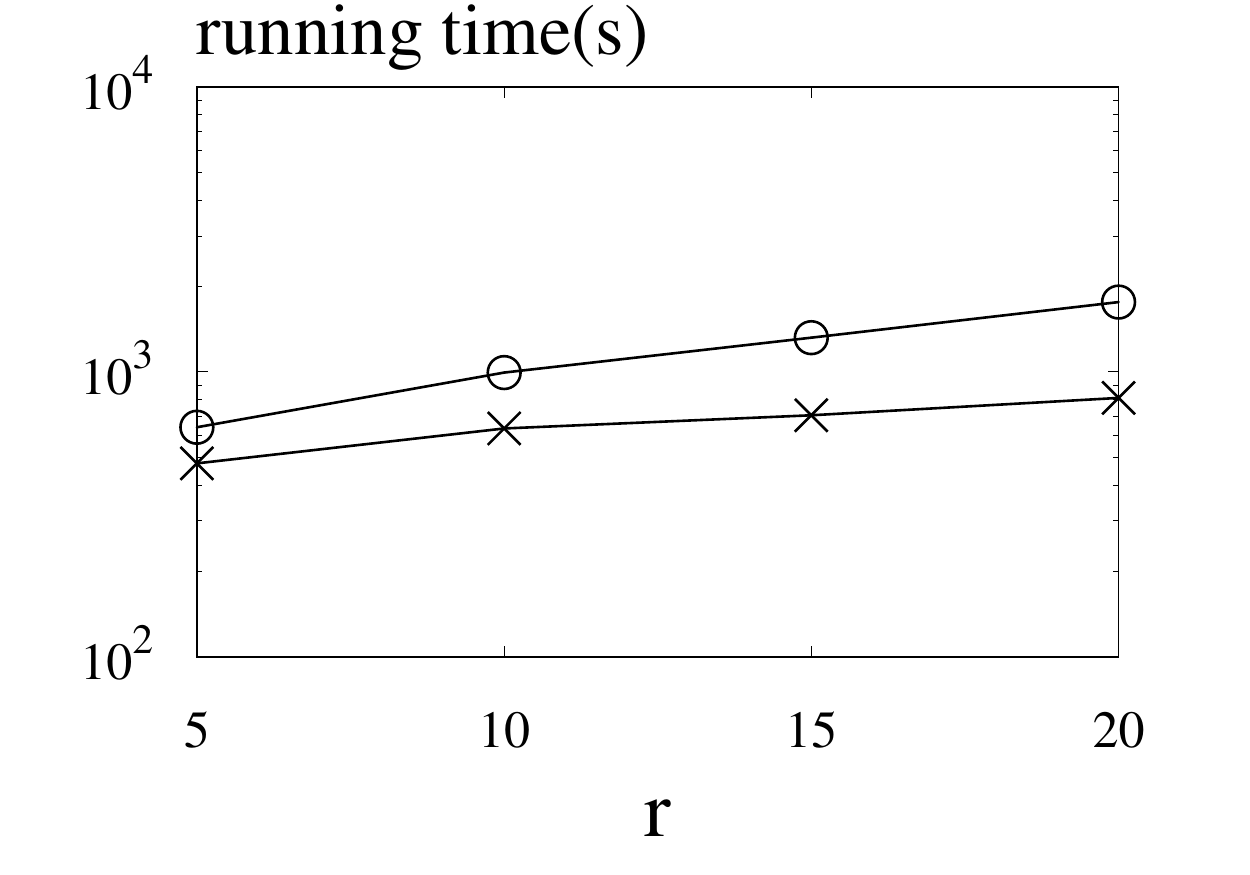}
			\\[-1mm]
			\hspace{-4mm} (a) Email &
			\hspace{-4mm} (b) DBLP &
			\hspace{-4mm} (c) Youtube &
			\hspace{-4mm} (d) Orkut &
			\hspace{-4mm} (e) Livejournal &
			\hspace{-4mm} (f) FriendSter \\[-1mm]
		\end{tabular}
		\vspace{-2mm}
		\caption{Running time vs. $r$ (sum, size-unconstrained)}
		\label{fig:time vs r sum}
		\vspace{-3mm}
	\end{small}
\end{figure*}

\noindent \textbf{Datasets.} We evaluate our experiments on $6$ real graphs: Email, DBLP, Youtube, Orkut, LiveJournal, and FriendSter. All datasets are downloaded from the Stanford Network Analysis Platform\footnote{http://snap.stanford.edu/}. The statistics of datasets are shown in Table~\ref{table:Dataset Table}. In Table~\ref{table:Dataset Table}, $k_{\max}$ indicates that graph does not contain a non-empty $(k_{\max} + 1)$-core. Moreover, the weight of vertices is the PageRank value of vertices with the damping factor being set as $0.85$. In order to demonstrate the effectiveness of the proposed algorithms, we illustrate the case study over Aminer dataset.

\noindent \textbf{Compared Algorithms.} To the best of our knowledge, there only exist algorithms for the top-$r$ $k$-influential community search problem when \fx $=$ $min$. The algorithm is simply extended to the case when the aggregation function is $max$. However, there exists no work studying the top-$r$ $k$-influential community search problem under other aggregation functions, e.g., $avg$, $sum$. Their algorithms were designed based on the nice property of $min$. Thus, if simply modify them to our problem, they would degrade to our na\"ive algorithm. We do not consider these algorithms in this paper. 

Due to the slowness of the exact algorithm, we implement and evaluate the following approaches.

\noindent \textit{Size-Unconstrained Problem}:
\begin{itemize}
			\item \textit{Na\"ive}: The solution to top-$r$ size-unconstrained $k$-influential community search problem for $sum$, which is proposed in Algorithm~\ref{alg:sum_naive}.
			\item \textit{Improve}: The method mentioned in Algorithm~\ref{alg:improve_framework}. We set $\epsilon$ equal to $0$.
			\item \textit{Approx}: The method mentioned in Algorithm~\ref{alg:improve_framework} when the $\epsilon > 0$.
\end{itemize}
\noindent \textit{Size-Constrained Problem}:
\begin{itemize}
			\item \textit{Random}: The method mentioned in Algorithm~\ref{alg:local_search} by using random strategy. The random approach does not sort the nearest neighbor set.
			\item \textit{Greedy}: The method mentioned in Algorithm~\ref{alg:local_search} by using greedy strategy. The greedy approach sorts the nearest neighbor set in decreasing order.
\end{itemize}

%

\noindent \textbf{Parameters.} There are four parameters that need to be considered: $k$, $r$, $\epsilon$, and $s$. We conduct experiments under various settings by varying four parameters. The value of $k$ varies from $4$ to $200$, $r$ varies from $\{5,10,15,20\}$, $\epsilon$ varies from $\{0.01, 0.05, 0.1, 0.2\}$, and $s$ varies from $\{5,10,15,20\}$. The default setting is $\epsilon = 0.1$, $r=5$, and $s=20$. However, since the na\"ive algorithm is inefficient, we will default to $k=4$ for small datasets, e.g., Email, DBLP, Youtube, and $k=40$ for large datasets, e.g., Orkut, LiveJournal, FriendSter. By default, we evaluate the performance of overlapping approaches. In the case study, we examine non-overlapping approaches to demonstrate the difference between our models and previous work.

\begin{figure*}[!t]
	\centering
	\begin{small}
		\begin{tabular}{cccccc}
			\hspace{-6mm} \includegraphics[height=22mm]{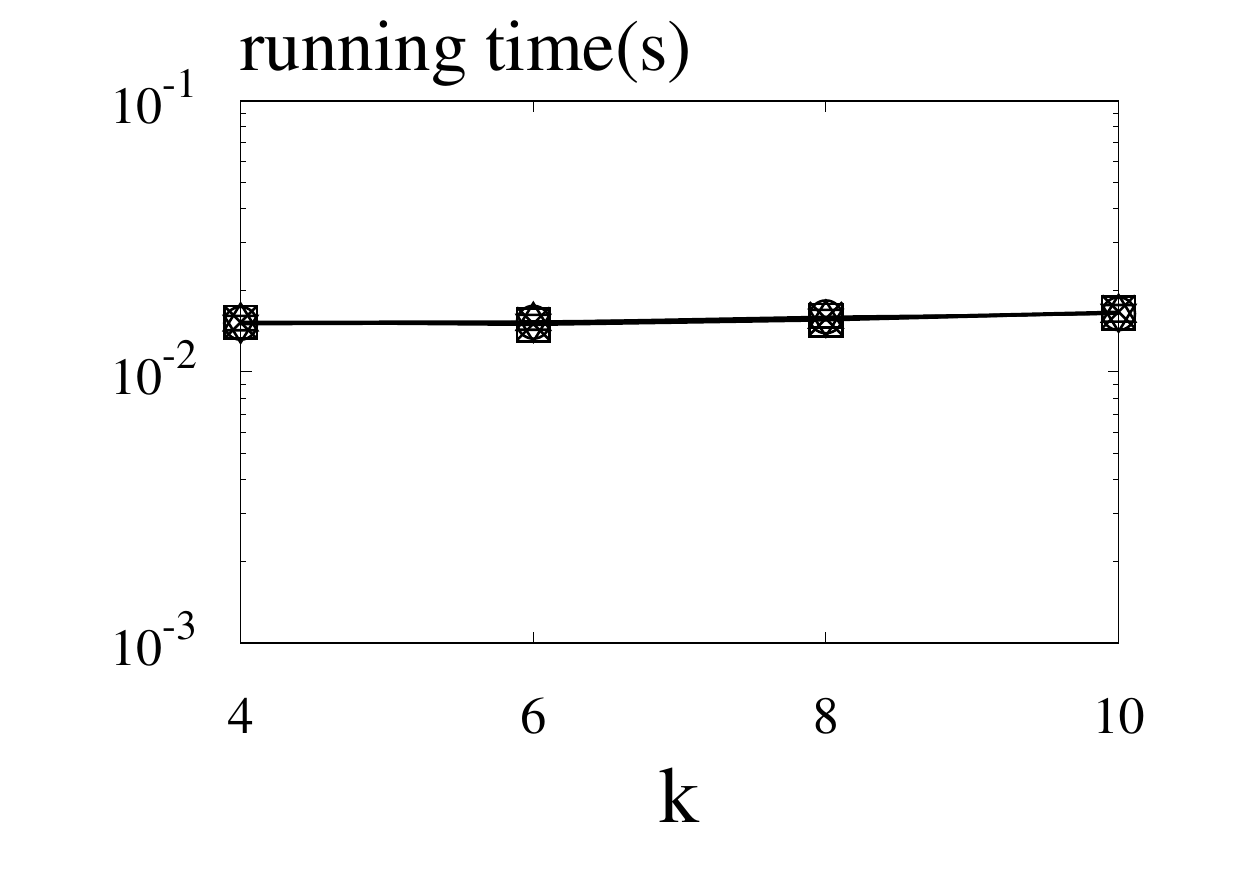} &
			\hspace{-6mm} \includegraphics[height=22mm]{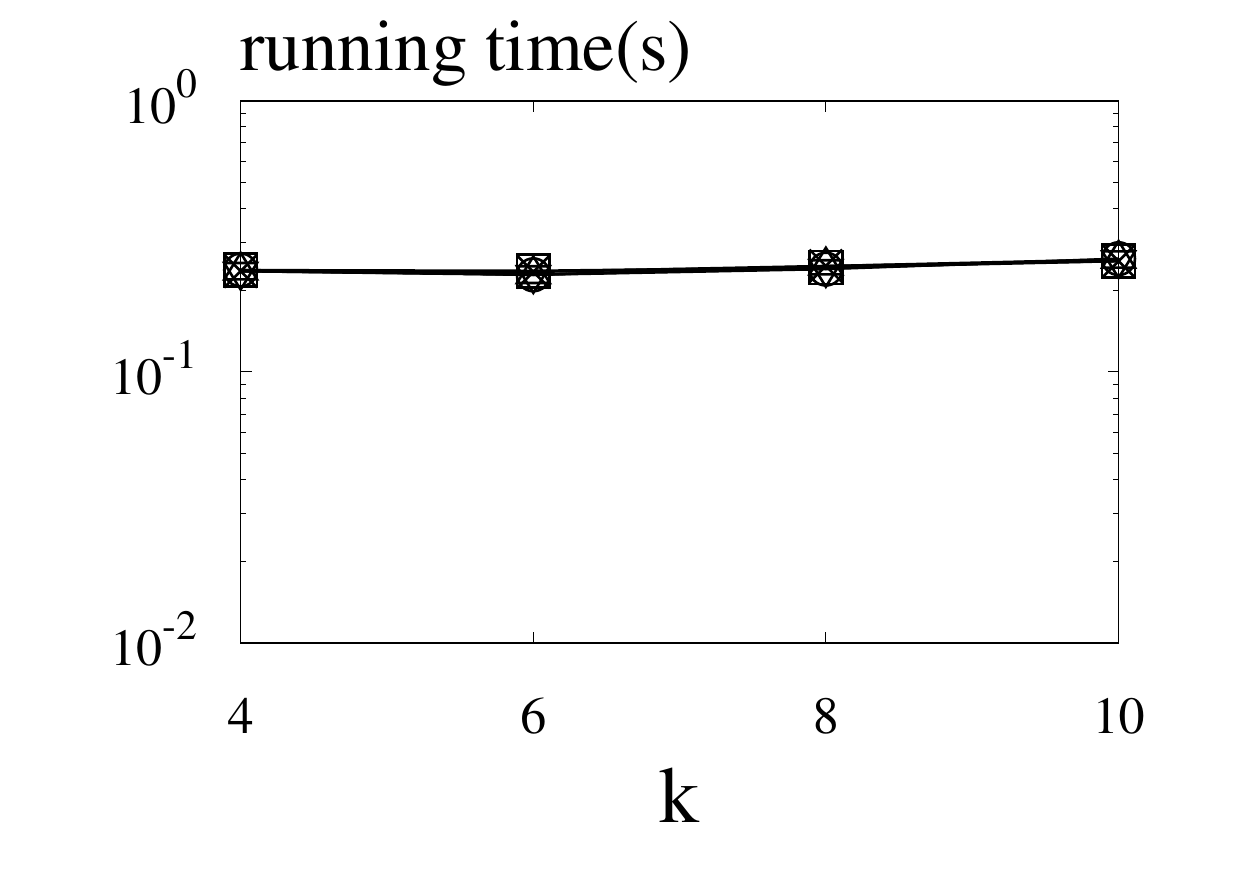} &
			\hspace{-6mm} \includegraphics[height=22mm]{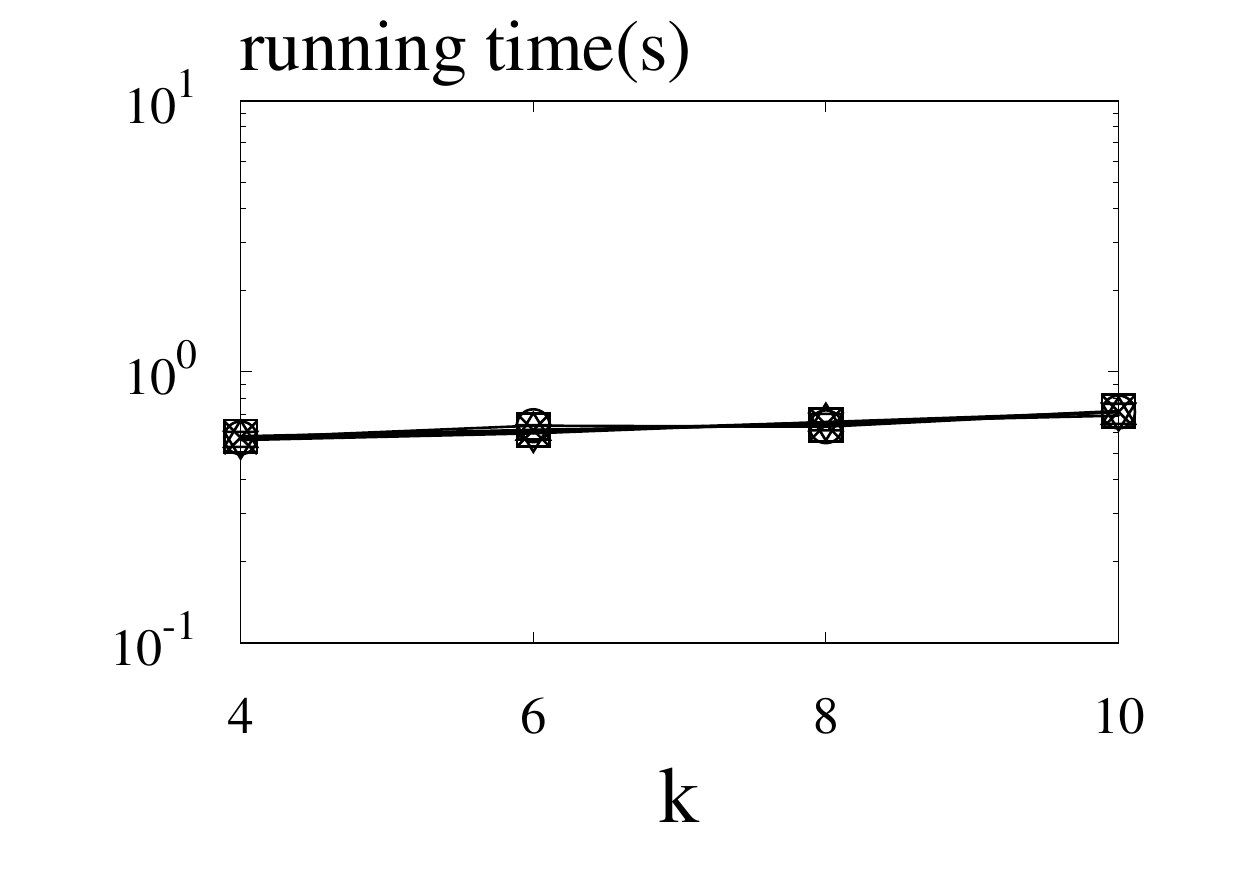} &
			\hspace{-6mm} \includegraphics[height=22mm]{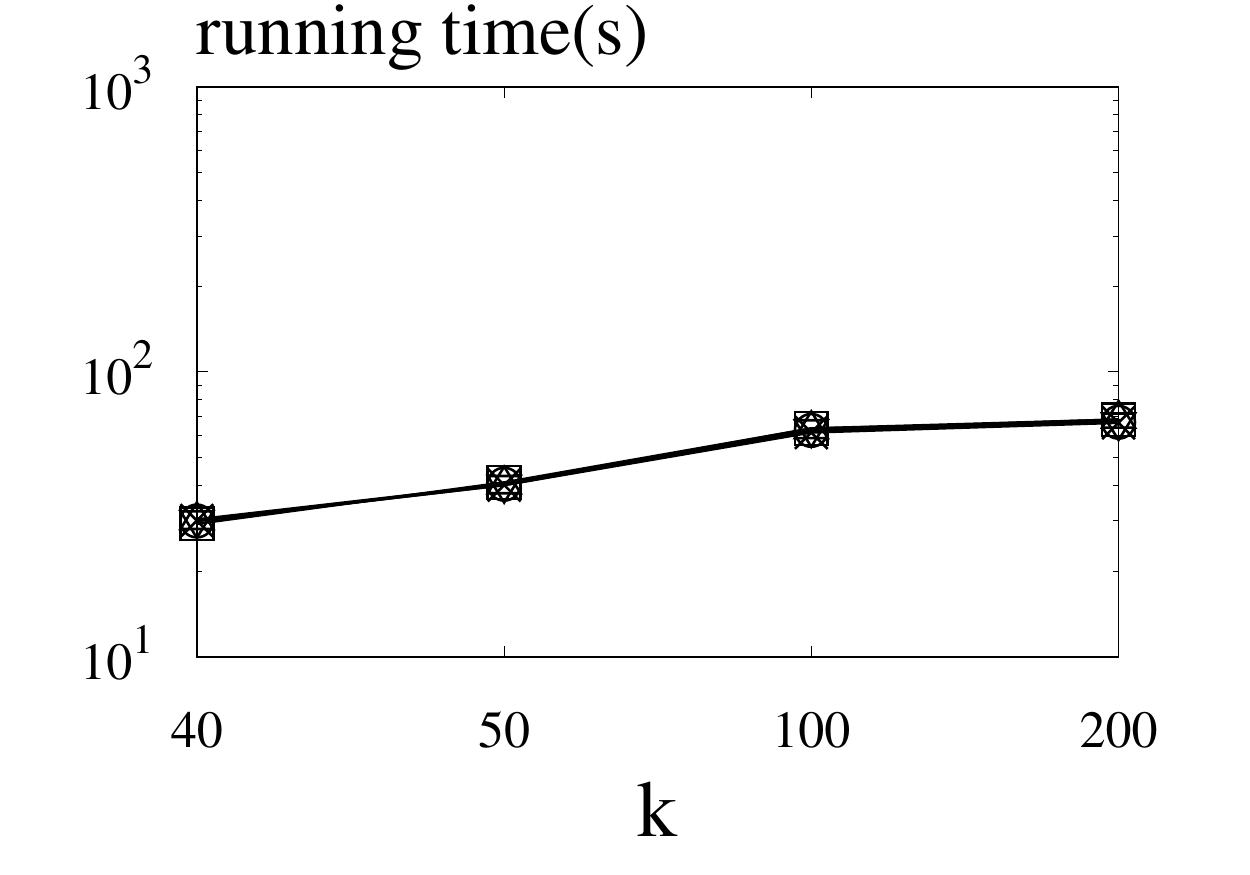} &
			\hspace{-6mm} \includegraphics[height=22mm]{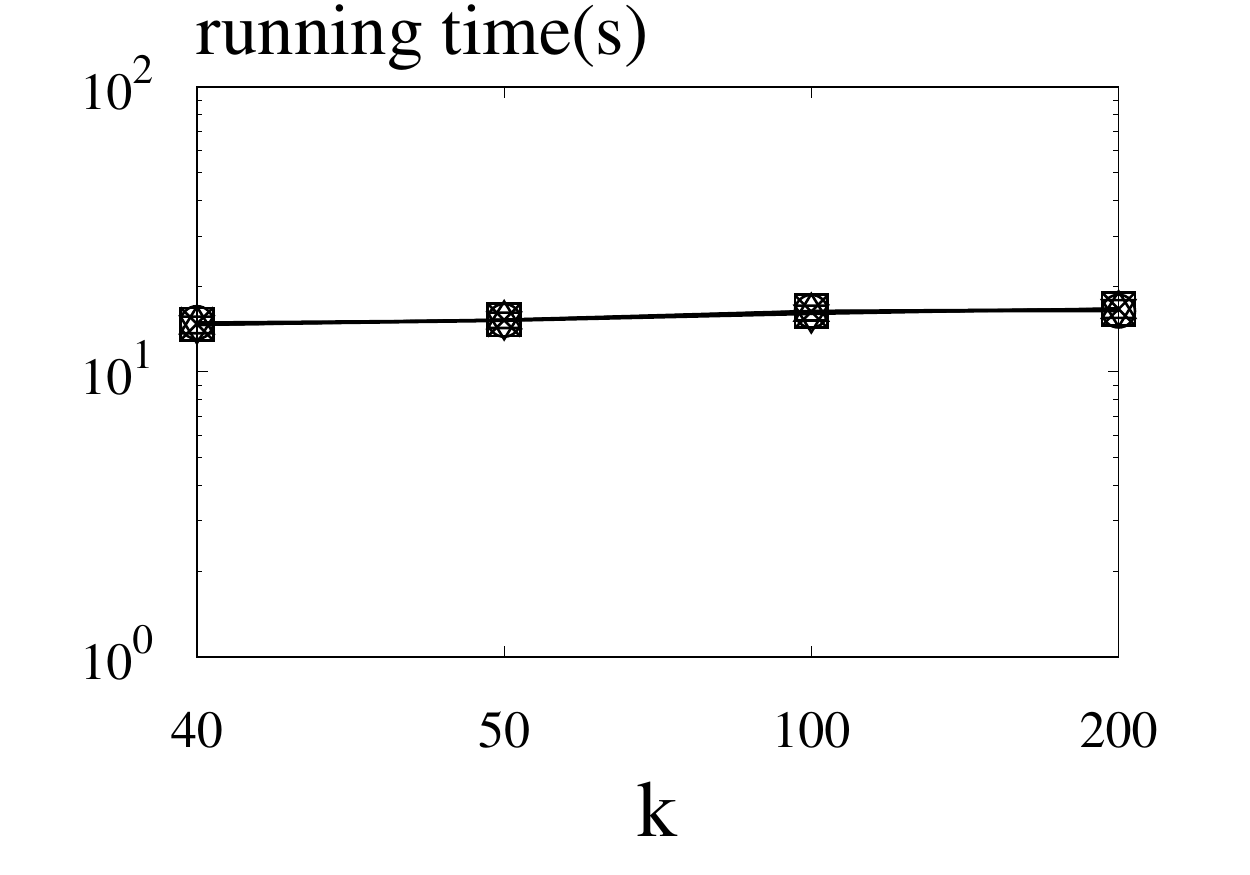} &
			\hspace{-6mm} \includegraphics[height=22mm]{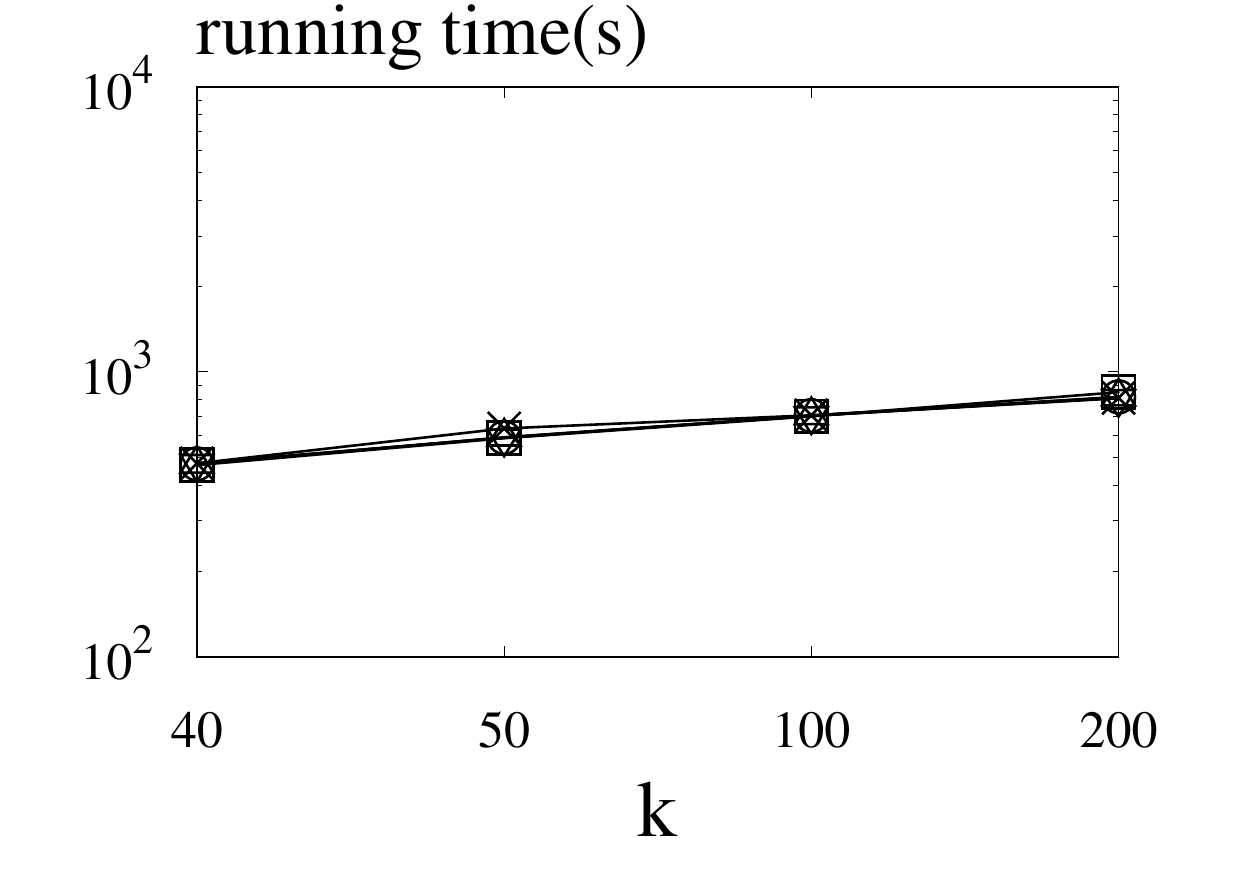}
			\\[-1mm]
			\hspace{-4mm} (a) Email &
			\hspace{-4mm} (b) DBLP &
			\hspace{-4mm} (c) Youtube &
			\hspace{-4mm} (d) Orkut &
			\hspace{-4mm} (e) Livejournal &
			\hspace{-4mm} (f) FriendSter \\[-1mm]
		\end{tabular}
		\vspace{-2mm}
		\caption{Running time vs. $k$ ($\epsilon$, sum, size-unconstrained)}
		\label{fig:time vs k eps sum}
		\vspace{-5mm}
	\end{small}
\end{figure*}

\begin{figure*}[!t]
	\centering
	\vspace{-1mm}
	\begin{small}
		\begin{tabular}{cccccc}
			\multicolumn{6}{c}{\hspace{-6mm} \includegraphics[height=10mm]{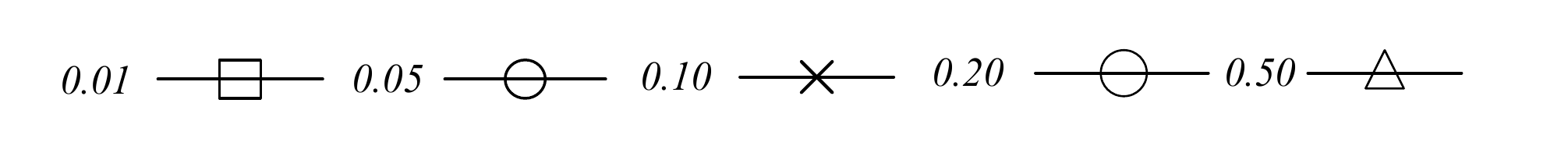}}  \\[-3mm]
			\hspace{-6mm} \includegraphics[height=22mm]{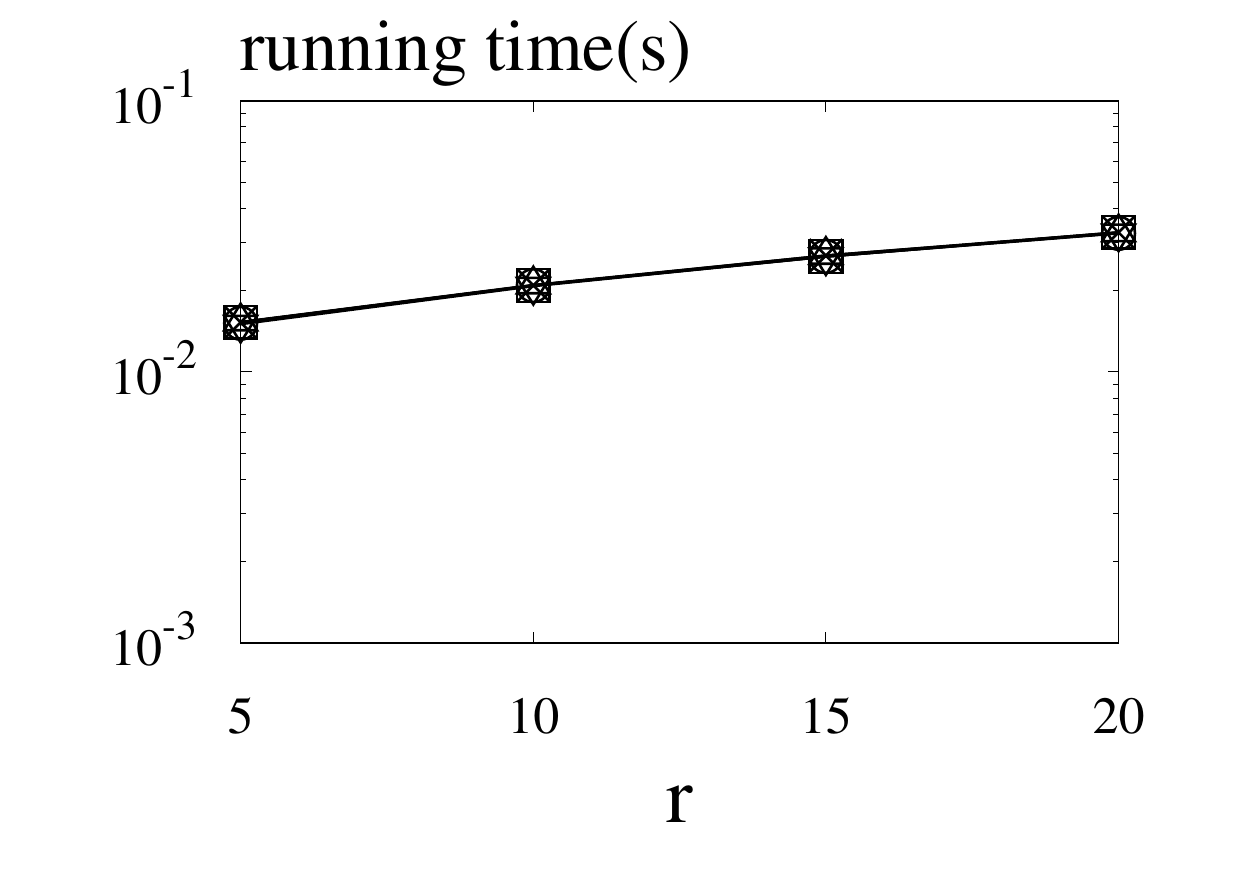} &
			\hspace{-6mm} \includegraphics[height=22mm]{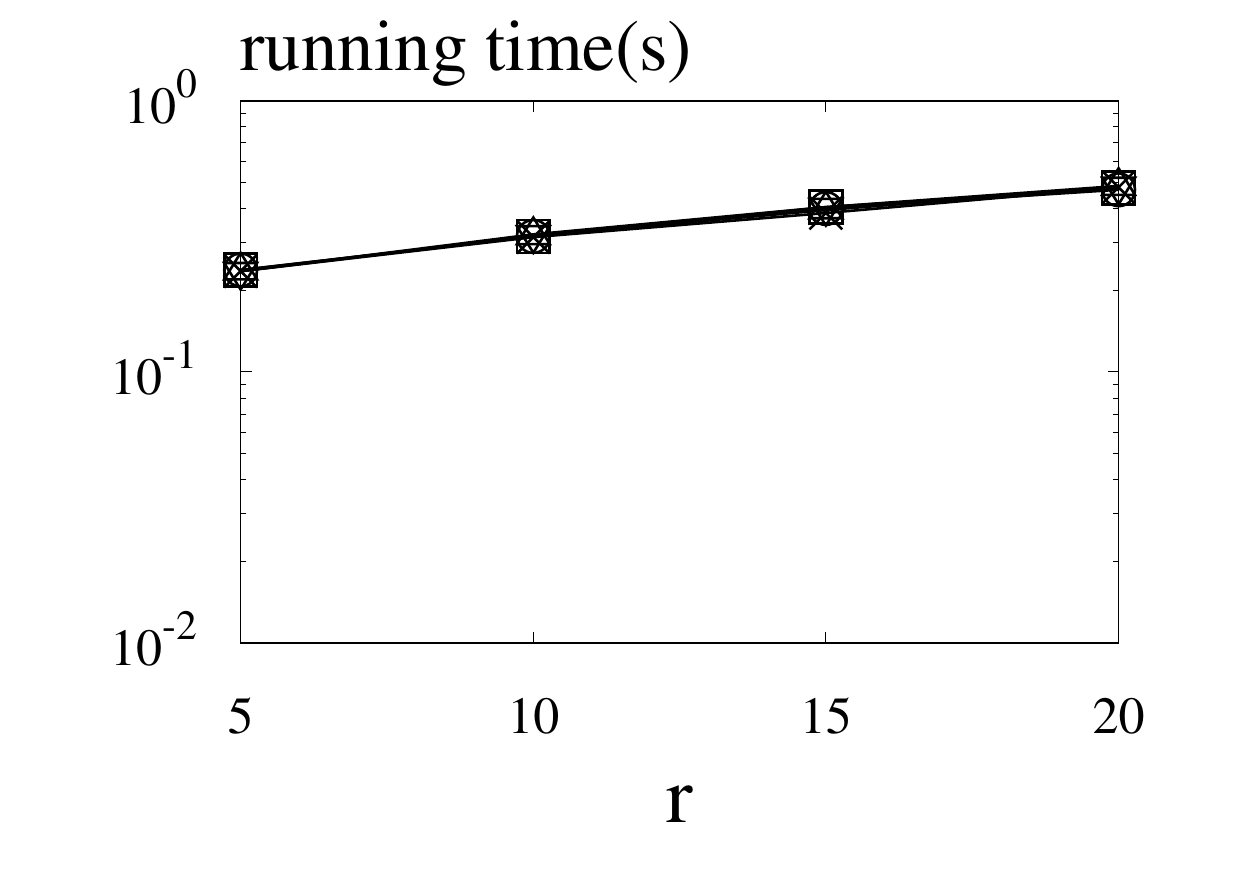} &
			\hspace{-6mm} \includegraphics[height=22mm]{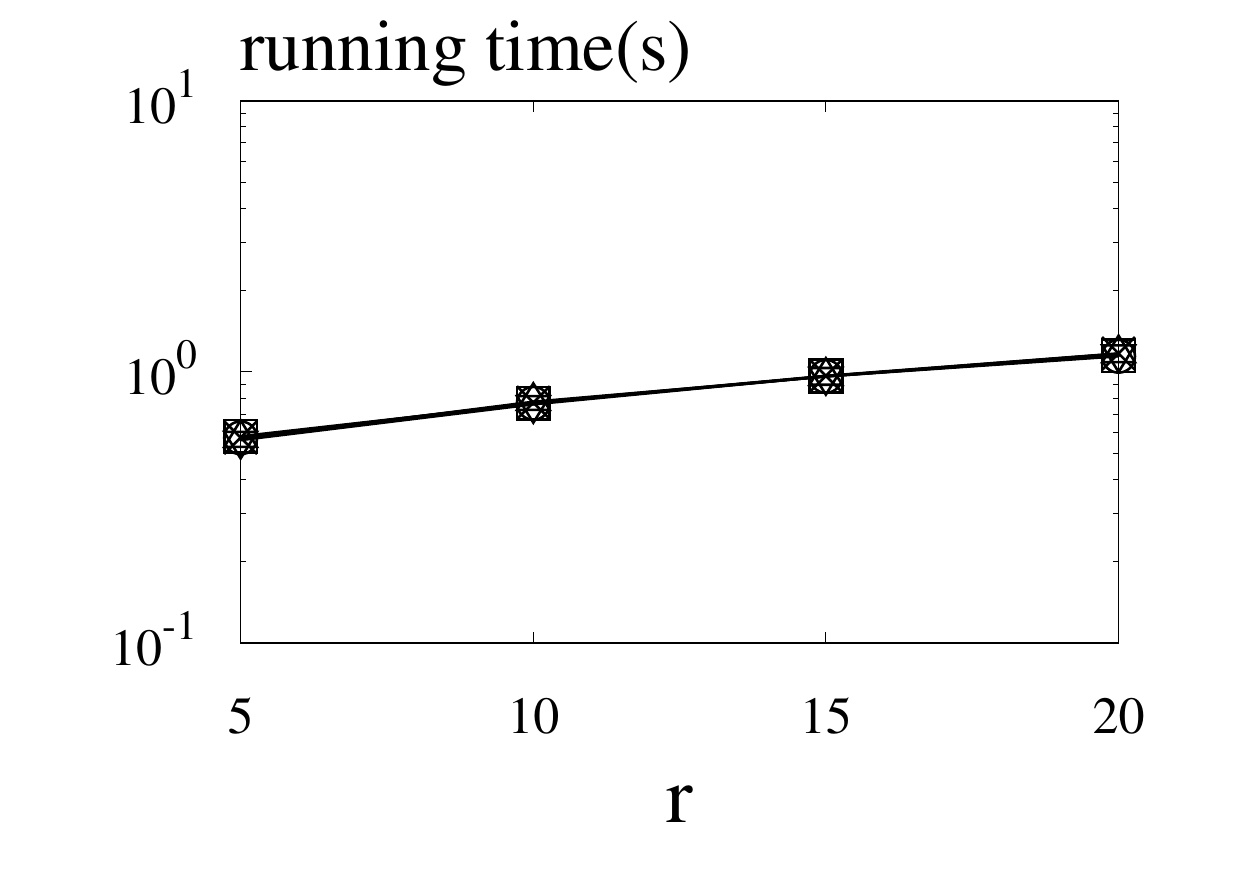} &
			\hspace{-6mm} \includegraphics[height=22mm]{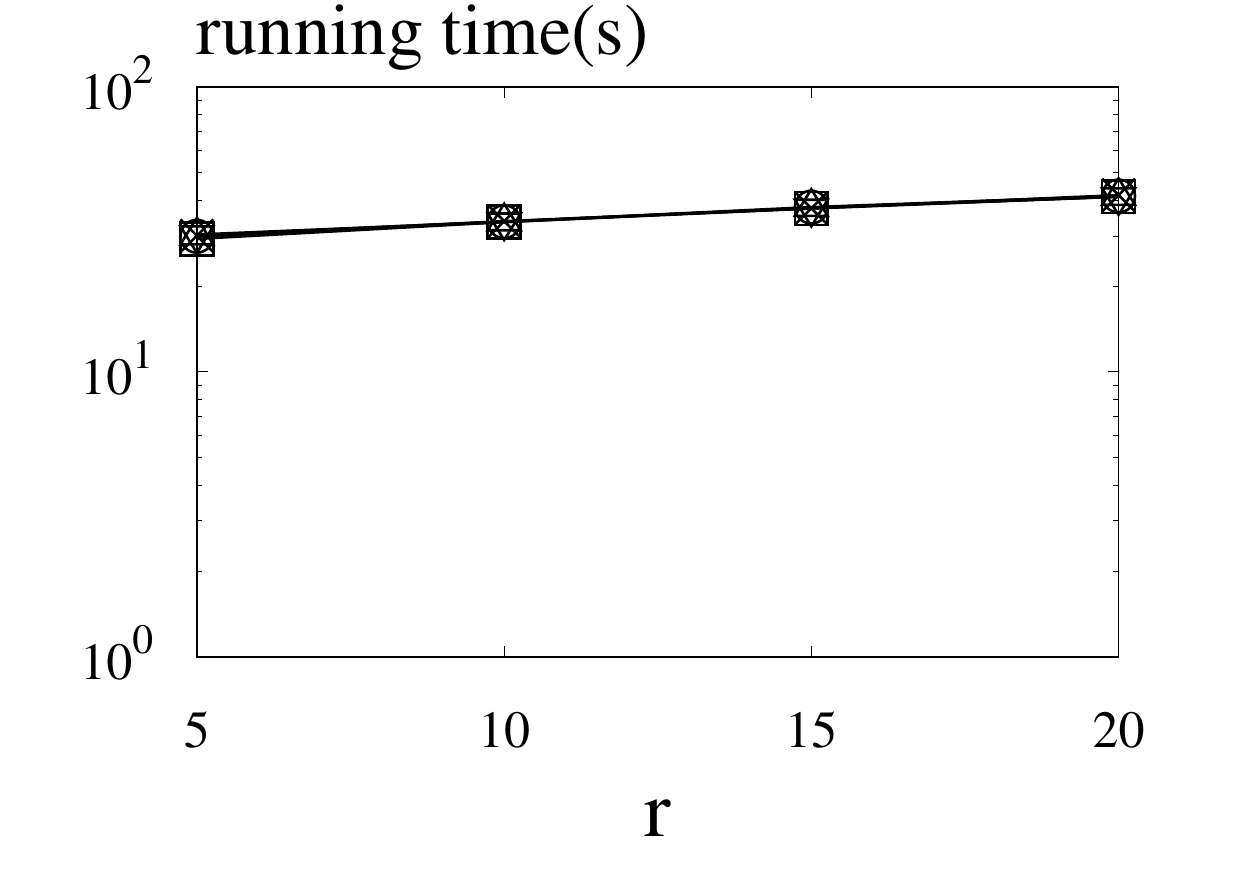} &
			\hspace{-6mm} \includegraphics[height=22mm]{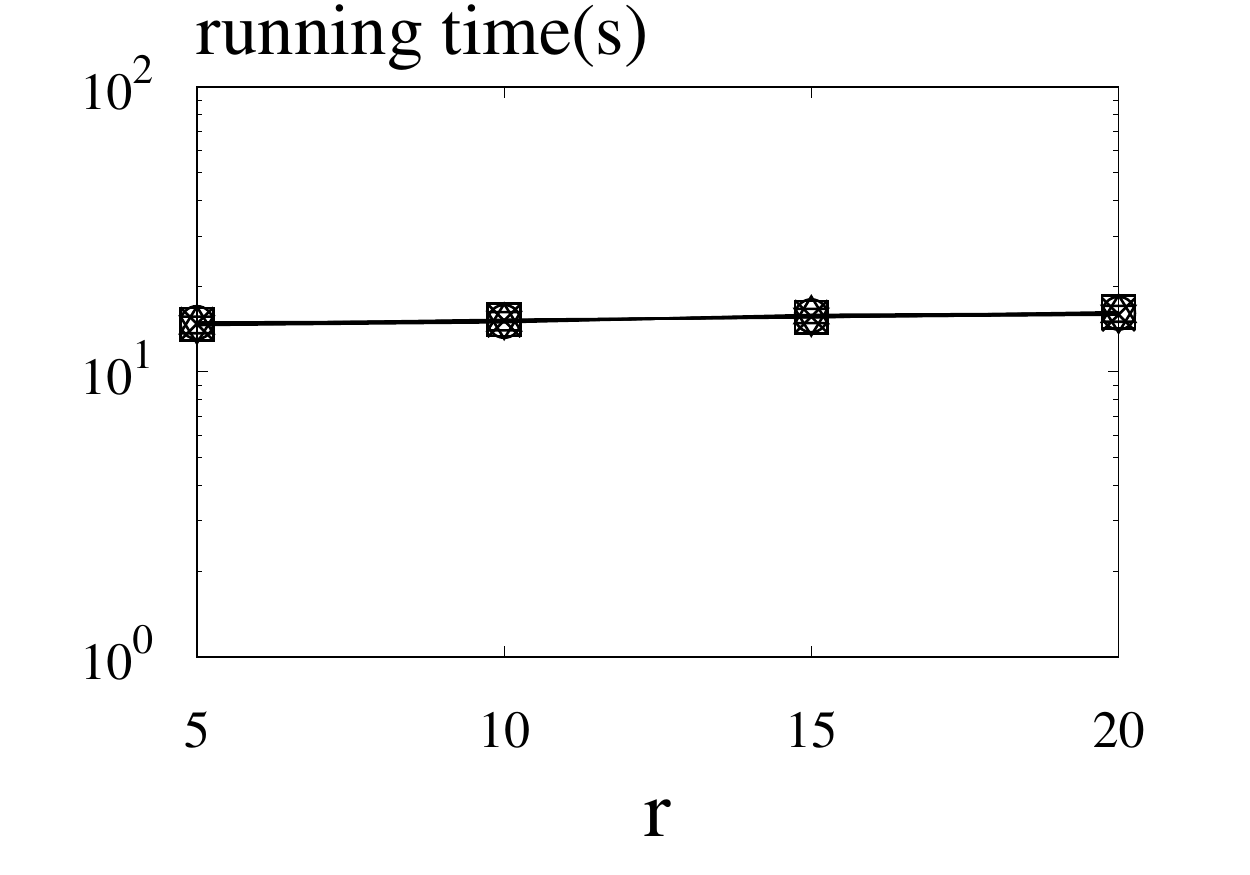} &
			\hspace{-6mm} \includegraphics[height=22mm]{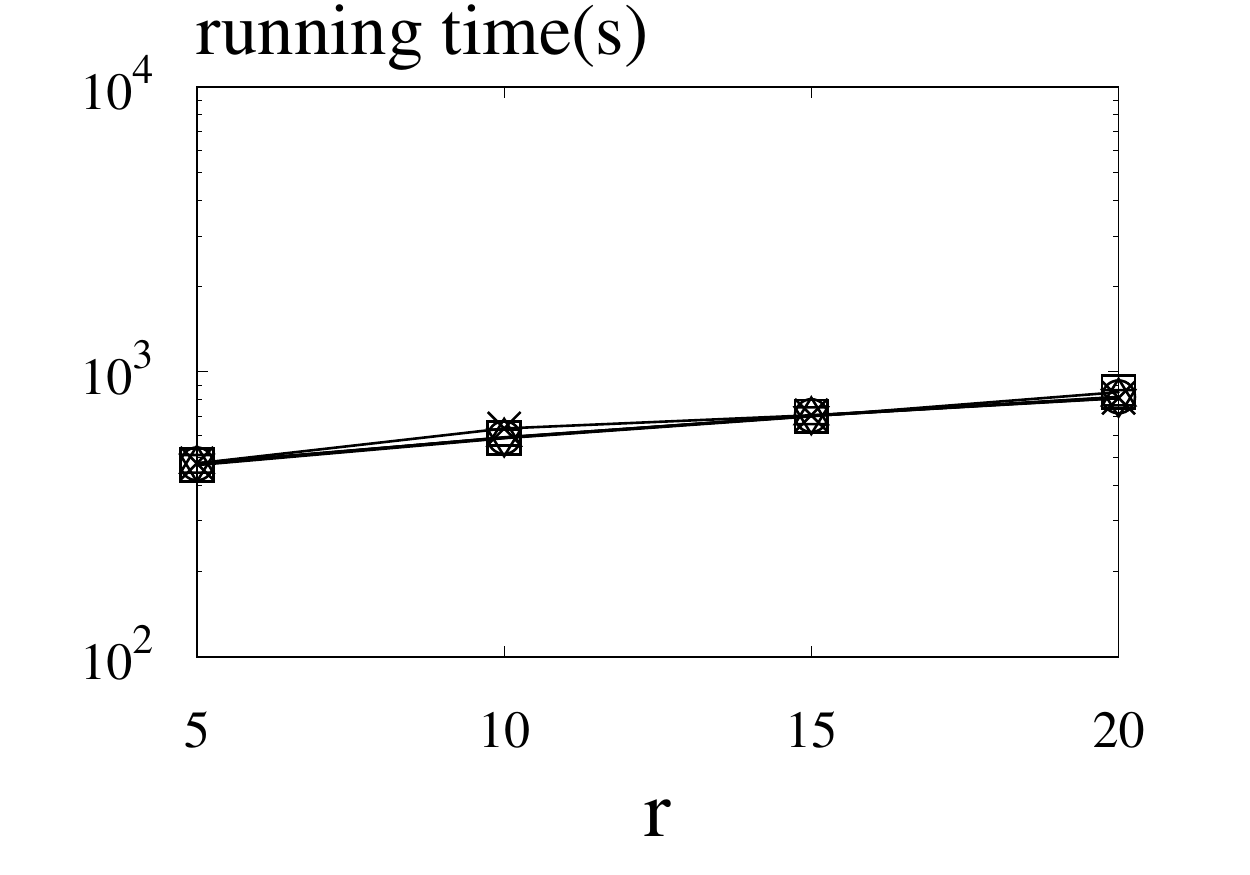}
			\\[-1mm]
			\hspace{-4mm} (a) Email &
			\hspace{-4mm} (b) DBLP &
			\hspace{-4mm} (c) Youtube &
			\hspace{-4mm} (d) Orkut &
			\hspace{-4mm} (e) Livejournal &
			\hspace{-4mm} (f) FriendSter \\[-1mm]
		\end{tabular}
		\vspace{-2mm}
		\caption{Running time vs. $r$ ($\epsilon$, sum, size-unconstrained)}
		\label{fig:time vs r eps sum}
		\vspace{-3mm}
	\end{small}
\end{figure*}
\subsection{Approaches to Size-Unconstrained Problems}

\noindent \textbf{Exp-I: Effect of $k$.}~Figure~\ref{fig:time vs k sum} evaluates the performance of three methods by fixing $r=5$, and varying $k$ (missing point indicates the algorithm cannot terminate in one day). What is striking in this figure is that as $k$ increases, the running time of the na\"ive algorithm decreases.

The main reason is that because the graph size decreases after pruning, and the running time of the core-decomposition algorithm is a tiny fraction of the overall time of the na\"ive algorithm. Nevertheless, the running time of the improved algorithm and approximated algorithms both grow as $k$ increases, and the running time of two algorithms is comparable. The reason for this is that the core-decomposition consumes a significant part of running time.

\noindent \textbf{Exp-II: Effect of $r$.} The effect of $r$ is evaluated in this experiment. As demonstrated in Figure~\ref{fig:time vs r sum}, when $r$ increases, the running time of algorithms increases. The reason is that it needs to output more $k$-influential communities and the iterations of the algorithm increases.

\noindent \textbf{Exp-III: Impact of $\epsilon$.} We compare the approximated algorithms to the top-$r$ $k$-influential community search problem under different aggregation functions. $\epsilon$ ranges from $\{0.01, 0.05, 0.1, 0.2, 0.5\}$. The result is demonstrated in Figure~\ref{fig:time vs k eps sum} and Figure~\ref{fig:time vs r eps sum}. What stands out in these figures is the running time of approximated algorithm remains almost unaltered by varying $\epsilon$. As a result, the approximated algorithm is insensitive to $\epsilon$. The reason is that the top-$r$ $k$-influential community is always computed within $r$ iterations at the beginning.
\subsection{Approaches to Size-constrained Problems}
\begin{figure*}[!t]
	\centering
	\vspace{-1mm}
	\begin{small}
		\begin{tabular}{cccccc}
			\hspace{-6mm} \includegraphics[height=22mm]{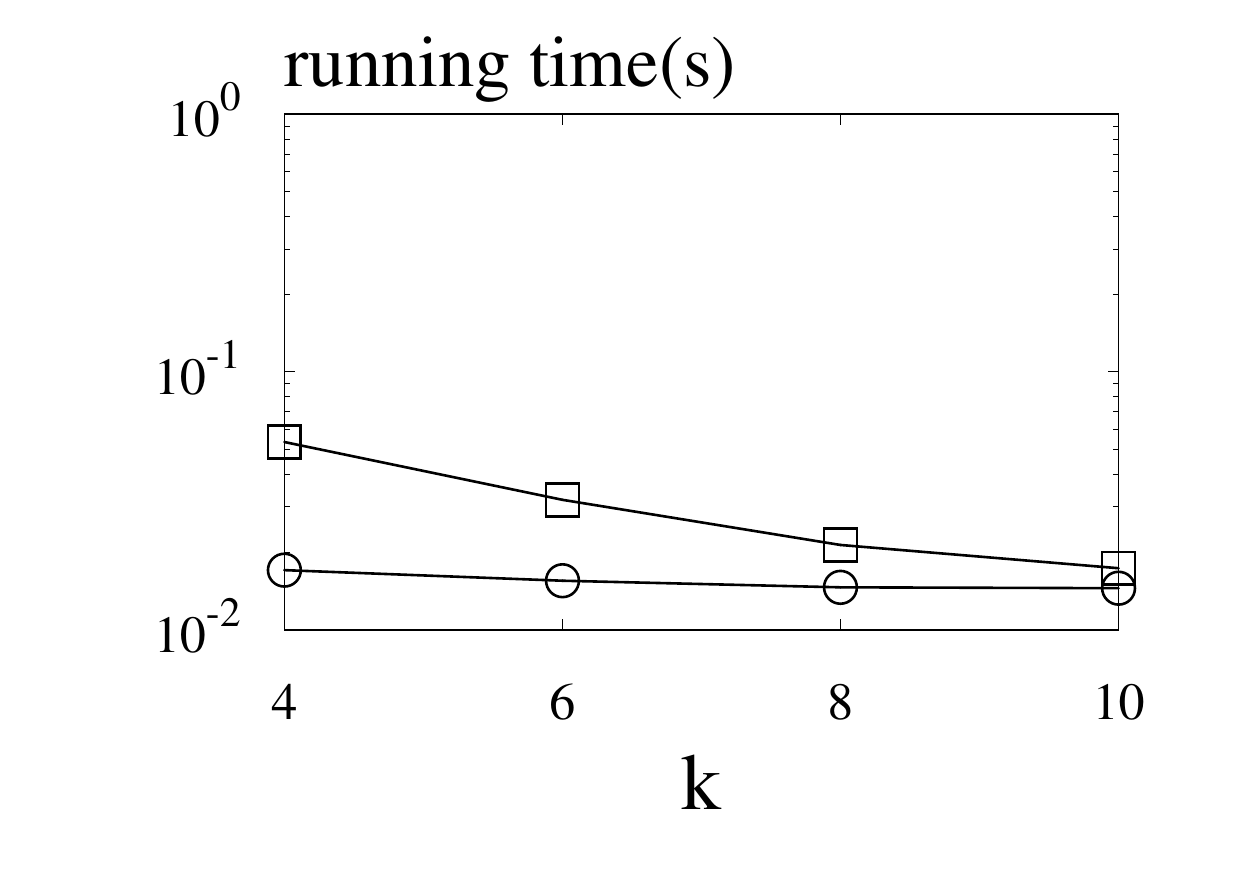} &
			\hspace{-6mm} \includegraphics[height=22mm]{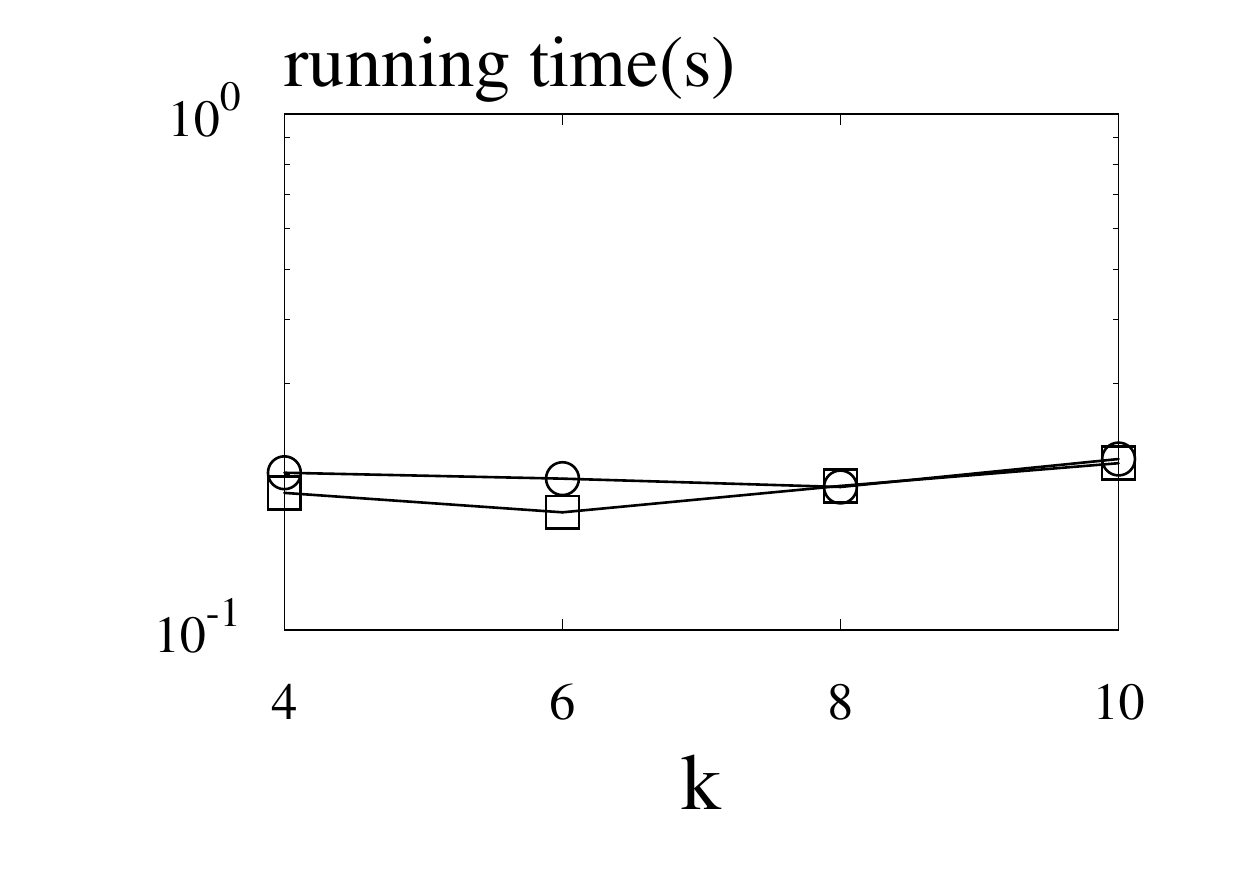} &
			\hspace{-6mm} \includegraphics[height=22mm]{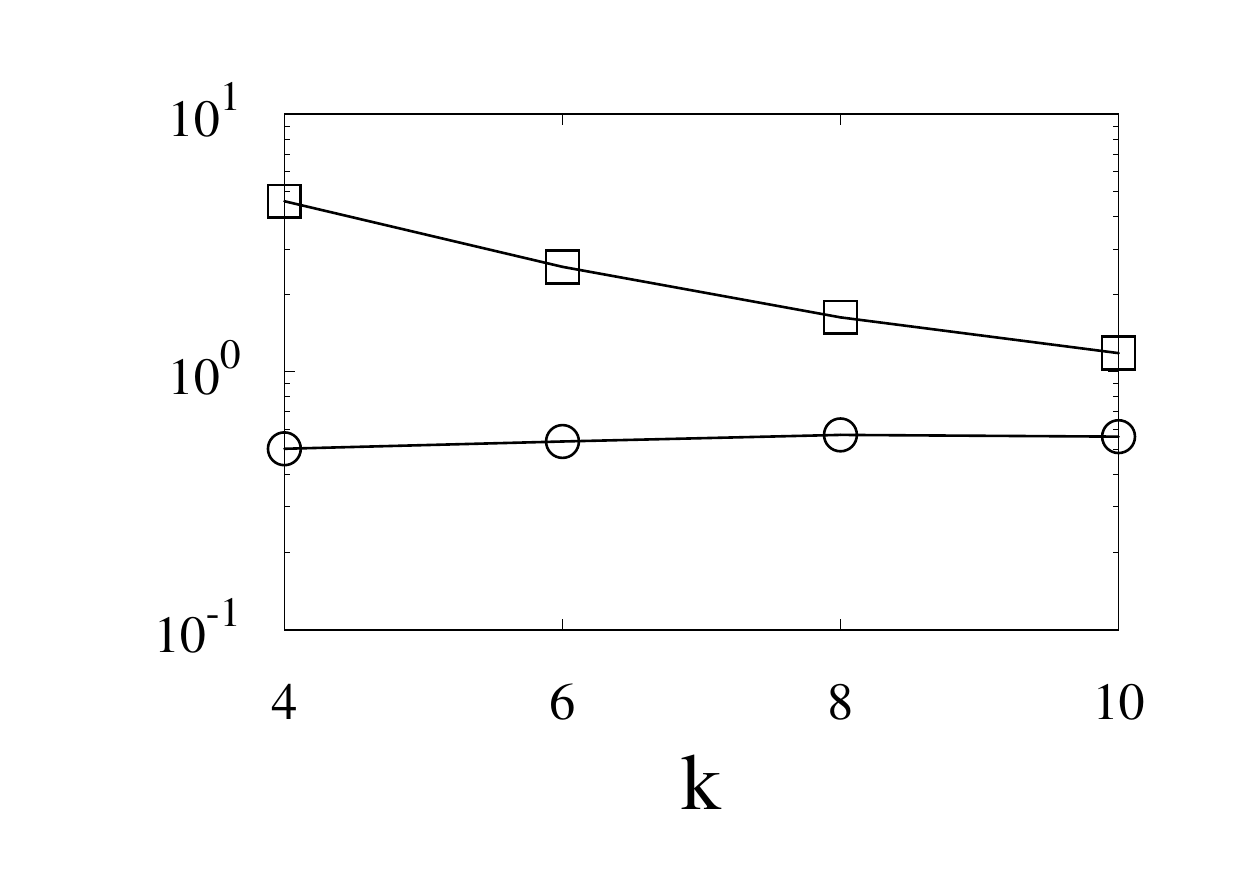} &
			\hspace{-6mm} \includegraphics[height=22mm]{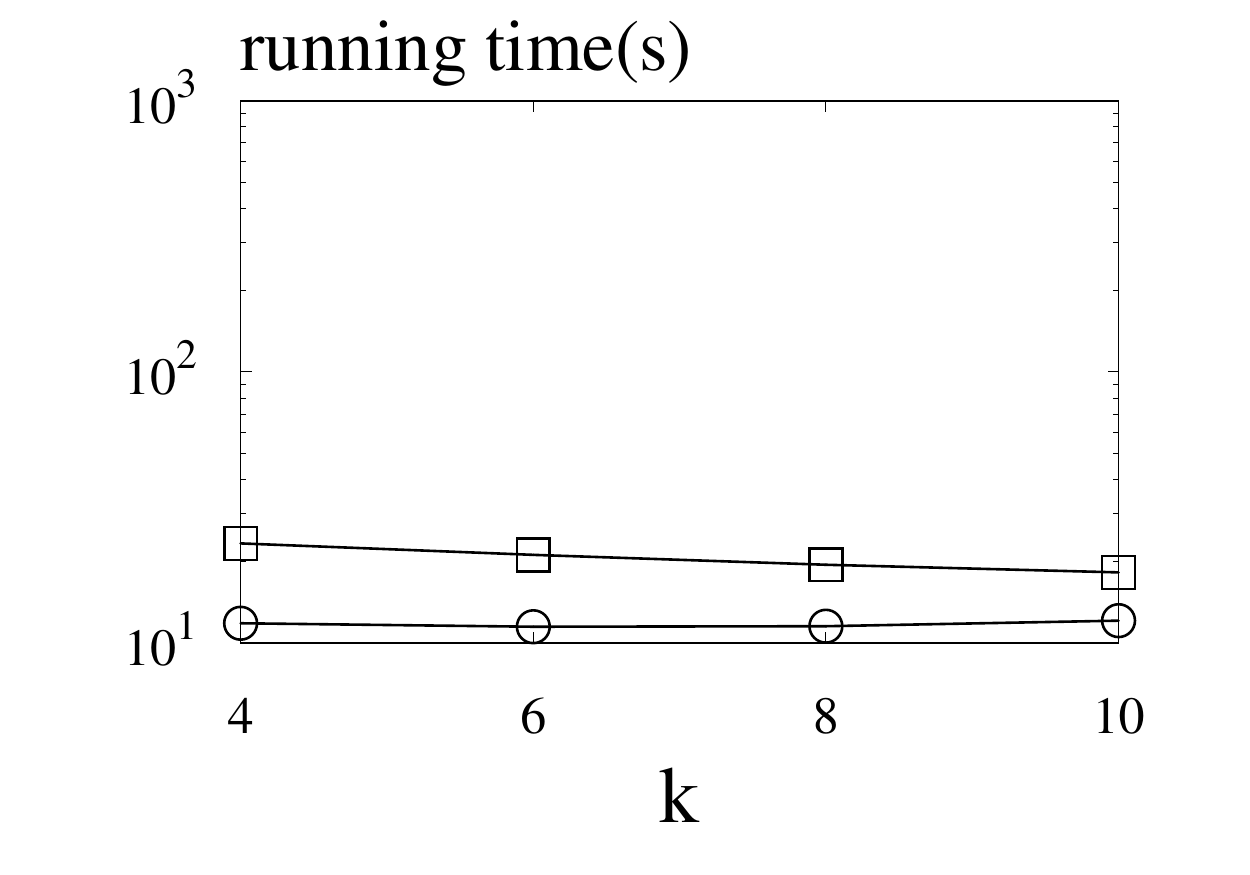} &
			\hspace{-6mm} \includegraphics[height=22mm]{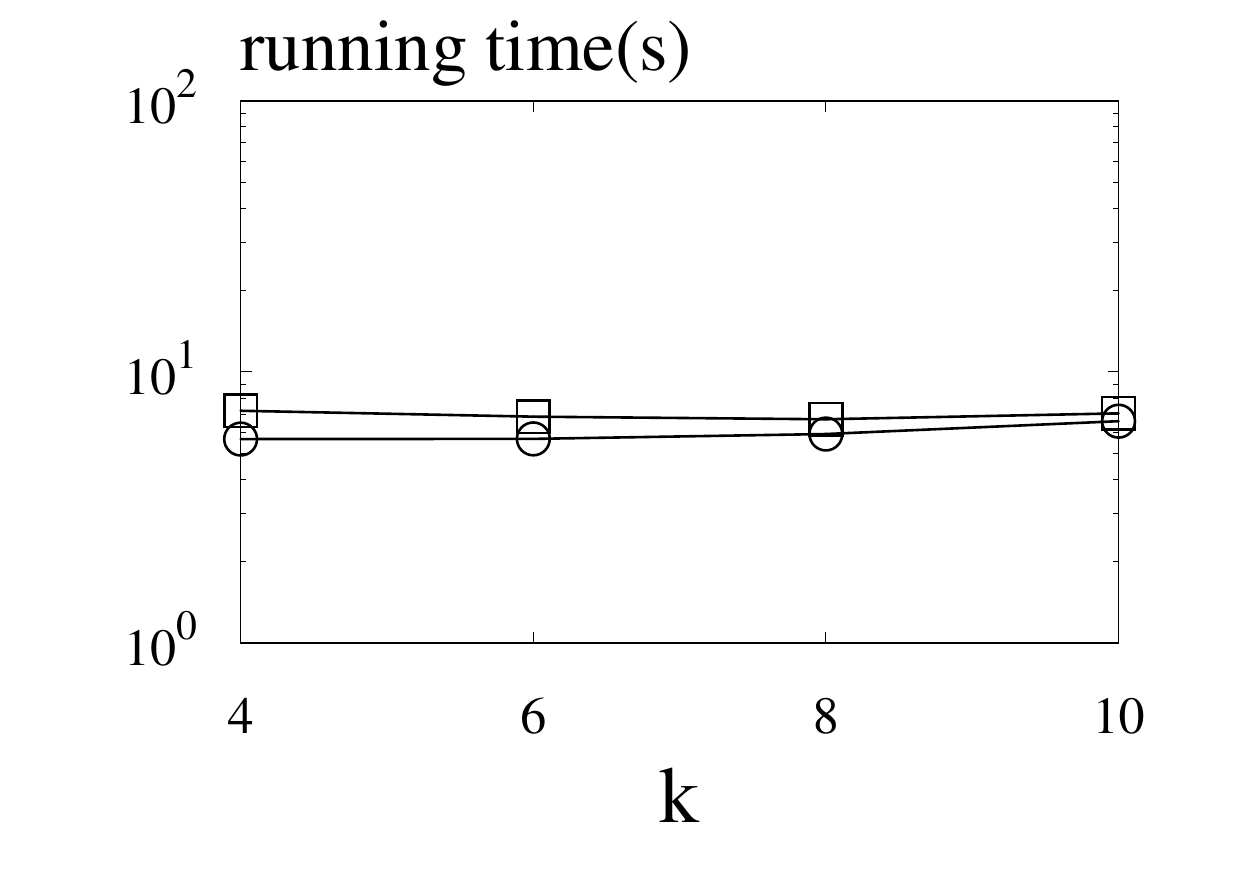} &
			\hspace{-6mm} \includegraphics[height=22mm]{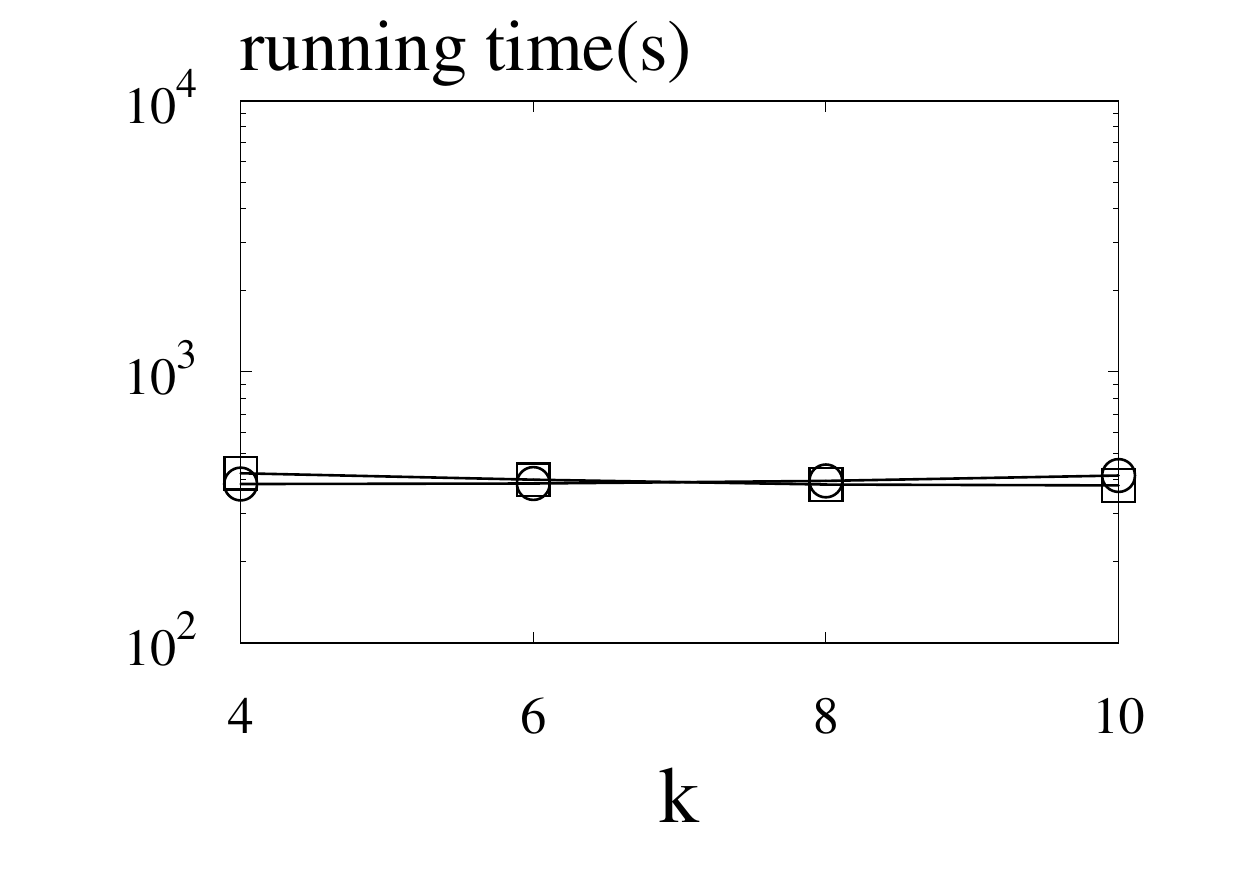}
			\\[-1mm]
			\hspace{-4mm} (a) Email &
			\hspace{-4mm} (b) DBLP &
			\hspace{-4mm} (c) Youtube &
			\hspace{-4mm} (d) Orkut &
			\hspace{-4mm} (e) Livejournal &
			\hspace{-4mm} (f) FriendSter \\[-1mm]
		\end{tabular}
		\vspace{-2mm}
		\caption{Running time vs. $k$ (sum, size-constrained)}
		\label{fig:time vs k size sum}
		\vspace{-5mm}
	\end{small}
\end{figure*}

\begin{figure*}[!t]
	\centering
	\vspace{-1mm}
	\begin{small}
		\begin{tabular}{cccccc}
			\multicolumn{6}{c}{\hspace{-6mm} \includegraphics[height=10mm]{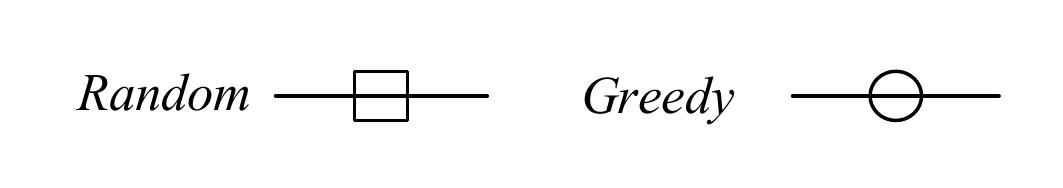}}  \\[-3mm]
			\hspace{-6mm} \includegraphics[height=22mm]{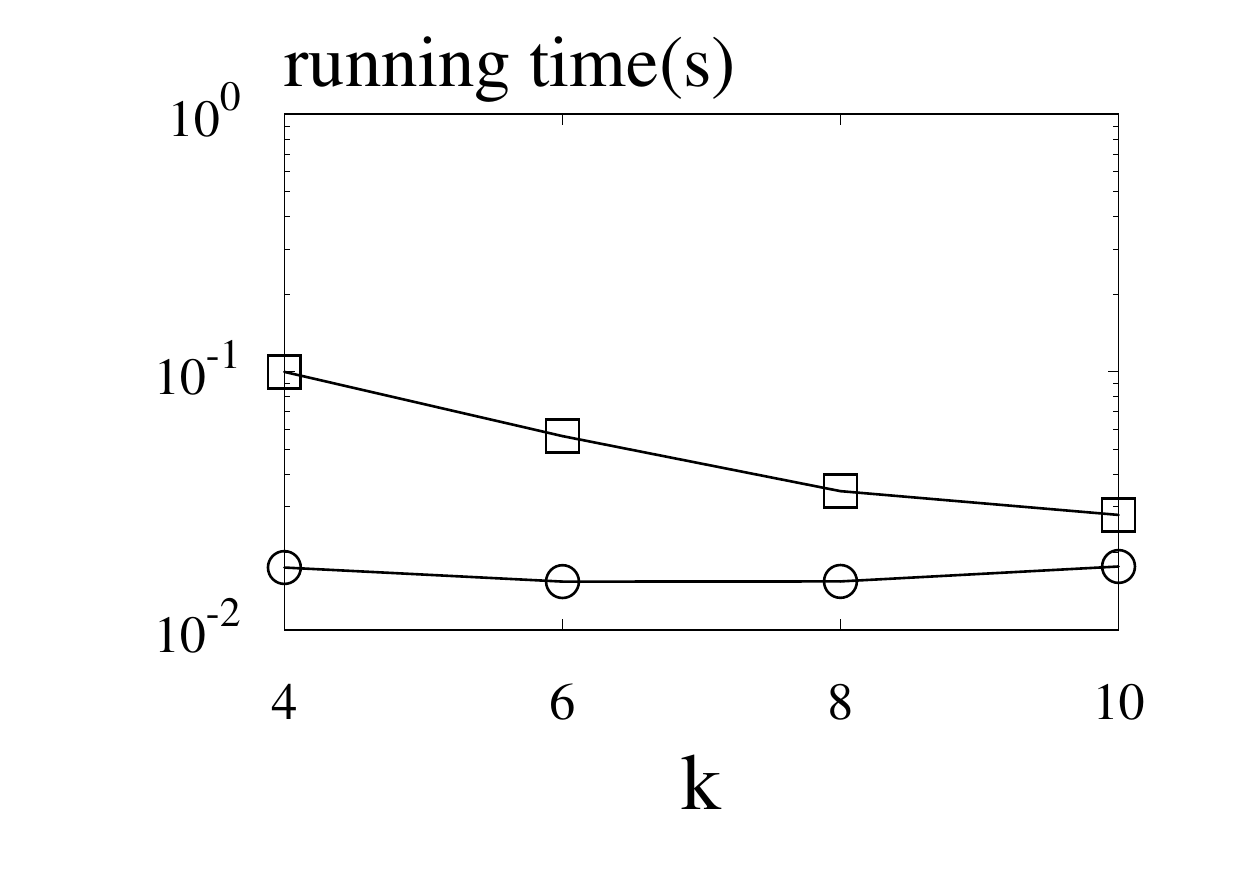} &
			\hspace{-6mm} \includegraphics[height=22mm]{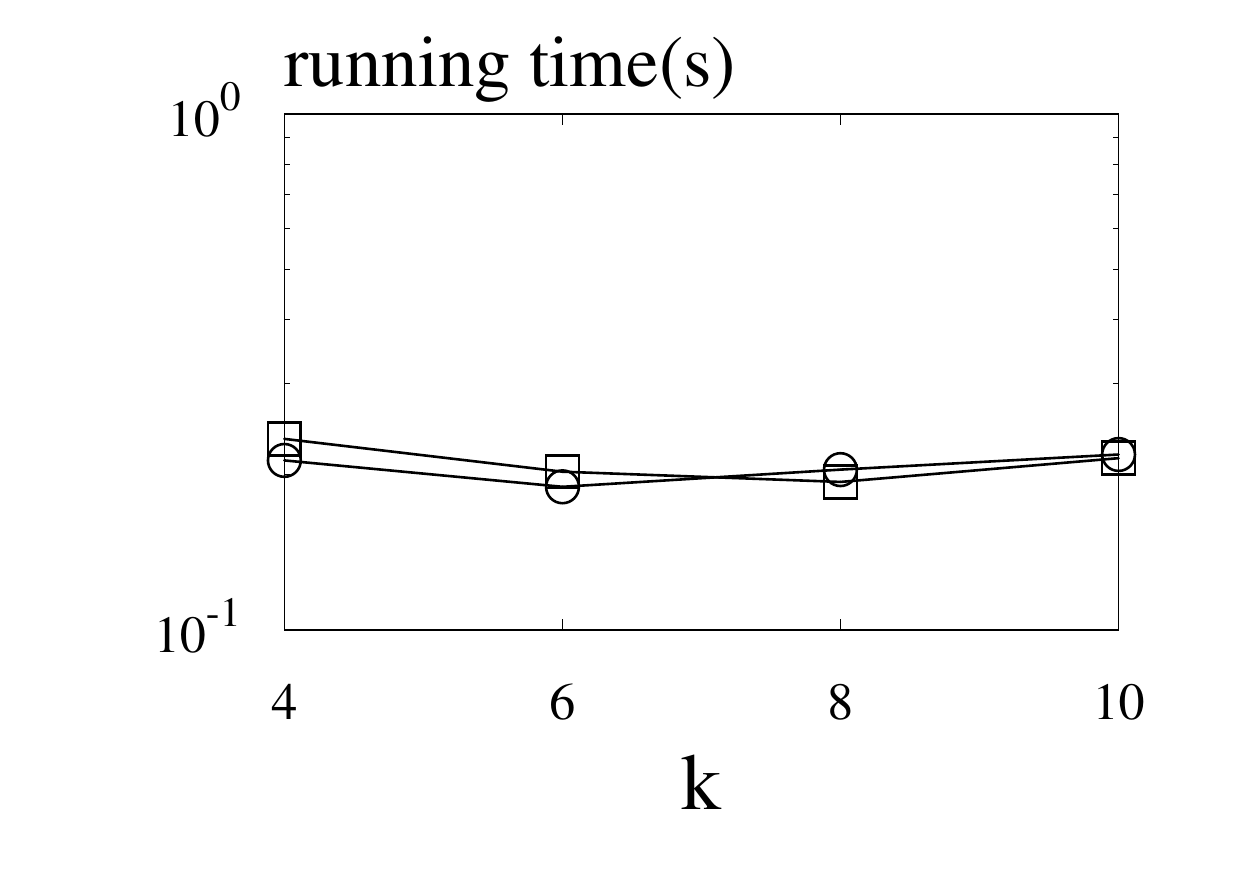} &
			\hspace{-6mm} \includegraphics[height=22mm]{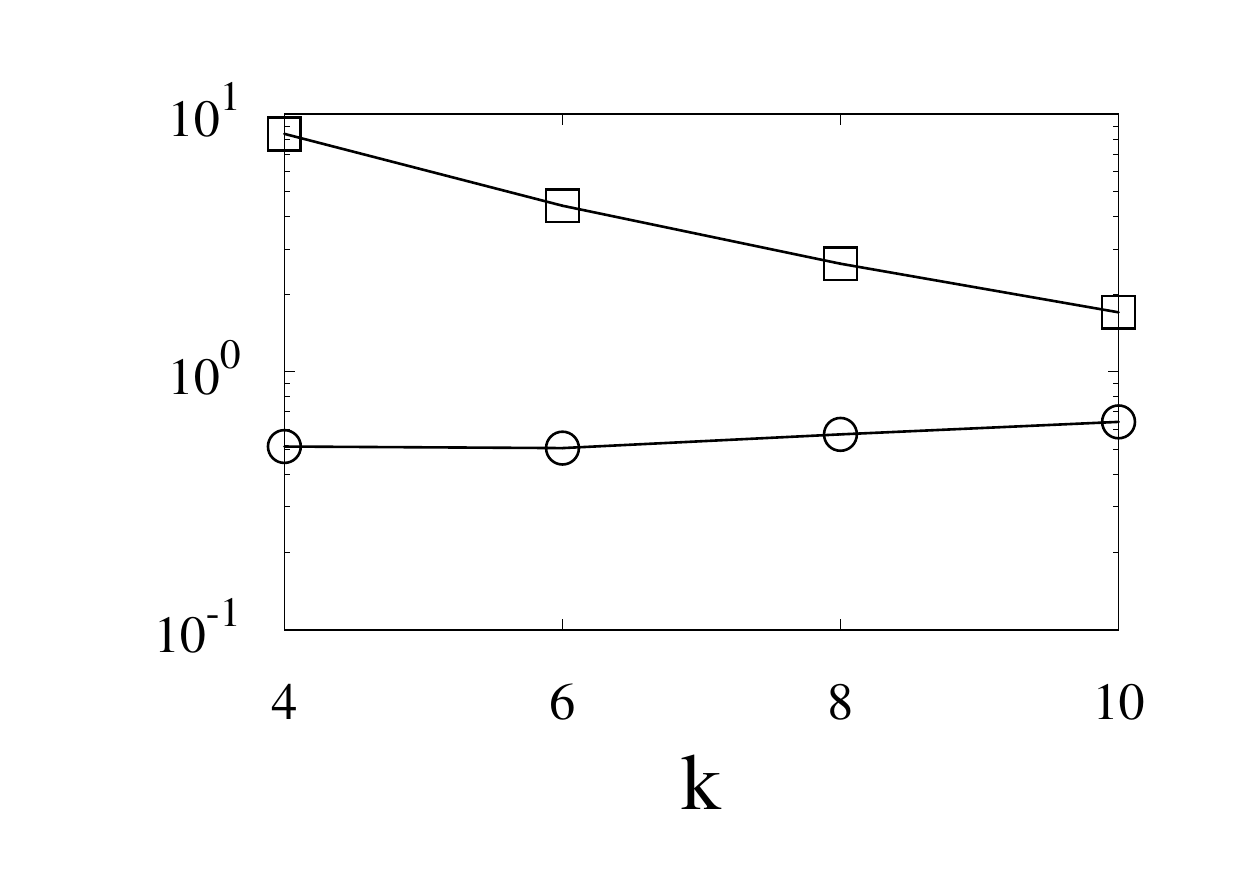} &
			\hspace{-6mm} \includegraphics[height=22mm]{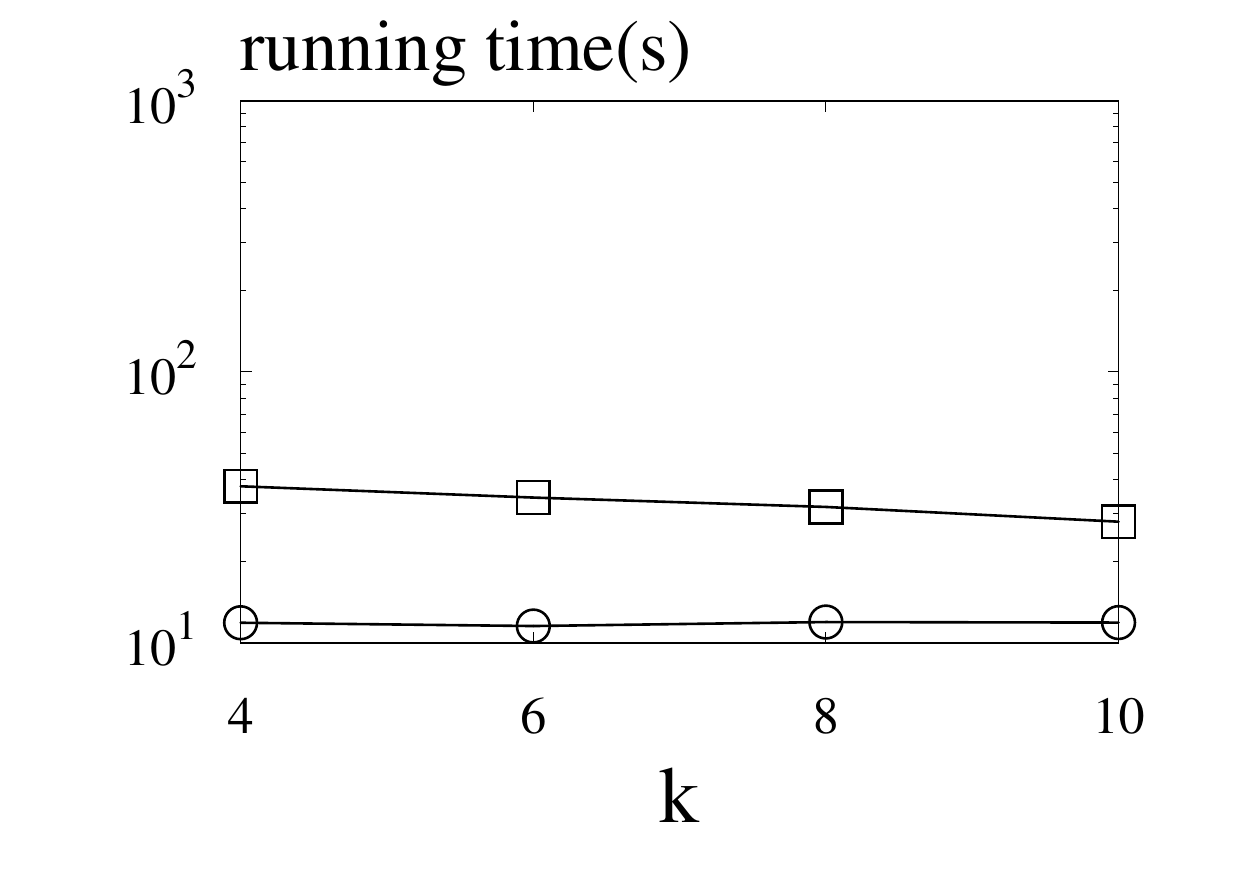} &
			\hspace{-6mm} \includegraphics[height=22mm]{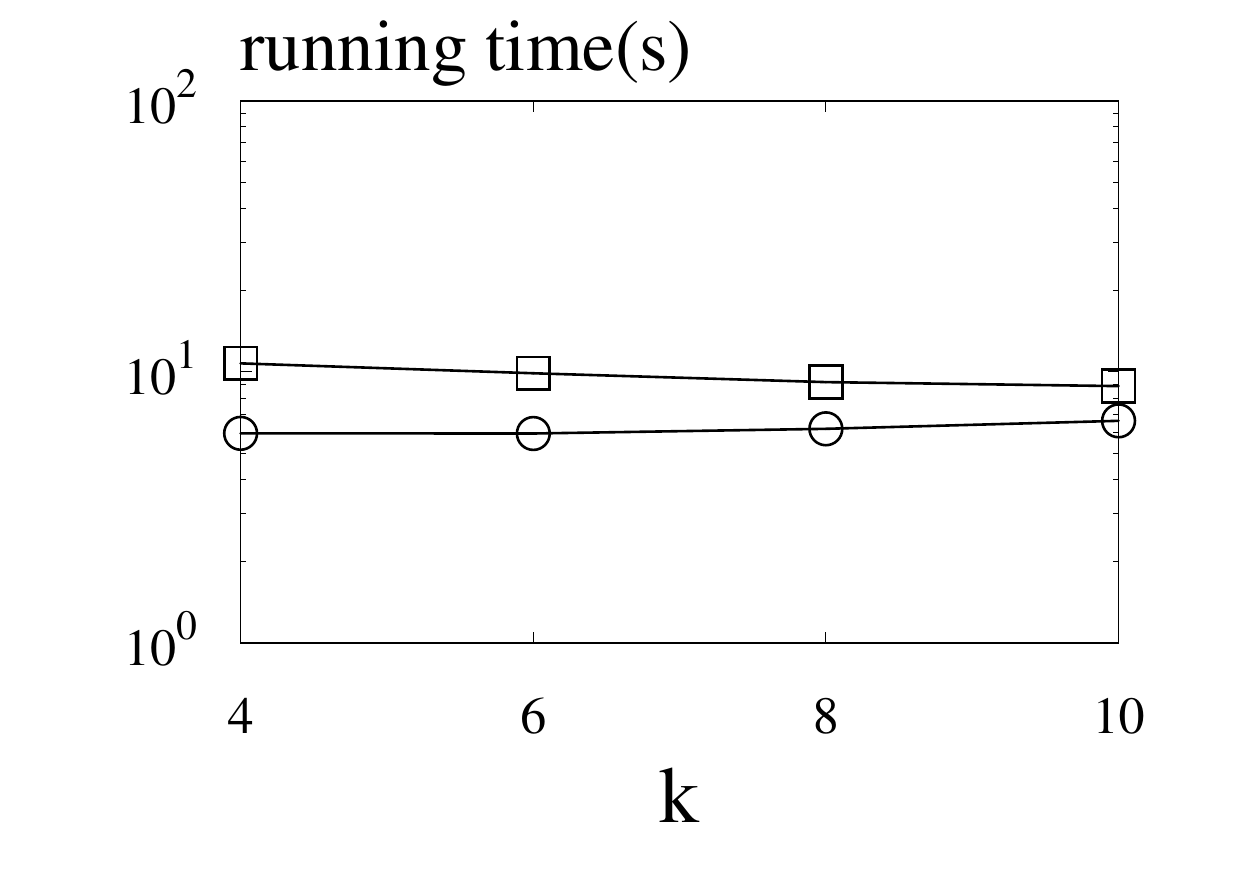} &
			\hspace{-6mm} \includegraphics[height=22mm]{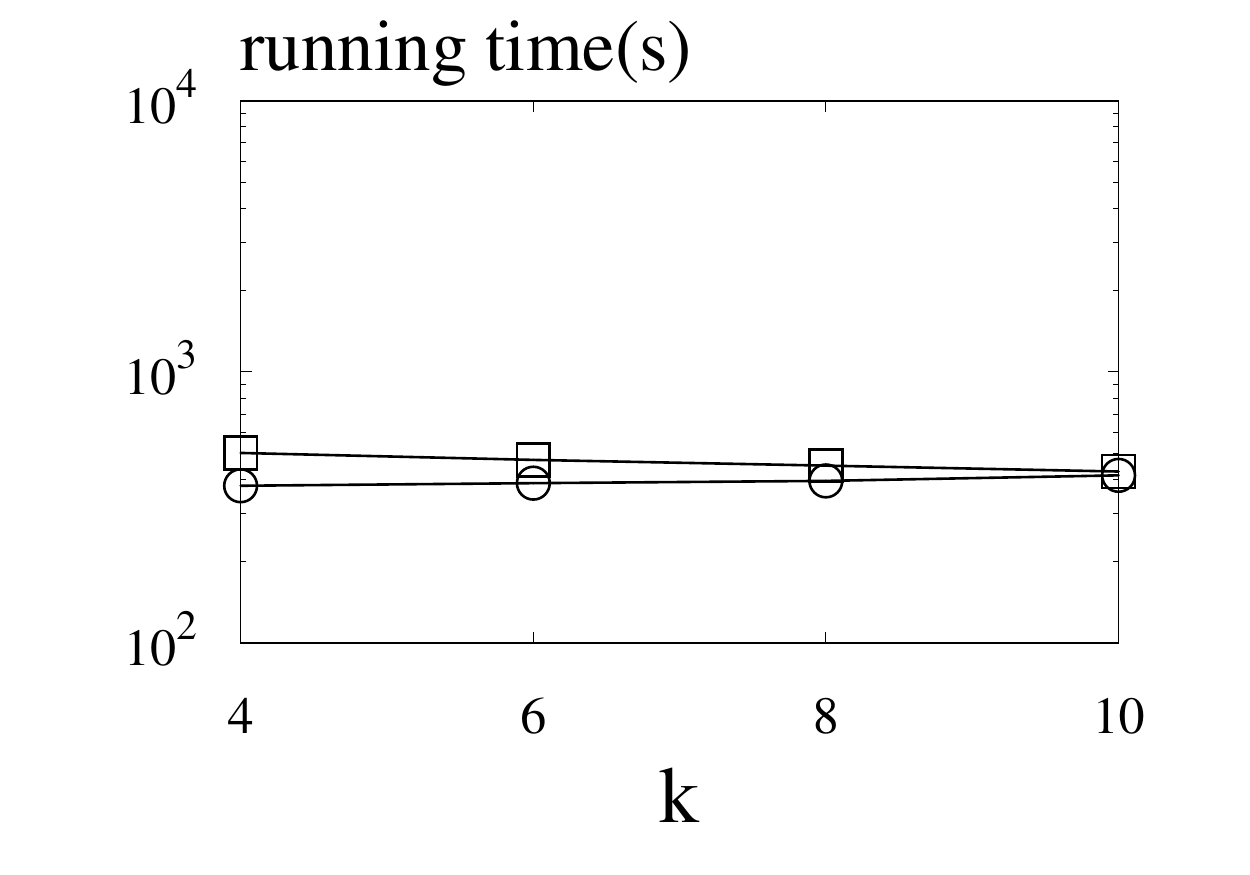}
			\\[-1mm]
			\hspace{-4mm} (a) Email &
			\hspace{-4mm} (b) DBLP &
			\hspace{-4mm} (c) Youtube &
			\hspace{-4mm} (d) Orkut &
			\hspace{-4mm} (e) Livejournal &
			\hspace{-4mm} (f) FriendSter \\[-1mm]
		\end{tabular}
		\vspace{-2mm}
		\caption{Running time vs. $k$ (avg, size-constrained)}
		\label{fig:time vs k size avg}
		\vspace{-3mm}
	\end{small}
\end{figure*}

\noindent \textbf{Exp-IV: Effect of $k$.} In this exp, we vary $k$ to evaluate the Local Search technique. Closer inspection of Figures~\ref{fig:time vs k size sum} and~\ref{fig:time vs k size avg} shows that the running time of algorithms decreases because the size of the graph decreases as $k$ increases. As the size of the graph decreases, the iteration of the local algorithm decreases as well.

\begin{figure*}[!t]
	\centering
	\vspace{-1mm}
	\begin{small}
		\begin{tabular}{cccccc}
			\hspace{-6mm} \includegraphics[height=22mm]{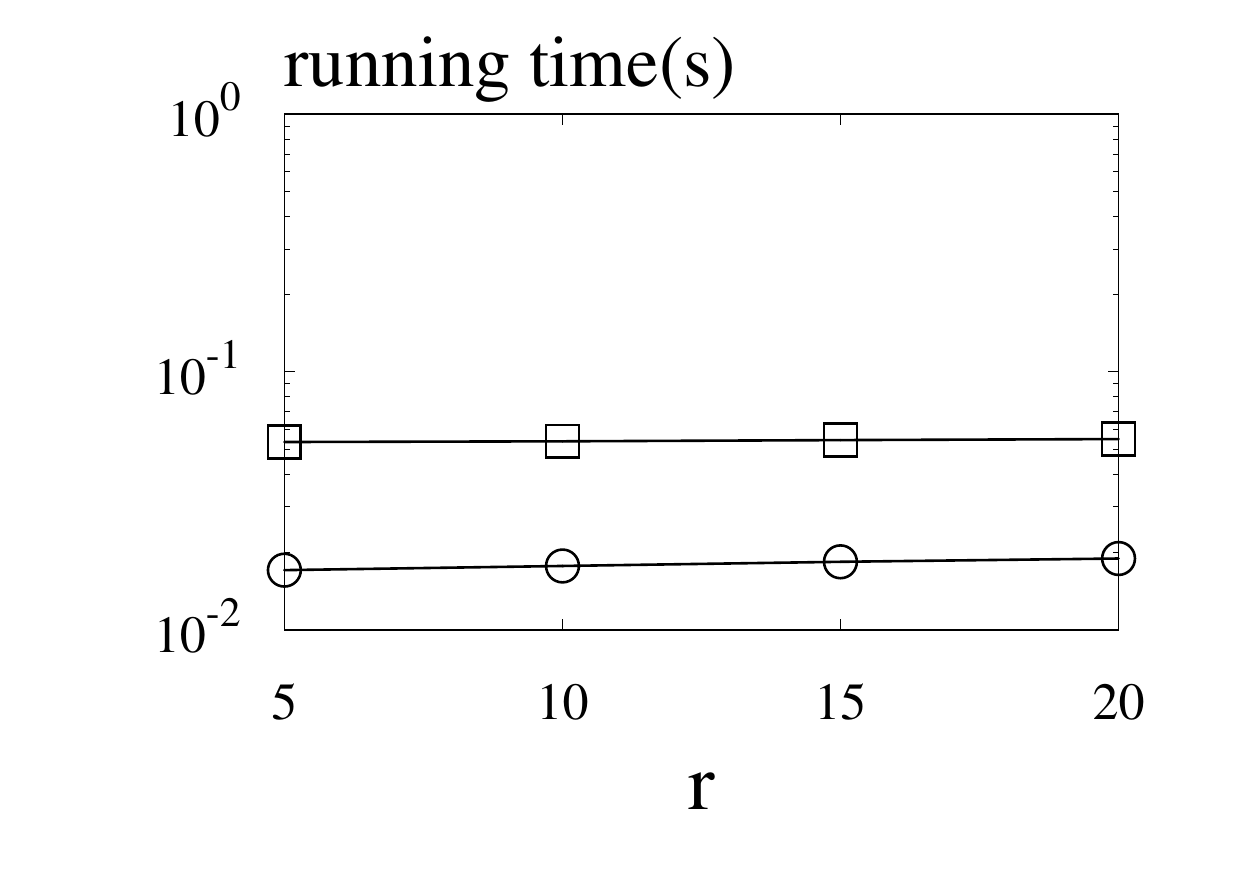} &
			\hspace{-6mm} \includegraphics[height=22mm]{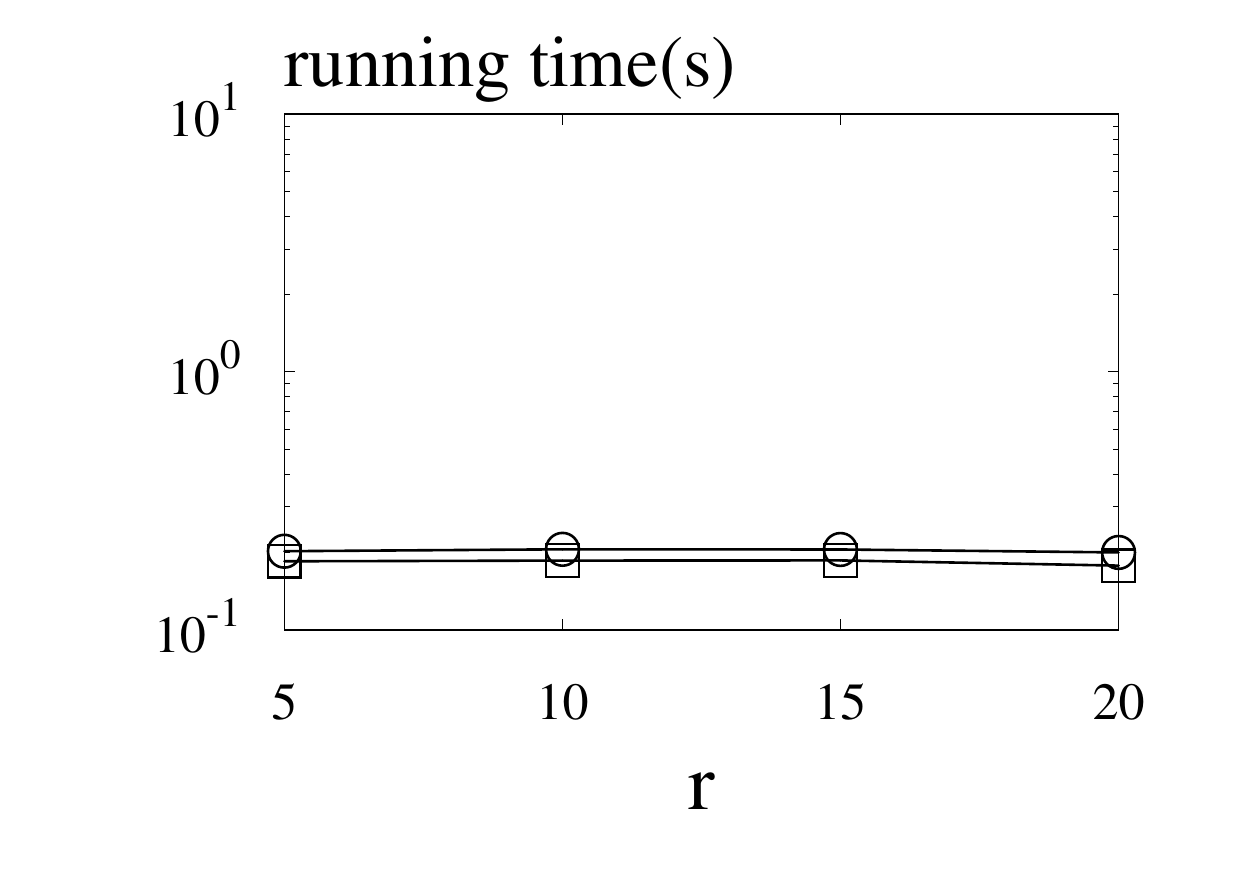} &
			\hspace{-6mm} \includegraphics[height=22mm]{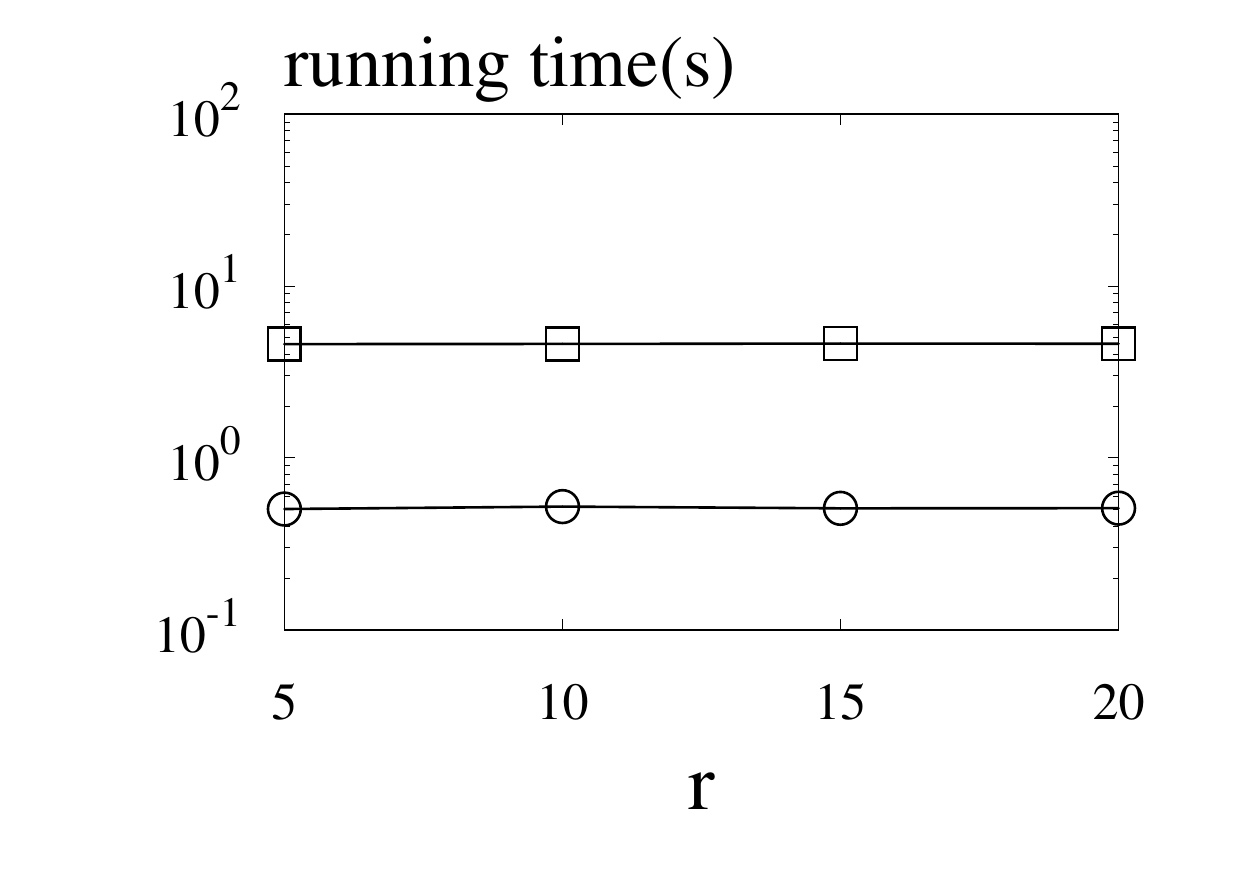} &
			\hspace{-6mm} \includegraphics[height=22mm]{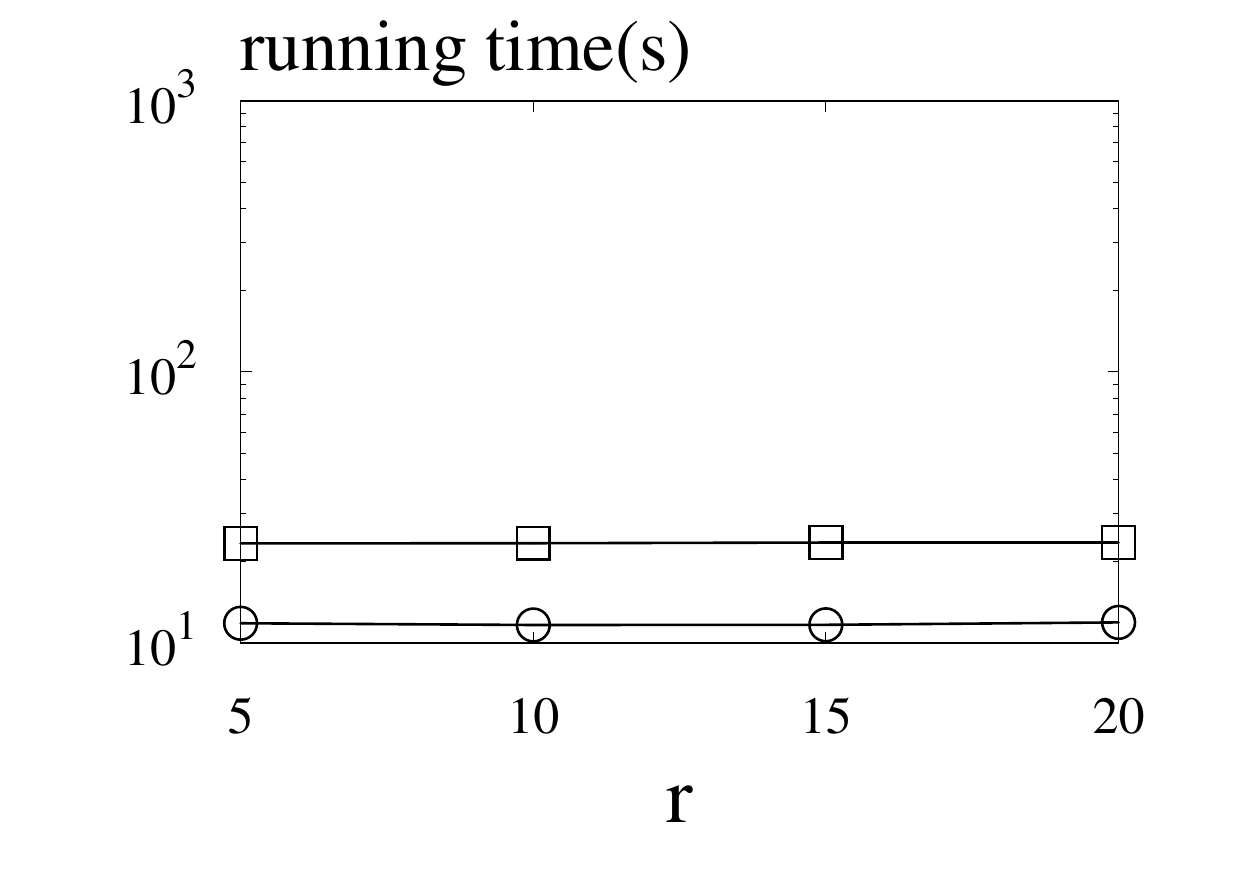} &
			\hspace{-6mm} \includegraphics[height=22mm]{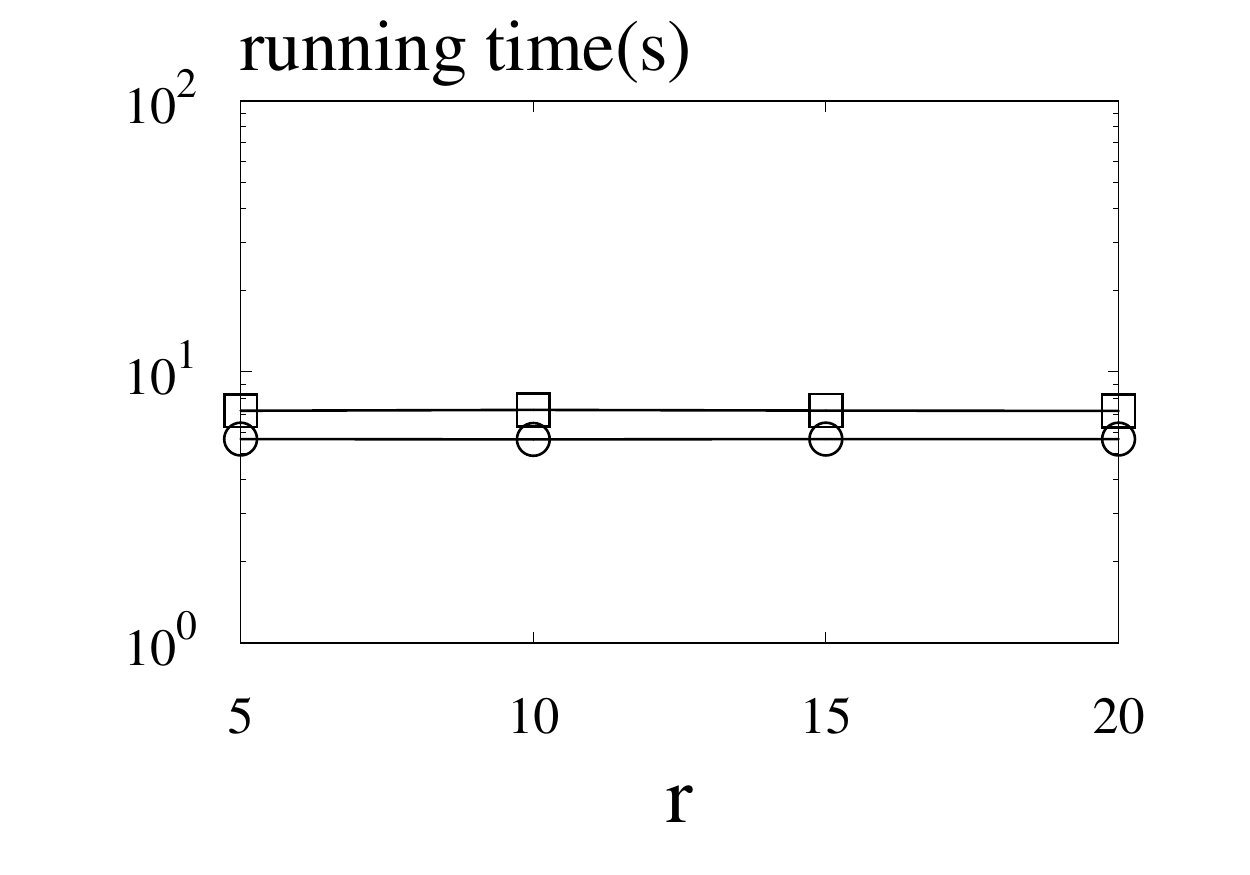} &
			\hspace{-6mm} \includegraphics[height=22mm]{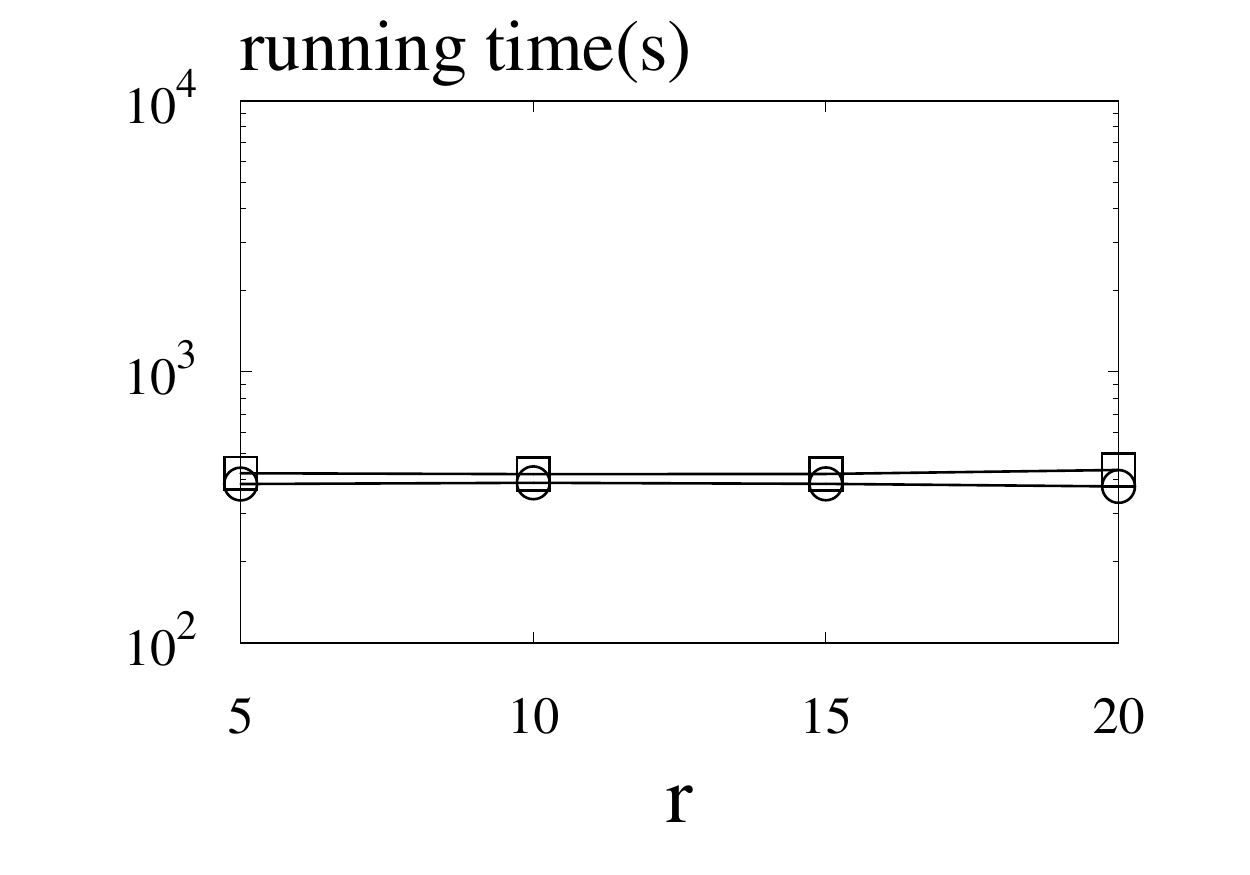}
			\\[-1mm]
			\hspace{-4mm} (a) Email &
			\hspace{-4mm} (b) DBLP &
			\hspace{-4mm} (c) Youtube &
			\hspace{-4mm} (d) Orkut &
			\hspace{-4mm} (e) Livejournal &
			\hspace{-4mm} (f) FriendSter \\[-1mm]
		\end{tabular}
		\vspace{-2mm}
		\caption{Running time vs. $r$ (sum, size-constrained)}
		\label{fig:time vs r size sum}
		\vspace{-4mm}
	\end{small}
\end{figure*}

\begin{figure*}[!t]
	\centering
	\vspace{-1mm}
	\begin{small}
		\begin{tabular}{cccccc}
			\multicolumn{6}{c}{\hspace{-6mm} \includegraphics[height=10mm]{size_legend.pdf}}  \\[-3mm]
			\hspace{-6mm} \includegraphics[height=22mm]{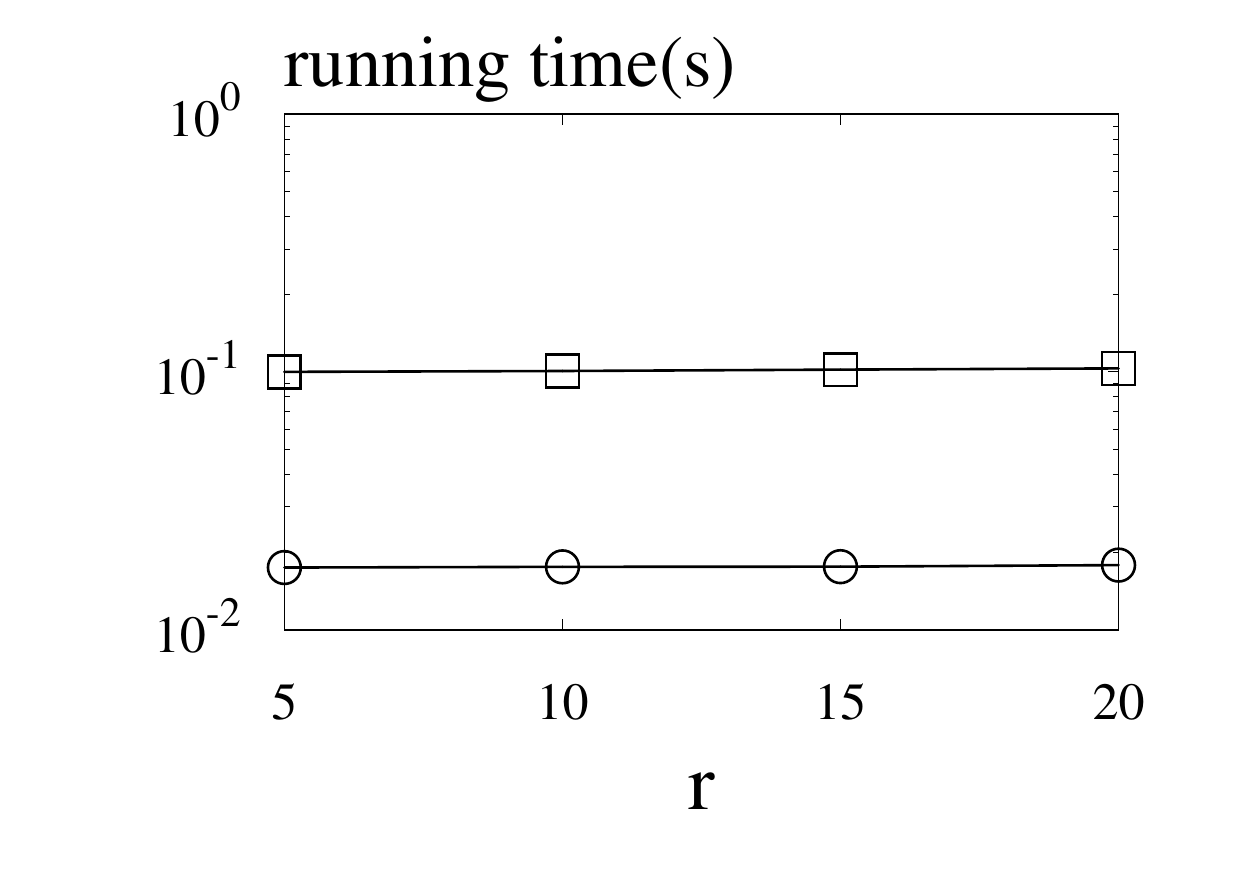} &
			\hspace{-6mm} \includegraphics[height=22mm]{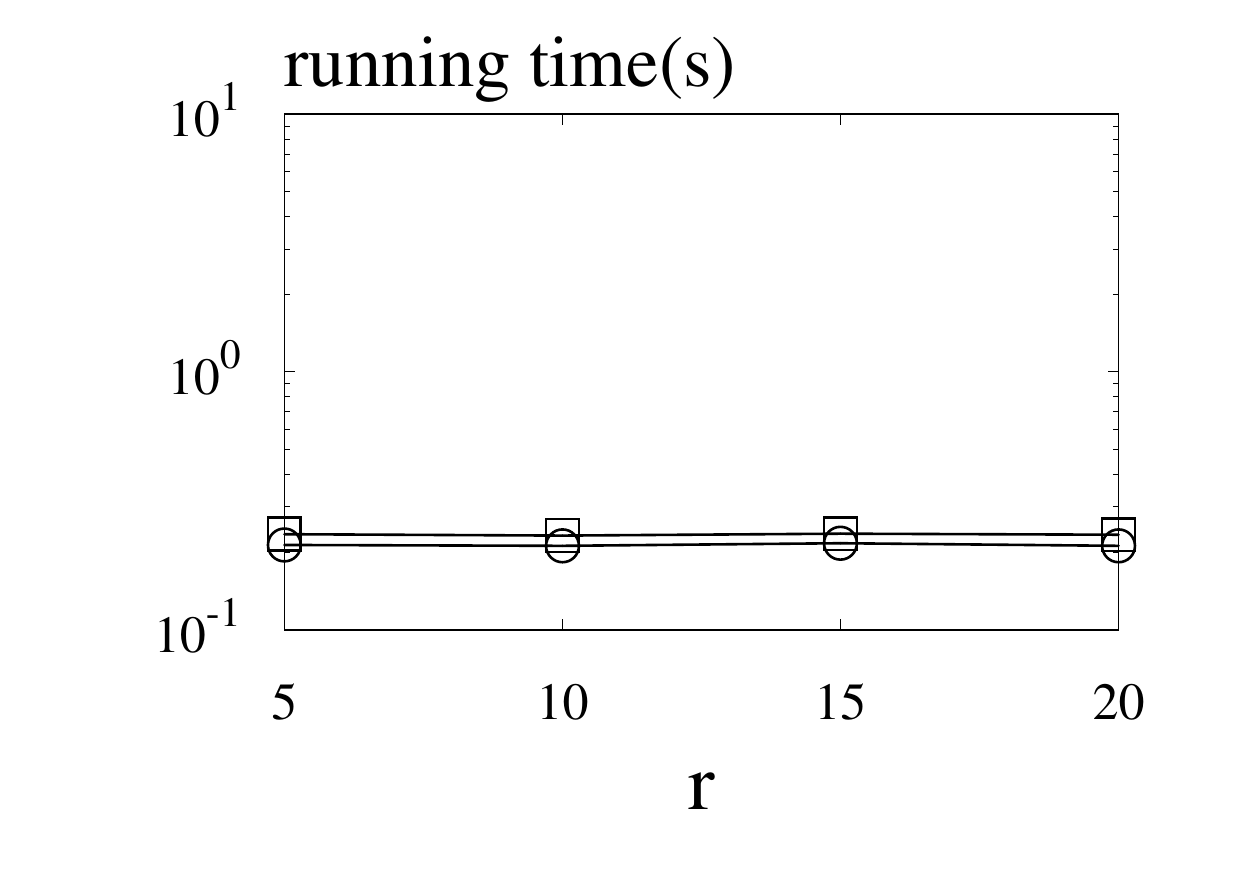} &
			\hspace{-6mm} \includegraphics[height=22mm]{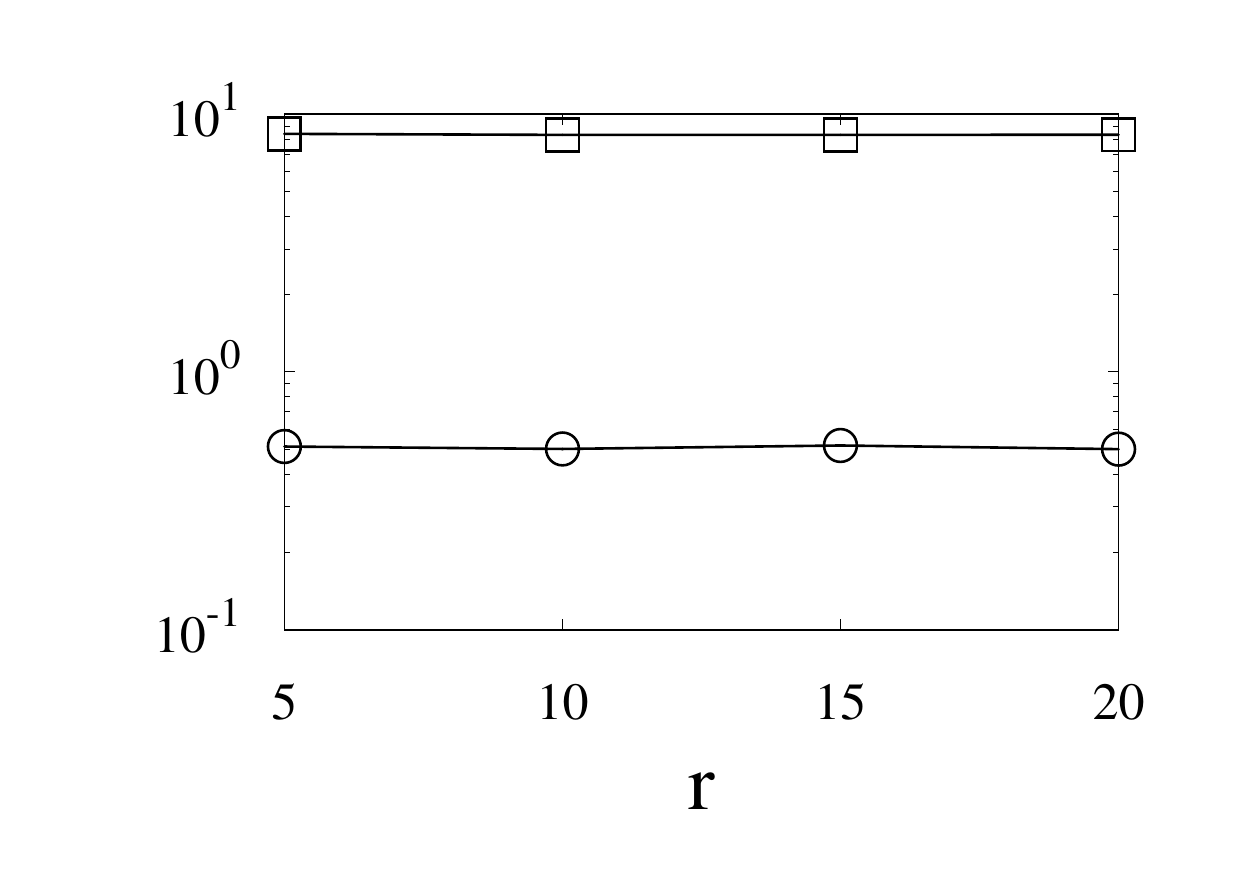} &
			\hspace{-6mm} \includegraphics[height=22mm]{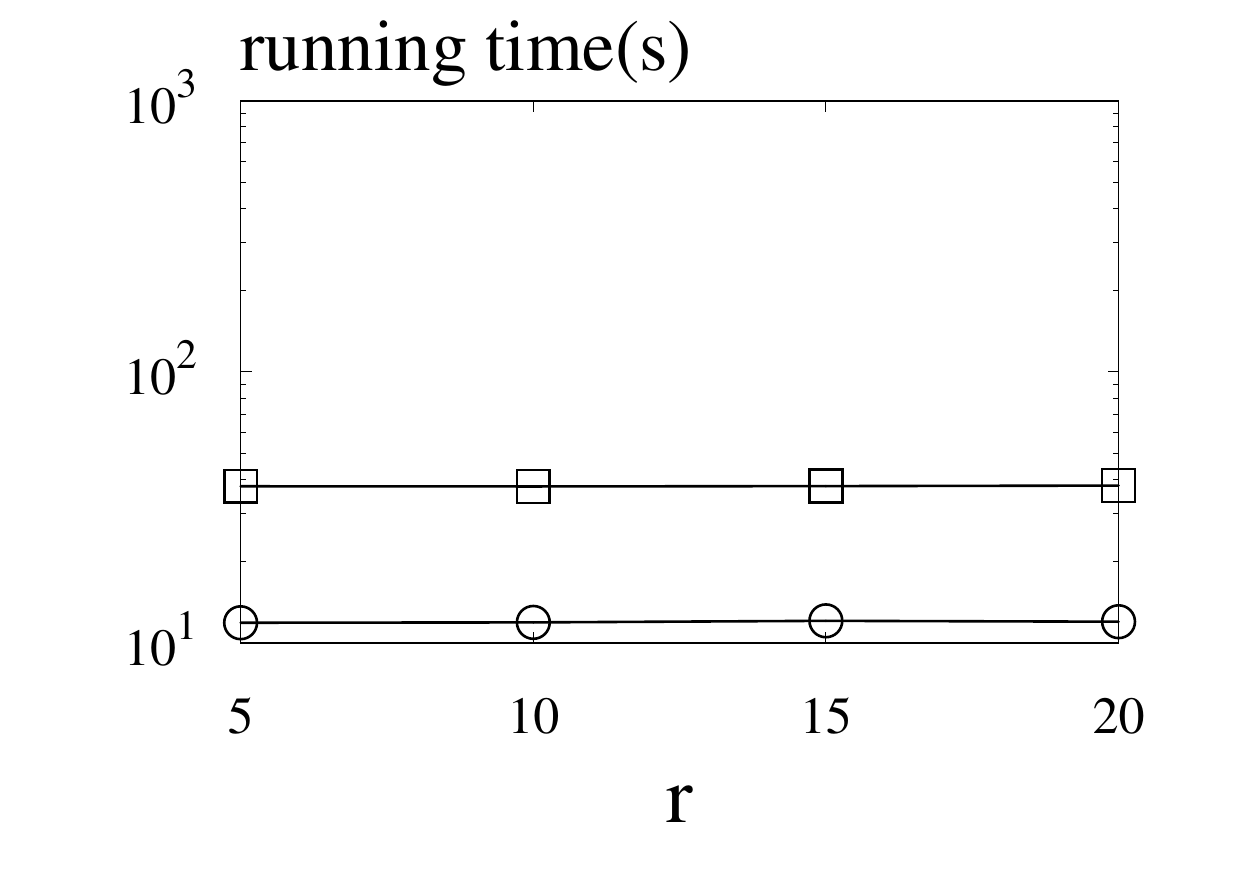} &
			\hspace{-6mm} \includegraphics[height=22mm]{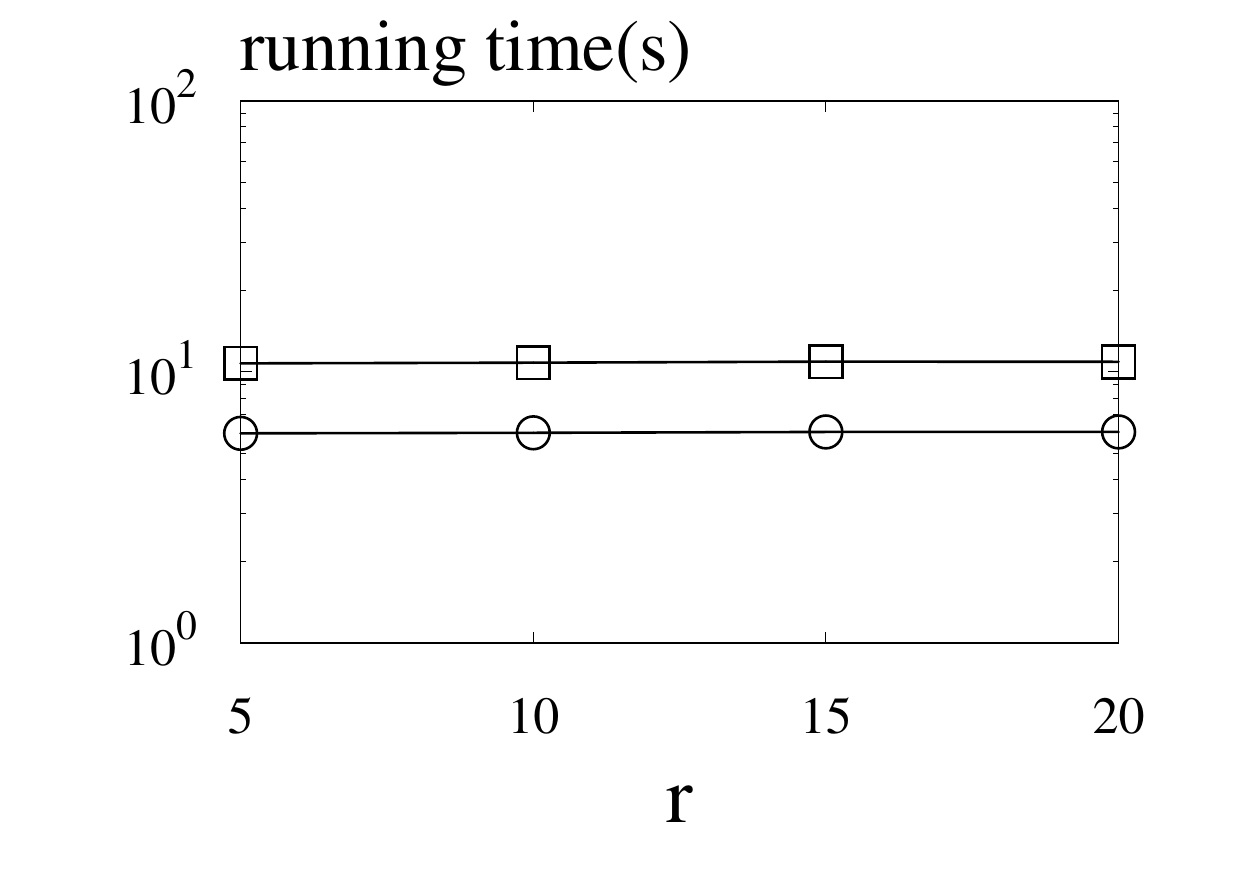} &
			\hspace{-6mm} \includegraphics[height=22mm]{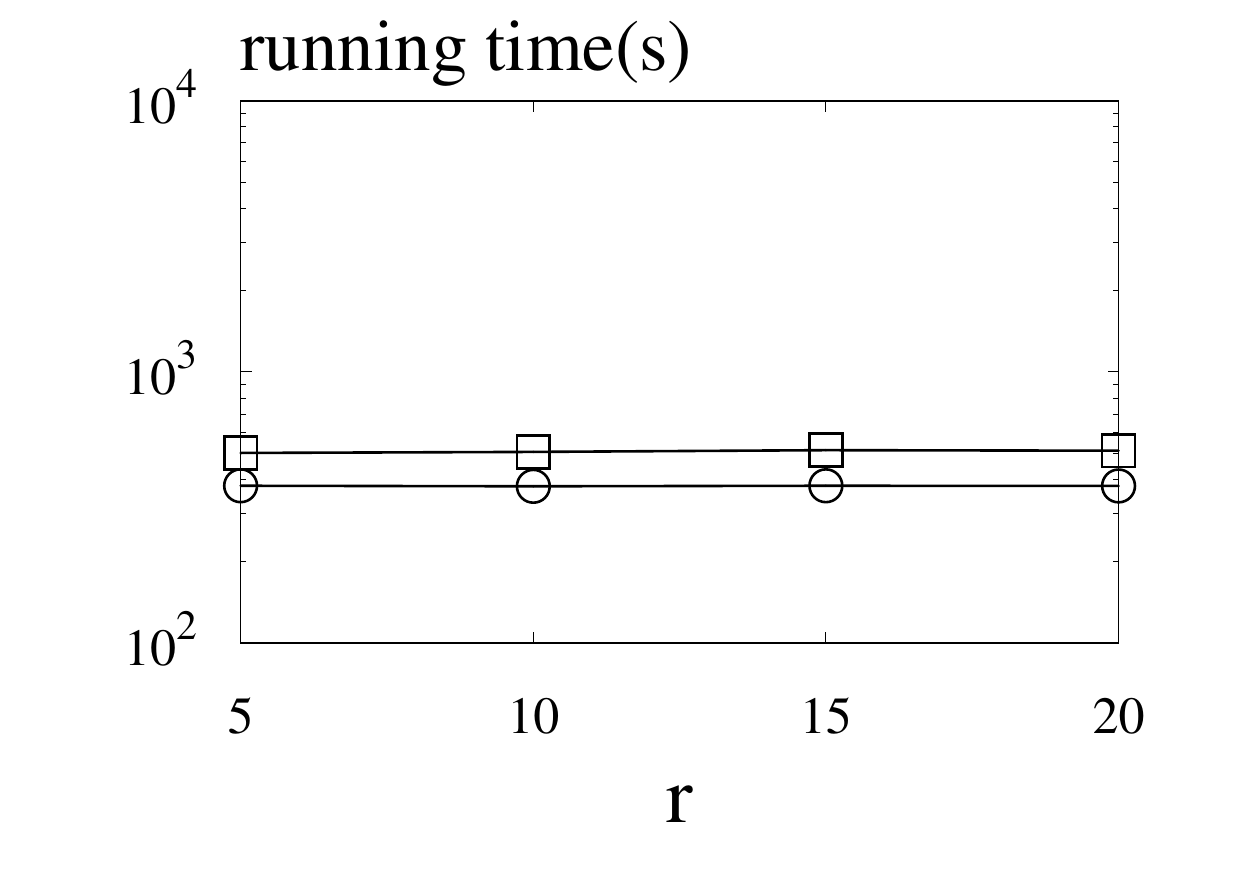}
			\\[-1mm]
			\hspace{-4mm} (a) Email &
			\hspace{-4mm} (b) DBLP &
			\hspace{-4mm} (c) Youtube &
			\hspace{-4mm} (d) Orkut &
			\hspace{-4mm} (e) Livejournal &
			\hspace{-4mm} (f) FriendSter \\[-1mm]
		\end{tabular}
		\vspace{-2mm}
		\caption{Running time vs. $r$ (avg, size-constrained)}
		\label{fig:time vs r size avg}
		\vspace{-4mm}
	\end{small}
\end{figure*}

\begin{figure*}[!t]
	\centering
	\vspace{-1mm}
	\begin{small}
		\begin{tabular}{cccccc}
			\multicolumn{6}{c}{\hspace{-6mm} \includegraphics[height=10mm]{size_legend.pdf}}  \\[-3mm]
			\hspace{-6mm} \includegraphics[height=22mm]{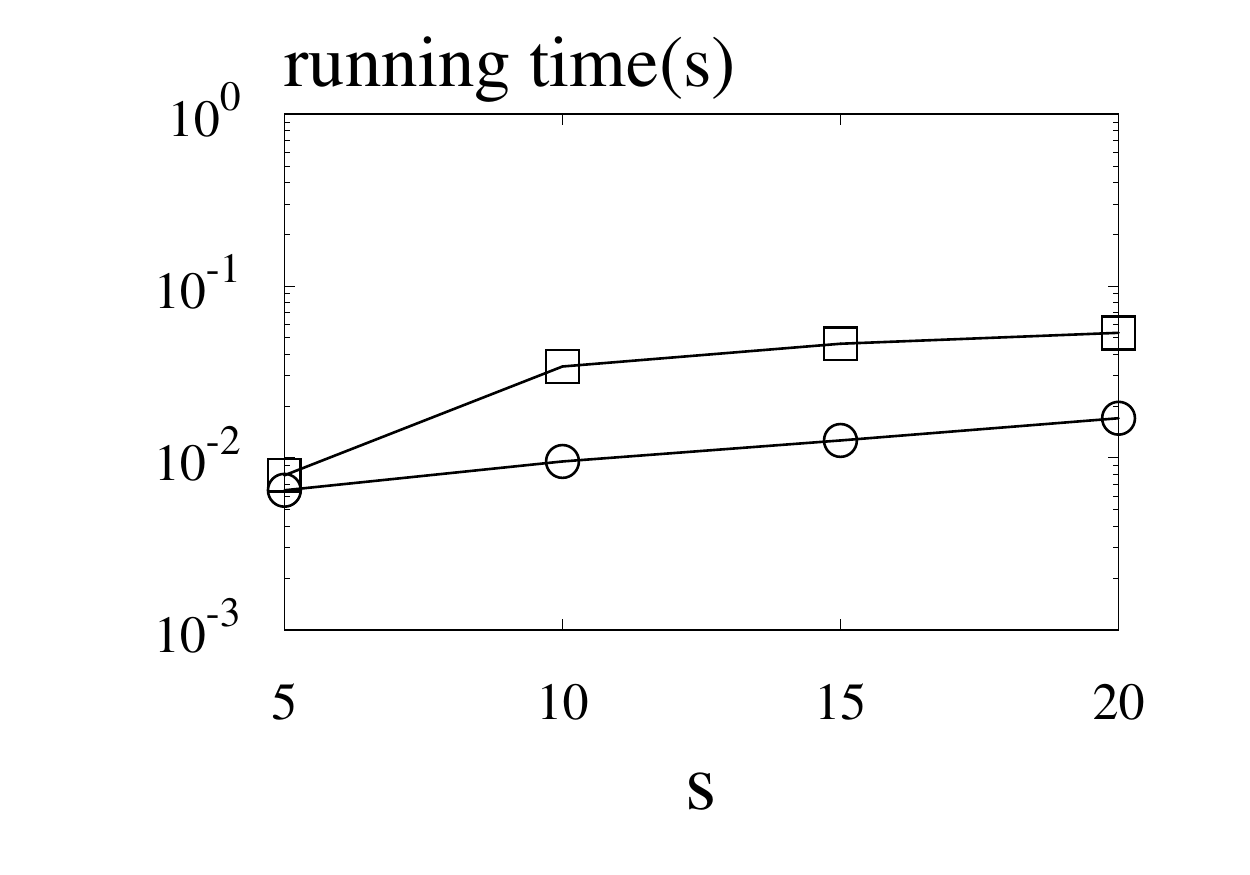} &
			\hspace{-6mm} \includegraphics[height=22mm]{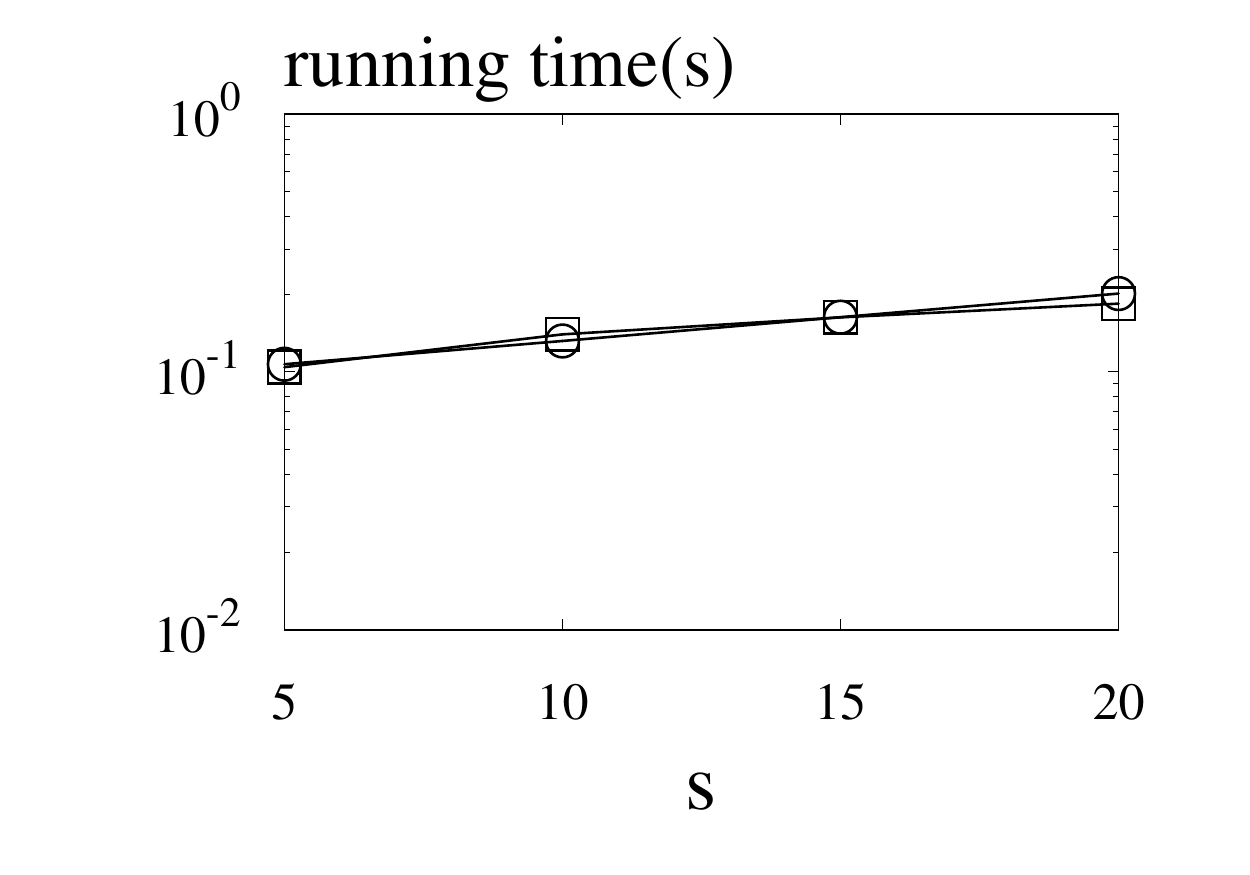} &
			\hspace{-6mm} \includegraphics[height=22mm]{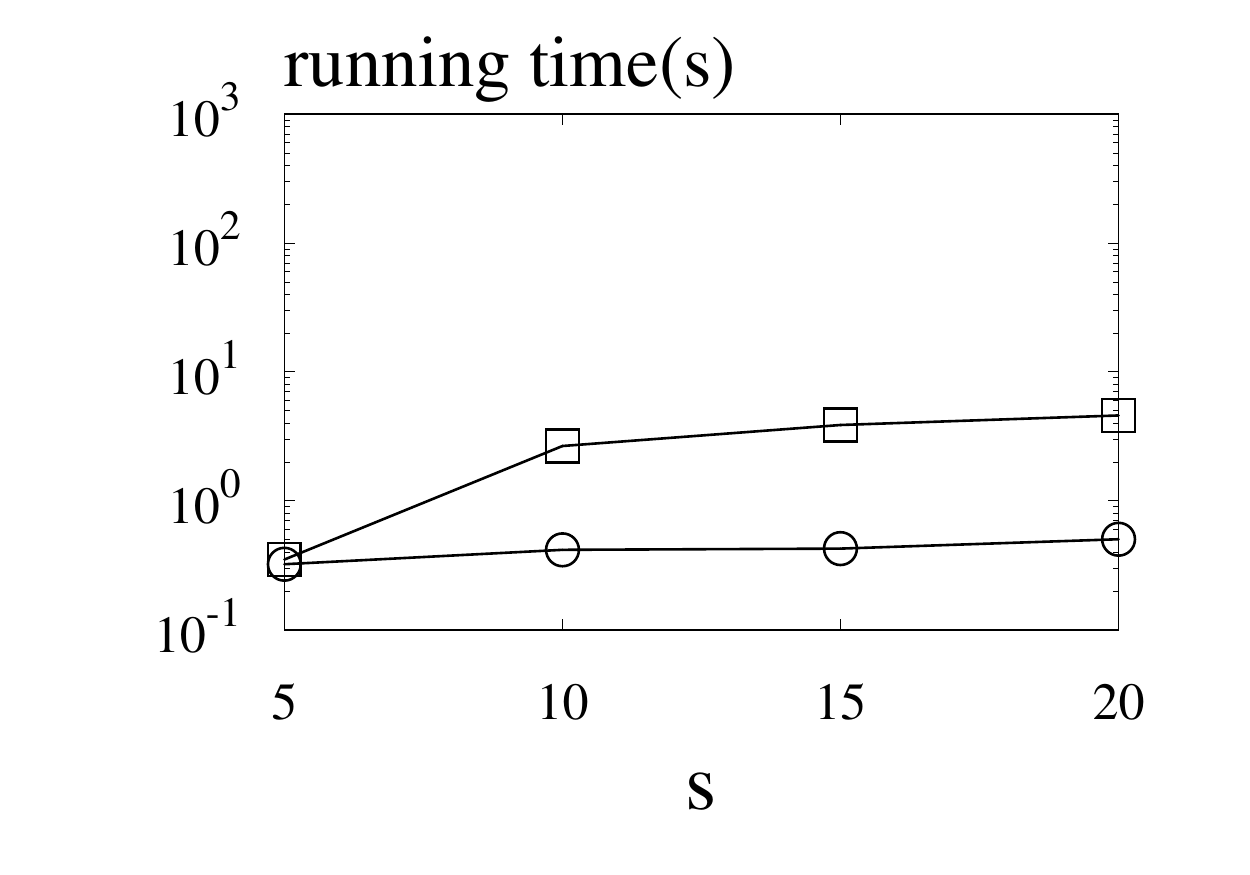} &
			\hspace{-6mm} \includegraphics[height=22mm]{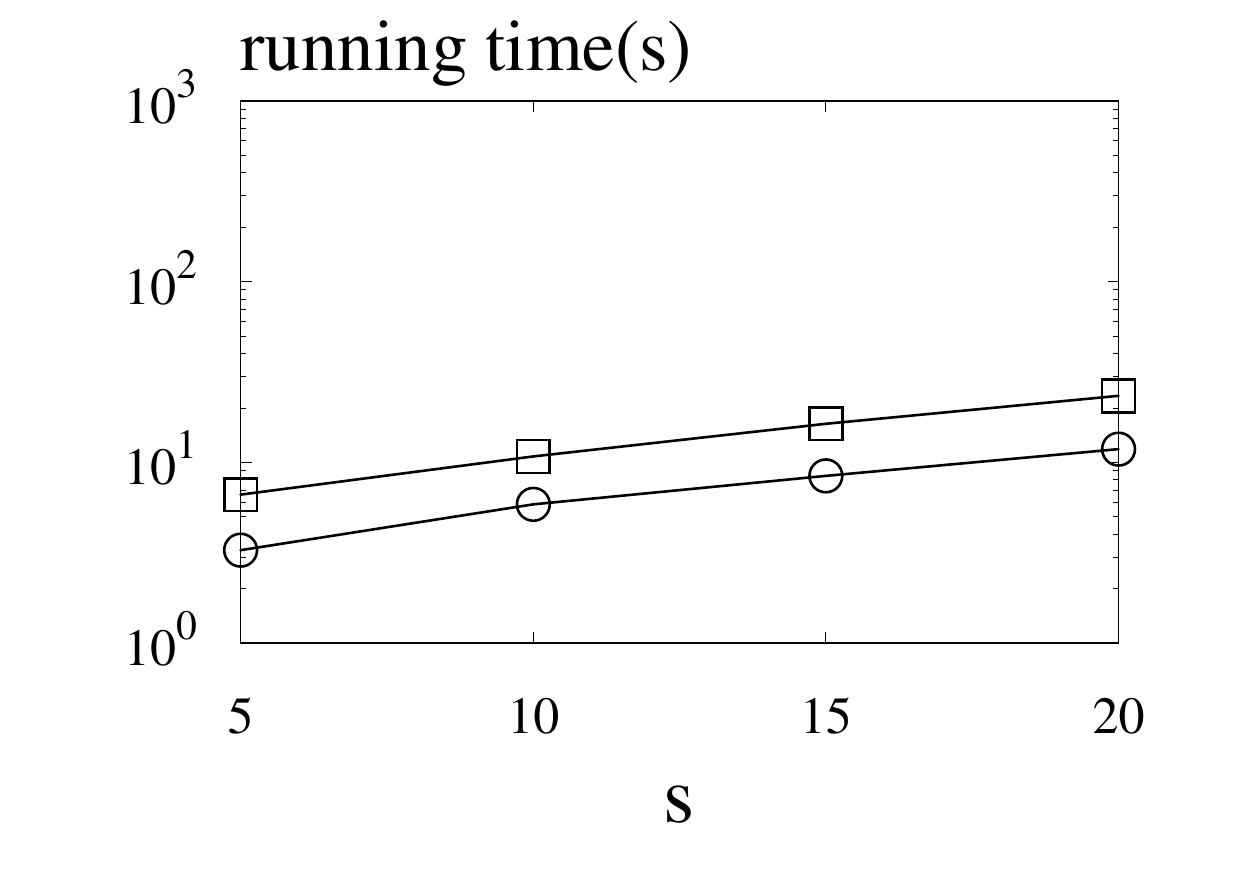} &
			\hspace{-6mm} \includegraphics[height=22mm]{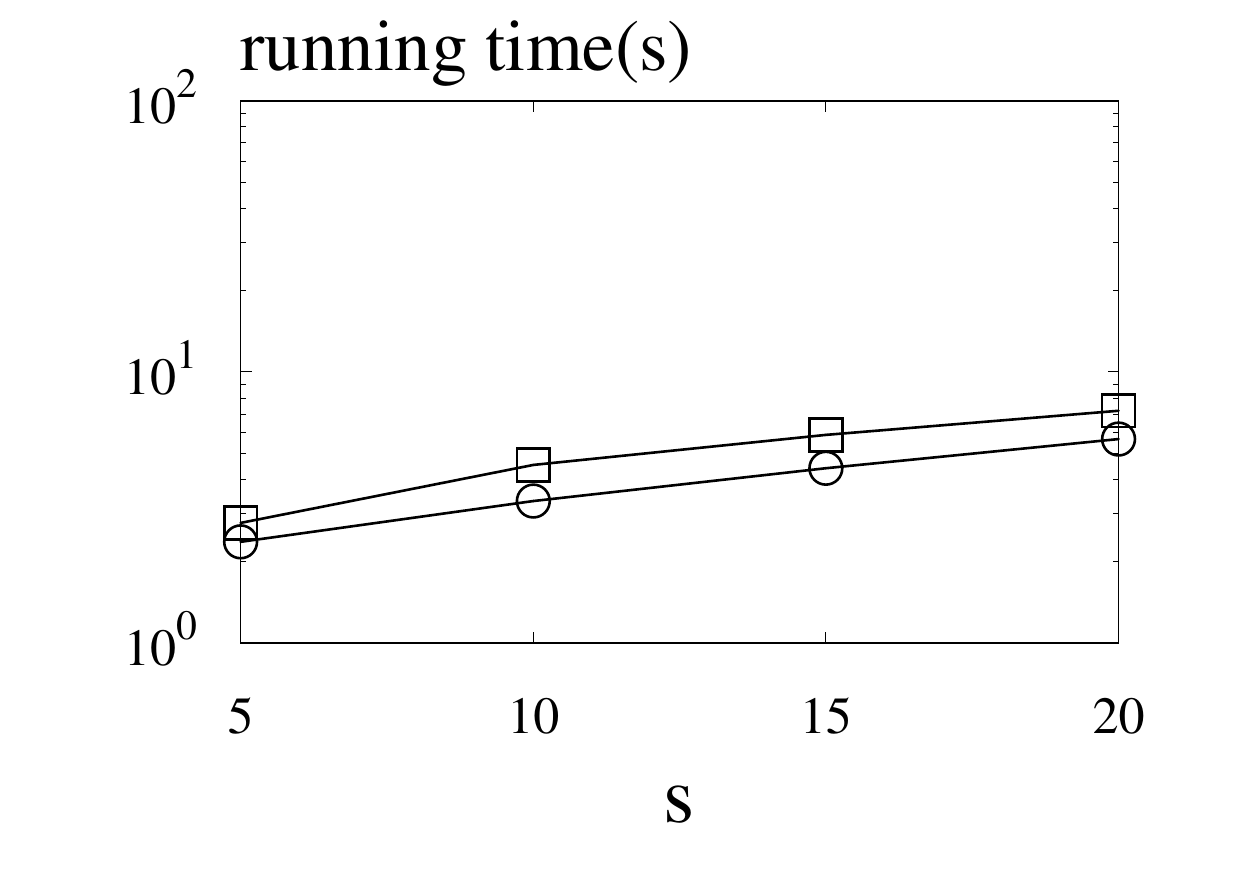} &
			\hspace{-6mm} \includegraphics[height=22mm]{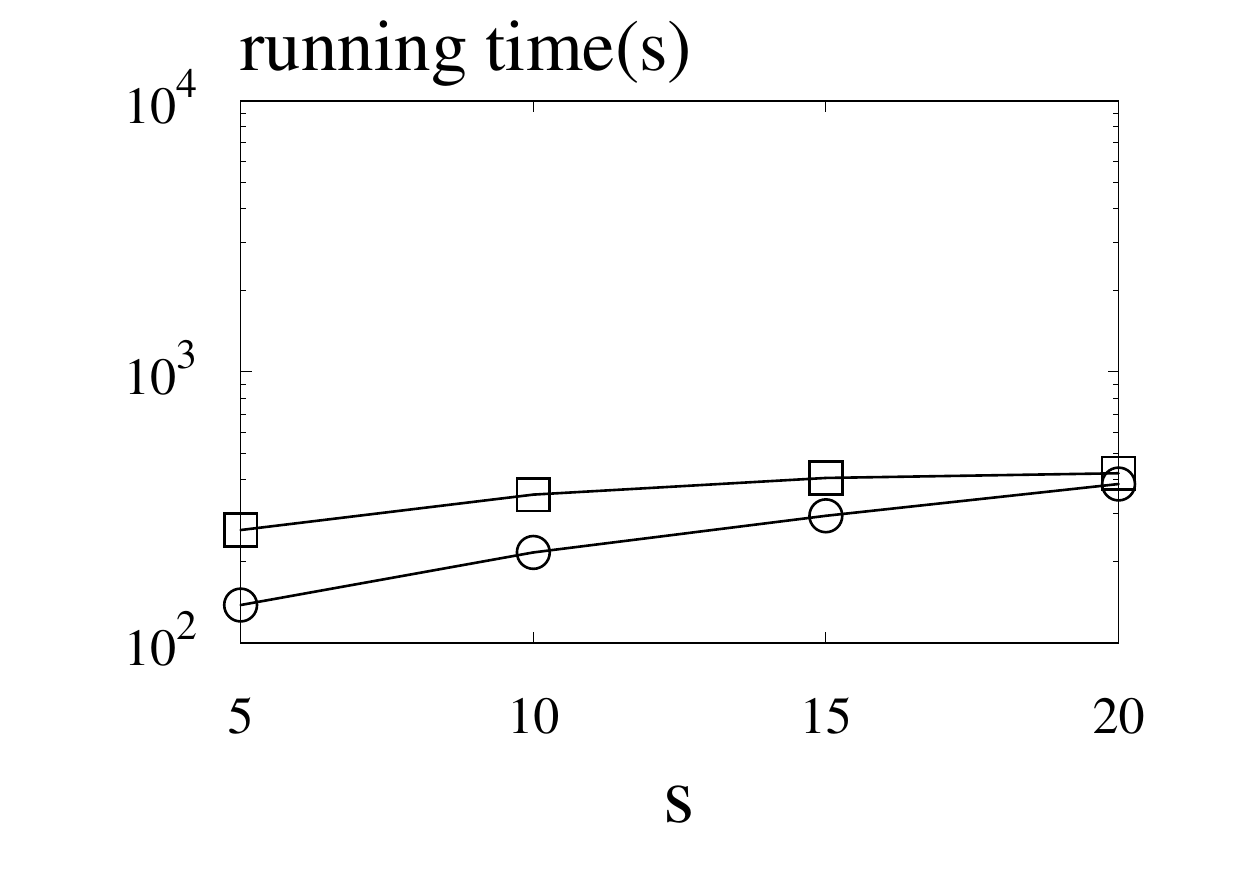}
			\\[-1mm]
			\hspace{-4mm} (a) Email &
			\hspace{-4mm} (b) DBLP &
			\hspace{-4mm} (c) Youtube &
			\hspace{-4mm} (d) Orkut &
			\hspace{-4mm} (e) Livejournal &
			\hspace{-4mm} (f) FriendSter \\[-1mm]
		\end{tabular}
		\vspace{-2mm}
		\caption{Running time vs. $s$ (sum, size-constrained)}
		\label{fig:time vs s size sum}
		\vspace{-4mm}
	\end{small}
\end{figure*}

\begin{figure*}[!t]
	\centering
	\vspace{-1mm}
	\begin{small}
		\begin{tabular}{cccccc}
			\multicolumn{6}{c}{\hspace{-6mm} \includegraphics[height=10mm]{size_legend.pdf}}  \\[-3mm]
			\hspace{-6mm} \includegraphics[height=22mm]{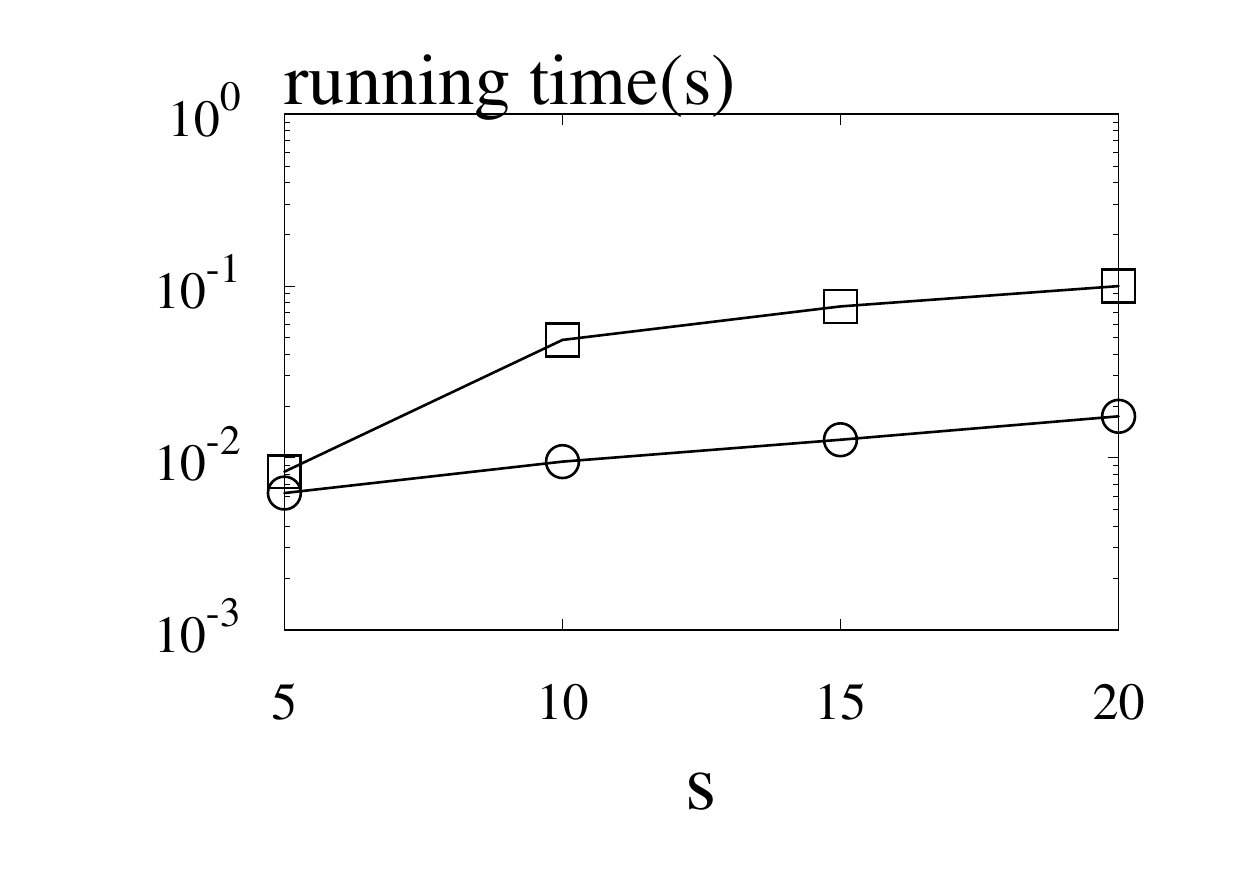} &
			\hspace{-6mm} \includegraphics[height=22mm]{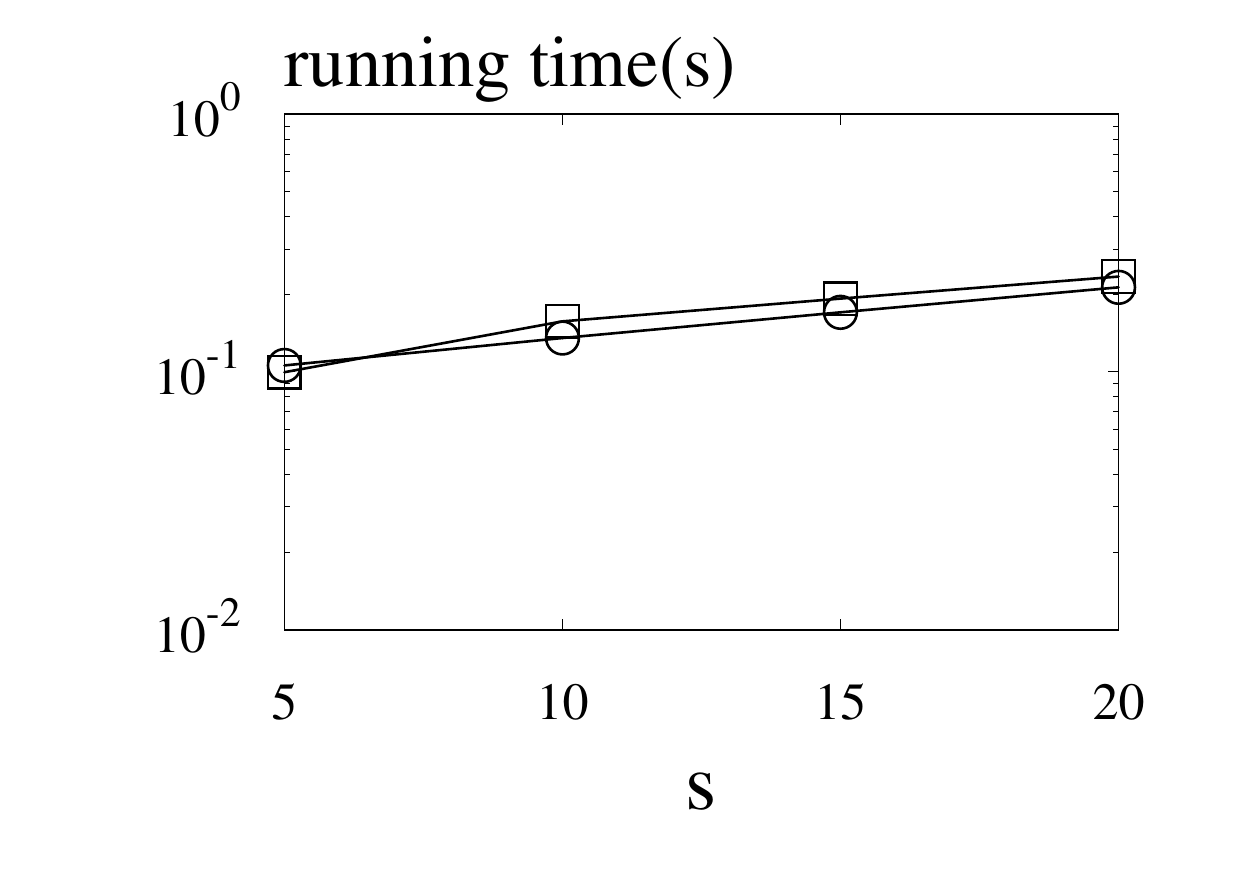} &
			\hspace{-6mm} \includegraphics[height=22mm]{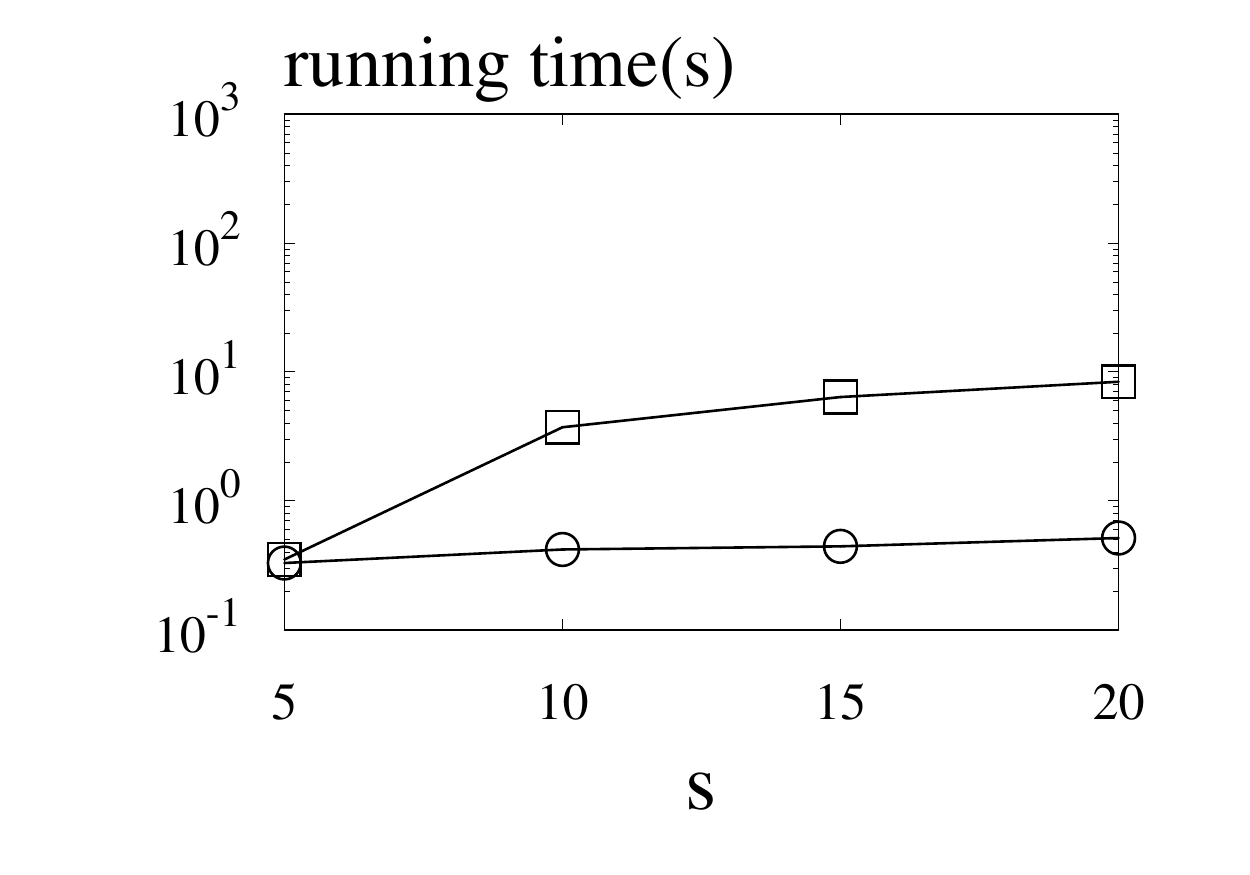} &
			\hspace{-6mm} \includegraphics[height=22mm]{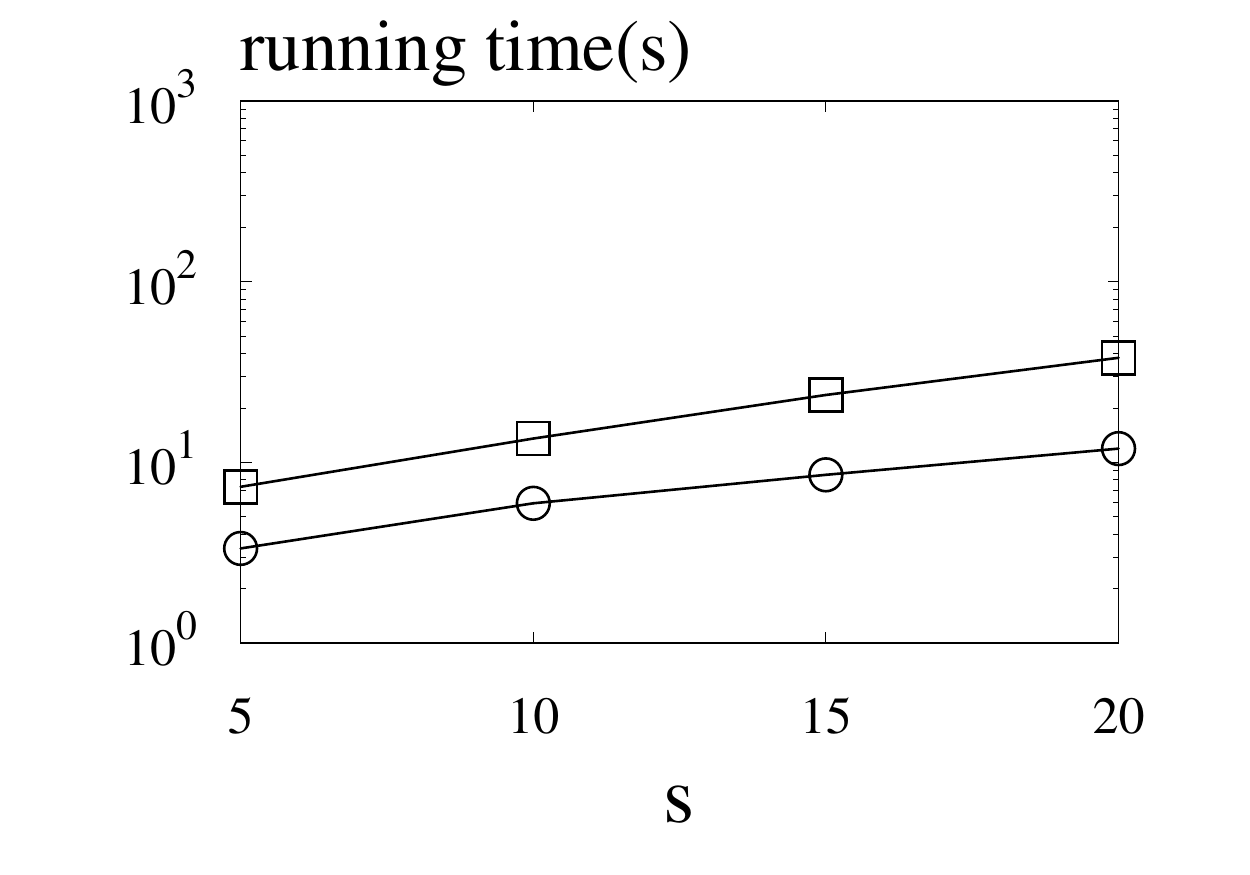} &
			\hspace{-6mm} \includegraphics[height=22mm]{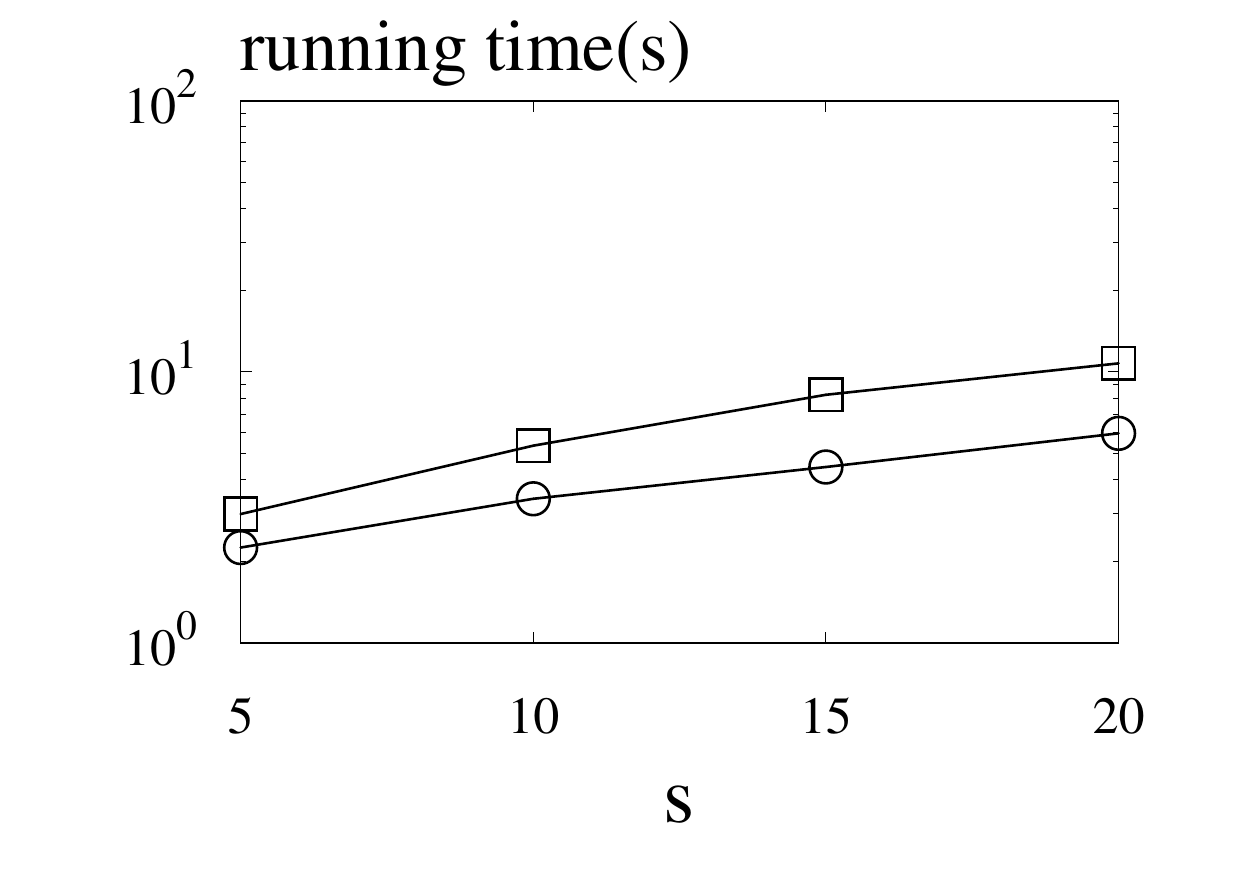} &
			\hspace{-6mm} \includegraphics[height=22mm]{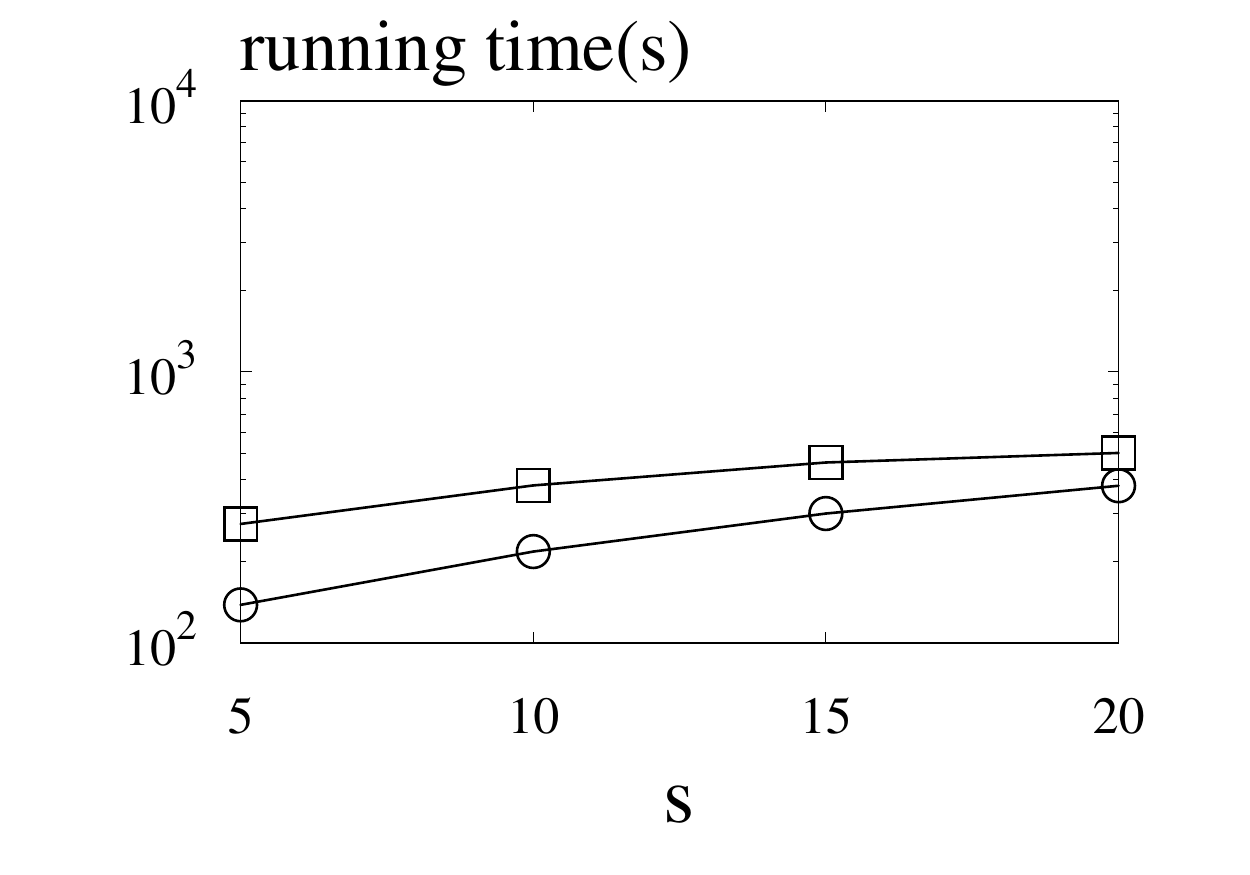}
			\\[-1mm]
			\hspace{-4mm} (a) Email &
			\hspace{-4mm} (b) DBLP &
			\hspace{-4mm} (c) Youtube &
			\hspace{-4mm} (d) Orkut &
			\hspace{-4mm} (e) Livejournal &
			\hspace{-4mm} (f) FriendSter \\[-1mm]
		\end{tabular}
		\vspace{-2mm}
		\caption{Running time vs. $s$ (avg, size-constrained)}
		\label{fig:time vs s size avg}
		\vspace{-2mm}
	\end{small}
\end{figure*}

\noindent \textbf{Exp-V: Effect of $r$.} We vary $r$ to evaluate the Local Search algorithm. Figure~\ref{fig:time vs r size sum} and Figure~\ref{fig:time vs r size avg} demonstrate that the performance of the Local Search algorithm is insensitive to $r$. Notably, $r$ would not be very large in reality since it is infeasible for uses to choose from numerous candidates. In such a scenario, the Local Search algorithm is insensitive to $r$, since the algorithm would always compute more than $r$ $k$-influential communities. Thus, when $r$ is not large, its value would not affect the performance of the algorithm.

\noindent \textbf{Exp-VI: Effect of $s$.} In this experiment, we vary $s$ to evaluate the efficiency of the Local Search algorithm. Figures~\ref{fig:time vs s size sum} and~\ref{fig:time vs s size avg} indicate that the running time of algorithms increases since we have to search more neighbor vertices at each iteration with the increase of $s$.

\begin{figure}[!t]
	\centering
	\vspace{-3mm}
	\begin{small}
		\begin{tabular}{c}
			\hspace{-6mm} \includegraphics[height=25mm]{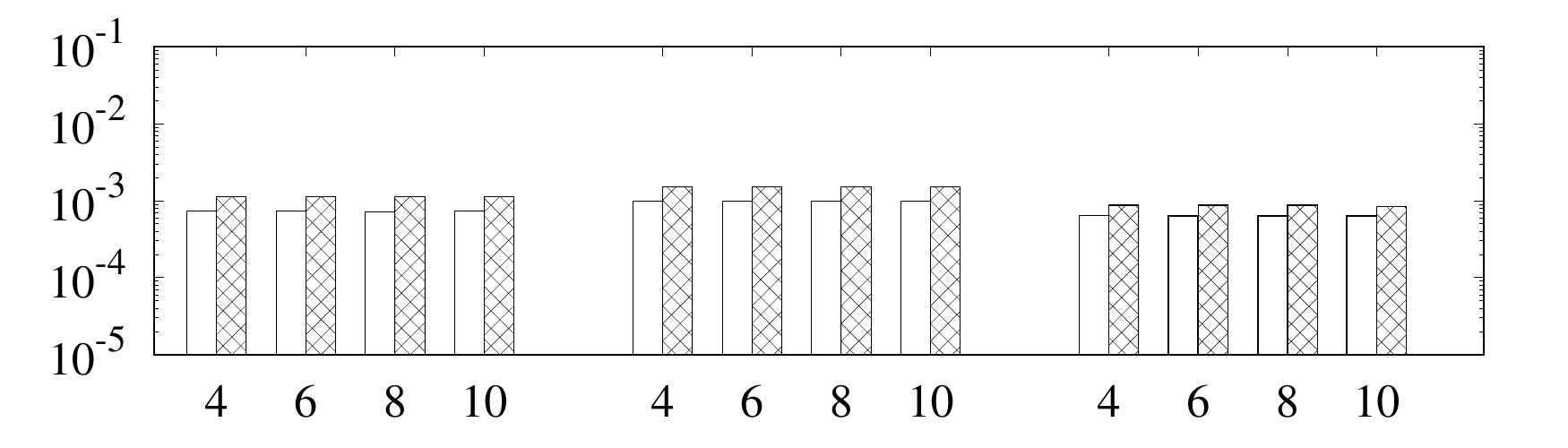}
			\\[-1mm]
			\hspace{1mm} (a) DBLP
			\hspace{8mm} (b) Orkut
			\hspace{8mm} (c) LiveJournal \\[-1mm]
		\end{tabular}
		\vspace{-2mm}
		\caption{$r$-th influence value (sum, size-constrained)}
		\label{fig:k inf sum}
		\vspace{-4mm}
	\end{small}
\end{figure}

\begin{figure}[!t]
	\centering
	\vspace{-3mm}
	\begin{small}
		\begin{tabular}{c}
			\multicolumn{1}{c}{\hspace{-6mm} \includegraphics[height=9mm]{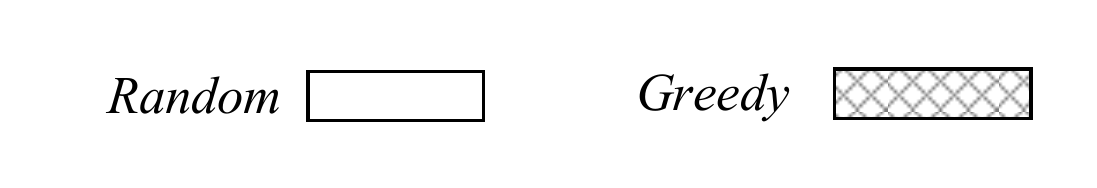}}  \\[-5mm]
			\hspace{-6mm} \includegraphics[height=25mm]{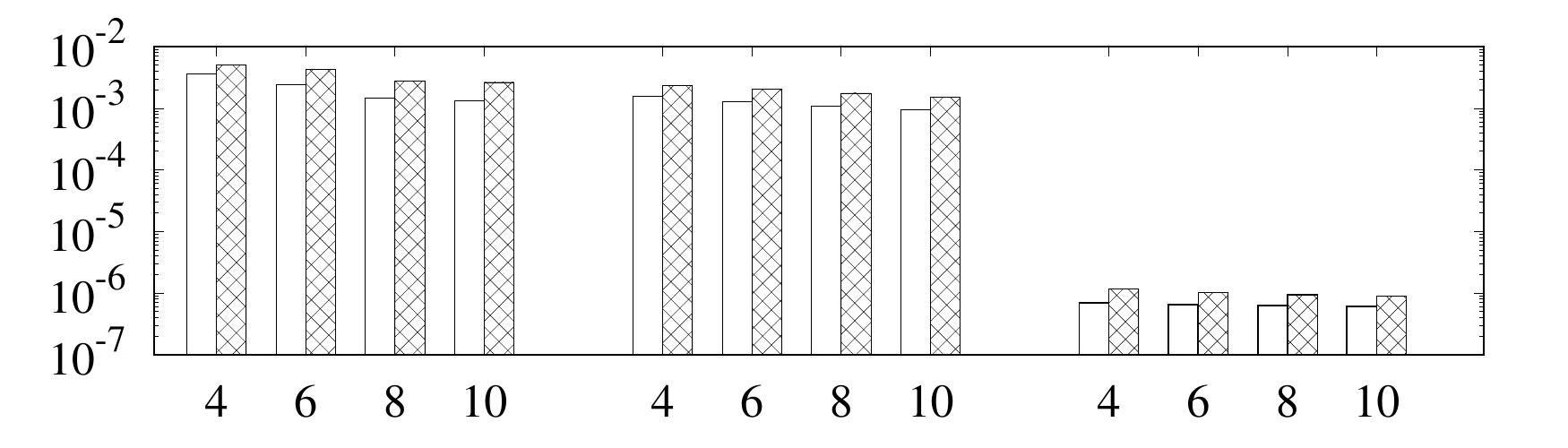}
			\\[-1mm]
			\hspace{1mm} (a) Email
			\hspace{8mm} (b) Youtube
			\hspace{8mm} (c) FriendSter \\[-1mm]
		\end{tabular}
		\vspace{-2mm}
		\caption{$r$-th influence value (avg, size-constrained)}
		\label{fig:k inf avg}
		\vspace{-4mm}
	\end{small}
\end{figure}

\noindent \textbf{Exp-VII: Effectiveness.} In this evaluation, we compare the greedy strategy with the random one. We fix $r=5$ and $s=20$, by varying $k$ from $\{4, 6, 8, 10\}$. Figure~\ref{fig:k inf sum} and Figure~\ref{fig:k inf avg} demonstrate that the influence value of $r$-th $k$-influential community obtained by greedy strategy is always larger than that computed by random strategy. This is because the size of the community is constrained. If we select the vertex with the largest influence value, a community with a larger influence value could be produced.

\subsection{Case study}
We evaluate a case study for the $k$-influentical community under various aggregation functions on a social network. The dataset could be downloaded from Aminer\footnote{https://www.aminer.org/data}, which is collected for the purpose of cross-domain recommendation. It includes five fields: Data Mining, Medical Informatics, Theory, Visualization, and Database. Each vertex represents a researcher, and the edge between two vertices indicates that they have co-authored at least $1$ publication. Figure~\ref{fig:case study} shows the top-$3$ non-overlapping $k$-influential community under different aggregation functions, when $k=4$. Since our algorithms are heuristic when the aggregation function is $sum$ or $avg$ for top-$r$ size-constrained $k$-influential community search problem. Thus, the result is not explicit. However, the result of the top-$r$ non-overlapping size-constrained $k$-influential community under aggregation functions, e.g., $sum$, $avg$, could satisfy different requirements in practice.

\begin{figure}[!t]
	\centering
	\begin{subfigure}[t]{0.15\textwidth}
		\centering
		\includegraphics[width=2.2cm]{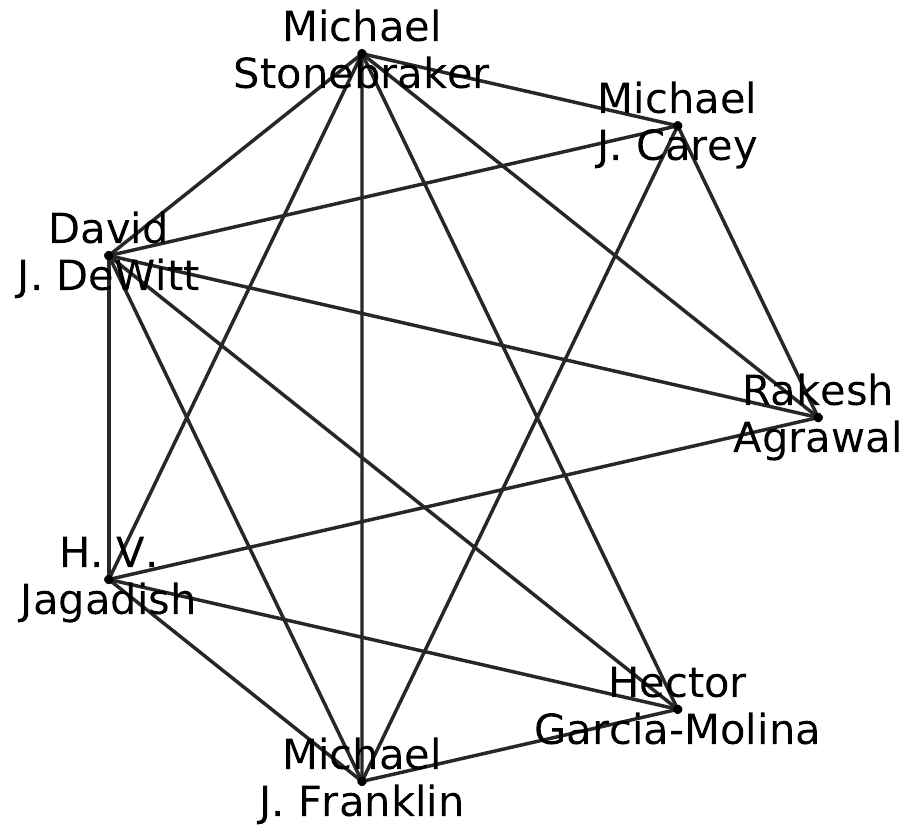}
		\caption{Min: top-$1$}
	\end{subfigure}%
	\begin{subfigure}[t]{0.15\textwidth}
		\centering
		\includegraphics[width=2.2cm]{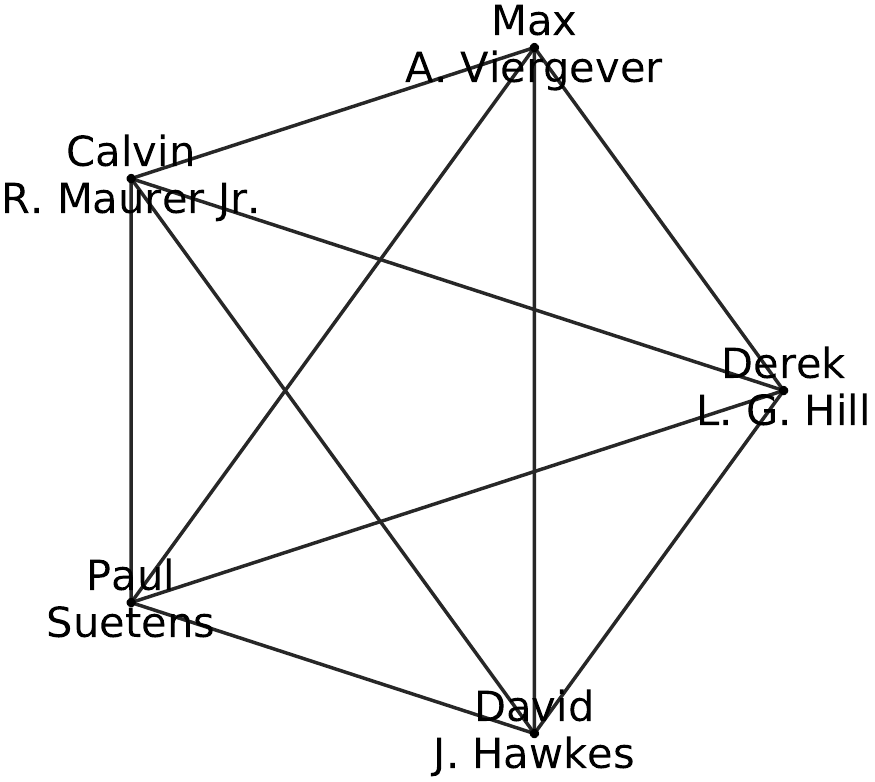}
		\caption{Min: top-$2$}
	\end{subfigure}%
	\begin{subfigure}[t]{0.15\textwidth}
		\centering
		\includegraphics[width=2.2cm]{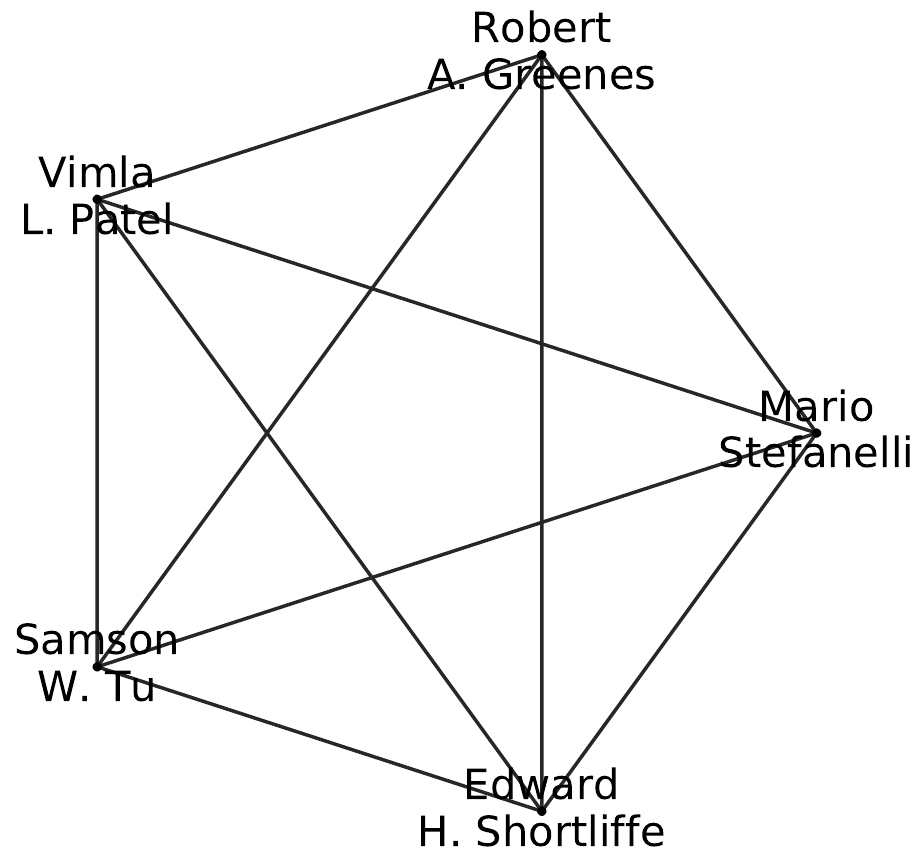}
		\caption{Min: top-$3$}
	\end{subfigure}%
	
	\begin{subfigure}[t]{0.15\textwidth}
		\centering
		\includegraphics[width=2.5cm]{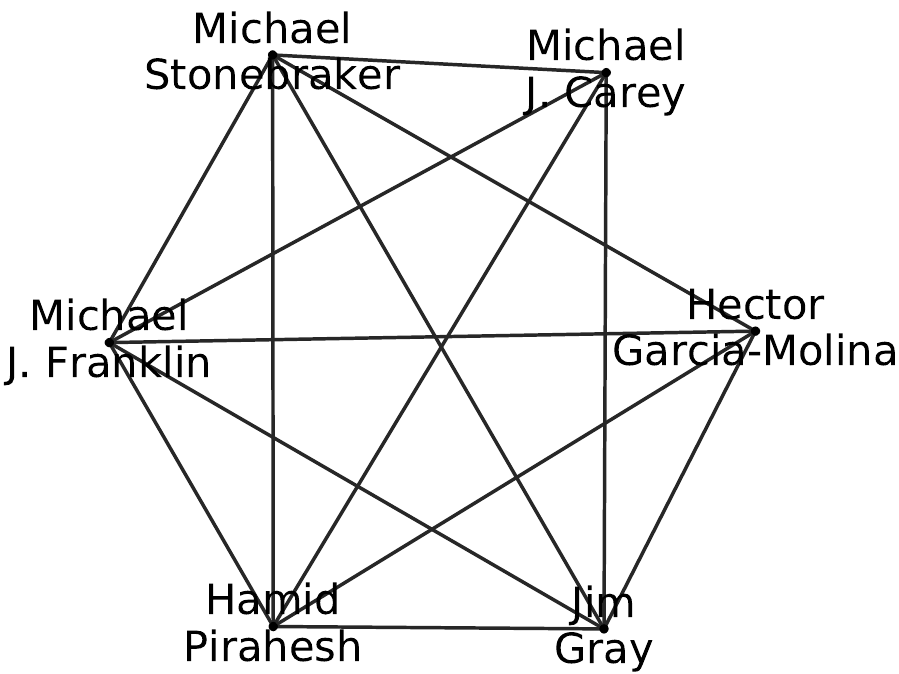}
		\caption{Avg: top-$1$}
	\end{subfigure}%
	\begin{subfigure}[t]{0.15\textwidth}
		\centering
		\includegraphics[width=2.5cm]{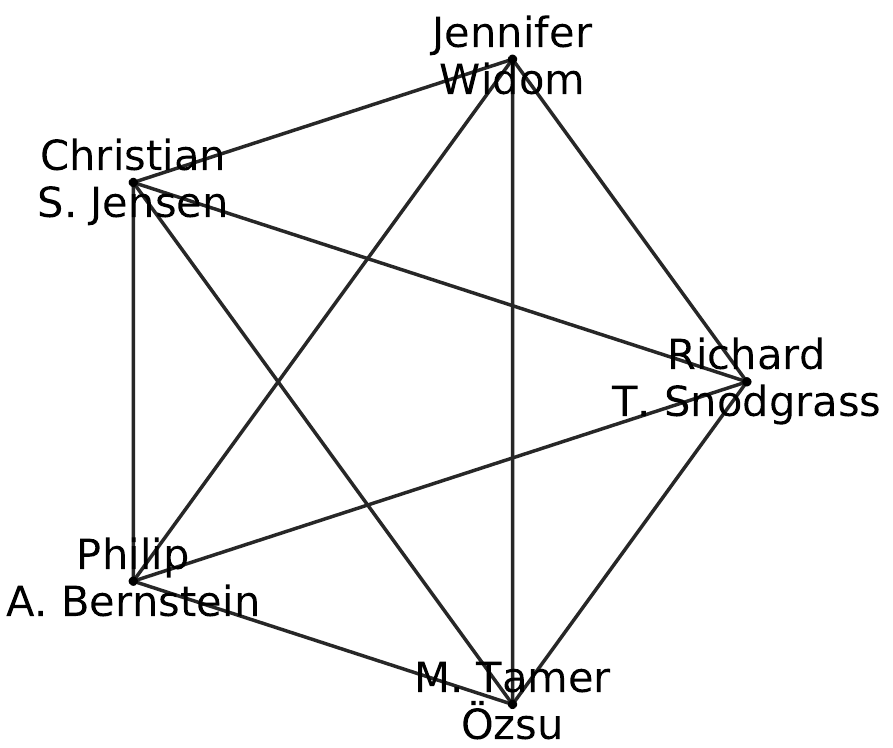}
		\caption{Avg: top-$2$}
	\end{subfigure}%
	\begin{subfigure}[t]{0.15\textwidth}
		\centering
		\includegraphics[width=2.2cm]{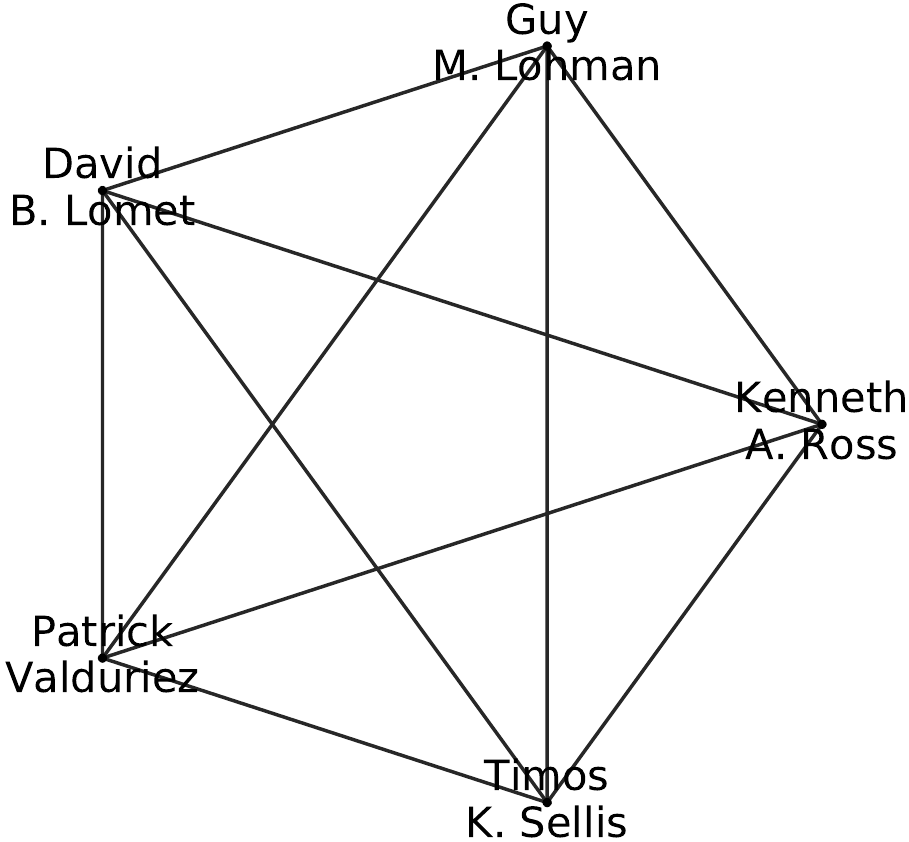}
		\caption{Avg: top-$3$}
	\end{subfigure}%
	
	\begin{subfigure}[t]{0.15\textwidth}
		\centering
		\includegraphics[width=2.5cm]{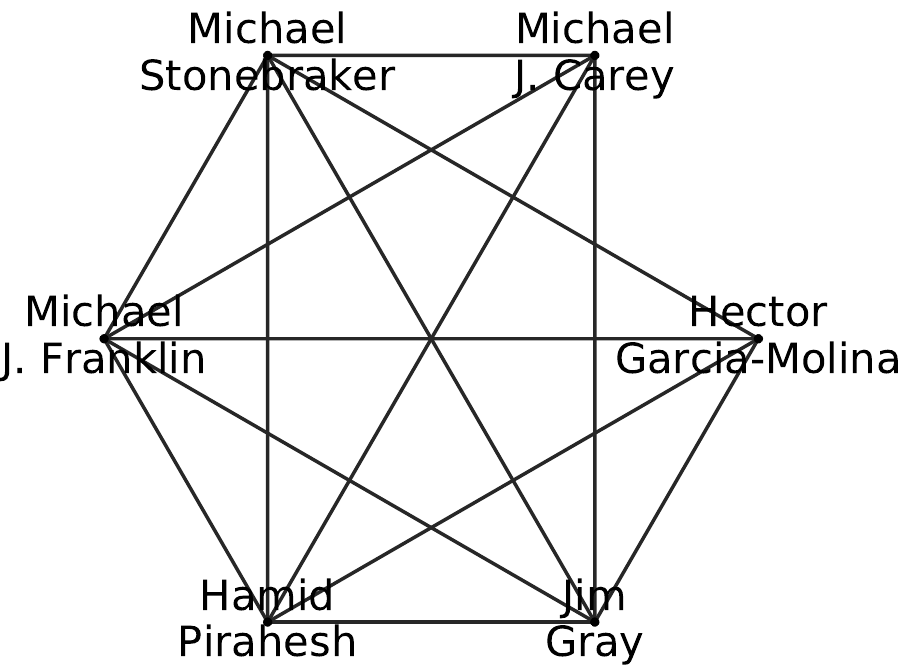}
		\caption{Sum: top-$1$}
	\end{subfigure}%
	\begin{subfigure}[t]{0.15\textwidth}
		\centering
		\includegraphics[width=2.5cm]{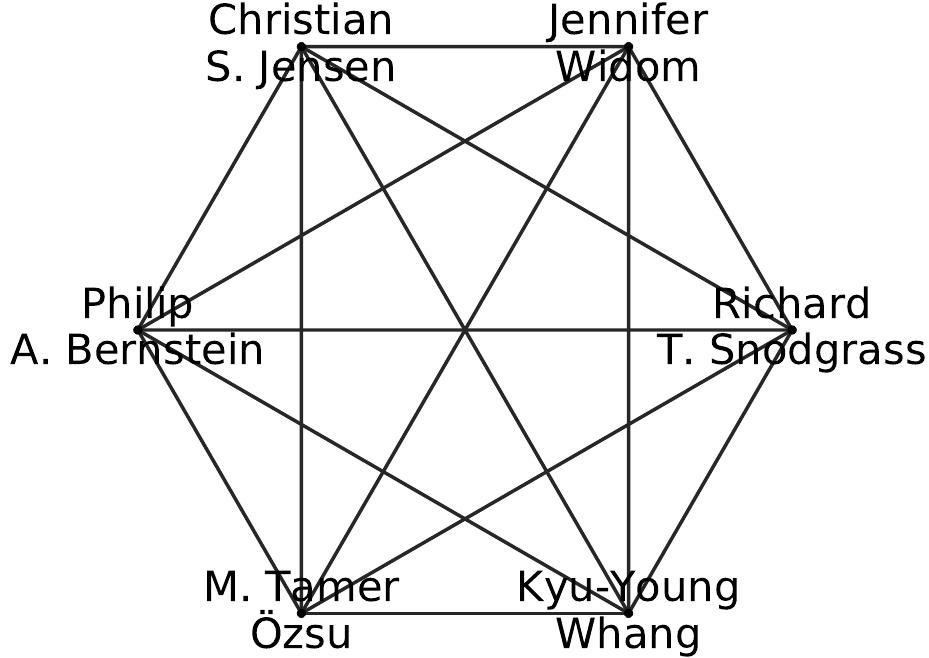}
		\caption{Sum: top-$2$}
	\end{subfigure}%
	\begin{subfigure}[t]{0.15\textwidth}
		\centering
		\includegraphics[width=2.5cm]{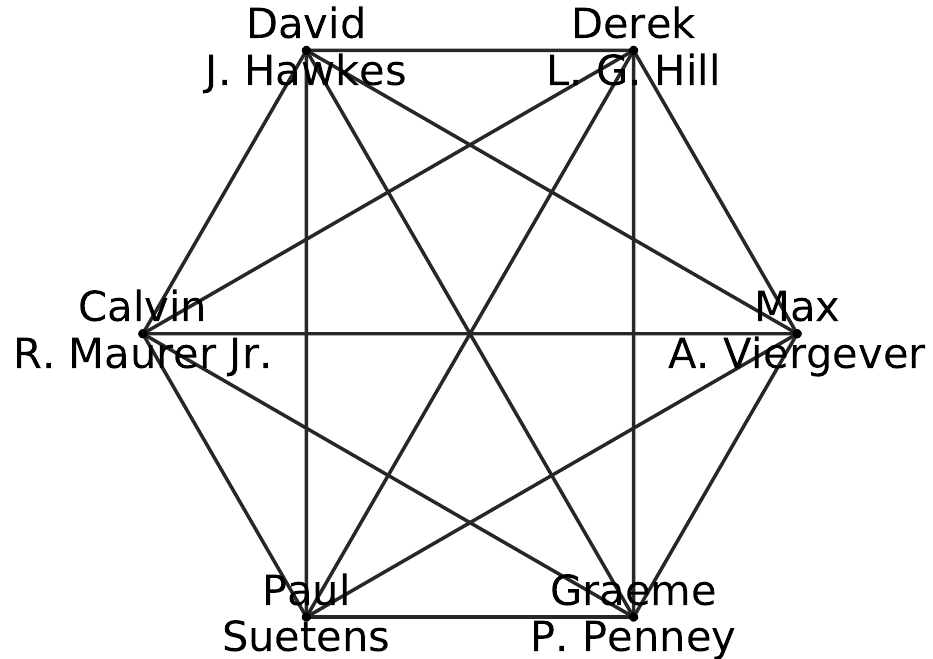}
		\caption{Sum: top-$3$}
	\end{subfigure}%
	\caption{Case Study: Aminer}
	\label{fig:case study}
\end{figure}

\section{related work}\label{sec:related}

\noindent \textbf{Community Detection.} Community detection has been studied for several decades. The goal of community detection is to retrieve all communities that fulfill constraints. It is first investigated in~\cite{gn02, jlc18}. Following that, many improved methods~\cite{kkks08} are proposed. Recently, some researchers study community detection over different kinds of graphs. Some investigate community detection by using machine learning techniques. For instance, Jian et al.~\cite{jwc20} present solutions about community detection over heterogeneous networks. Li et al.~\cite{lzhr20} investigate community detection in the presence of adversarial attacks.

\noindent \textbf{Community Search.} Community search has been widely studied by a large number of researchers~\cite{lai2021pefp,jin2021fast,peng2021efficient}. Sozio et al.~\cite{sg10} present a linear-time algorithm to find a maximal connected $k$-core that contains the set of query vertices. Next, Cui et al.~\cite{cxww14} provide more efficient algorithms for the above-mentioned problem. 
Recently, some scientists have concentrated their efforts on community search over various graphs or with various community models. For example, Fang et al. solve community search over spatial graphs in~\cite{fcll17}. 

Moreover, Fang et al.~\cite{fyzl20} propose effective and efficient algorithms for community search over heterogeneous graphs. Huang et al.~\cite{hl17}, Liu et al.~\cite{lzhx20} and Chen et al.~\cite{clzl18} study community search based on the $k$-truss community model. As for the influential community search problem, Li et al.~\cite{lqym15} firstly study the top-$r$ influential community search problem. Then, an improved online algorithm and a novel progressive method is proposed in~\cite{bclz17}. Both of them, however, ignore an essential point: in some cases, the aggregation function is not $min$, and existed technique cannot be applied directly.

\noindent \textbf{Cohesive Subgraph Mining.} Cohesive subgraphs discovery is a practical and fascinating problem in graph mining. There are some definitions to measure cohesive subgraphs. Among them, the maximal clique~\cite{ldwm19,zzzq19}, the $k$-core~\cite{zzzl20, seid83,peng2018efficient}, the $k$-truss~\cite{cohen08}, and the $k$-edge connected subgraphs~\cite{aiy13, cyql13} are widely-used models. Due to the widespread use of cohesive subgraphs in graph mining, an increasing number of people have focused their attention on this problem in recent years. To illustrate, $k$-core decomposition is studied in~\cite{ckco11, kbst15} and $k$-truss decomposition is investigated in~\cite{wc12, clsw20}.

\section{Future Works}
\label{sect:future}
\rev{
As for the size unconstrained problems that are NP-hard, there is no algorithms proposed. The main obstacle of this problem is the costly search space. For such a problem, a possible direction would be carefully design pruning rules and investigate approximation method. To speedup this process, a parallel or distributed context could also be investigated.}

\section{conclusions}\label{sec:conclusion}

In this paper, we investigate the problem of extracting the top-$r$ $k$-influential communities in social networks under various aggregation functions. As for the top-$r$ $k$-influential community search problem, if the aggregation function is $sum$, we propose an efficient algorithm. Furthermore, we prove the hardness of the problem under size constraint and provide some heuristic algorithms for the top-$r$ size-constraint $k$-influential community search problem. Finally, extensive experiments on $6$ real-world graphs indicate that our algorithms are efficient and effective. The case study reveals that our model has broad applications.

\section*{Acknowledgment}
This work is supported by Hong Kong RGC ECS grant (No. 24203419),
RGC CRF grant (No. C4158-20G), Hong Kong ITC ITF grant  (No. MRP/071/20X), and NSFC grant (No. U1936205). It is also supported by the Research Gants Council of Hong Kong, China under No. 14203618, No. 14202919 and No. 14205520.

\newpage
\bibliographystyle{IEEEtran}
\bibliography{main}



\end{document}